\documentclass[pdflatex,sn-mathphys-num]{sn-jnl}
\input{glyphtounicode}\input{glyphtounicode-cmr}\pdfgentounicode=1
\usepackage{etoolbox}
\patchcmd{\abstractheadfont}{\fontsize{9bp}{11bp}}{\fontsize{10bp}{11bp}}{}{\PackageWarning{readablefrontmatter}{Abstract heading font override failed}}
\patchcmd{\abstractfont}{\fontsize{9bp}{11bp}}{\fontsize{10bp}{12bp}}{}{\PackageWarning{readablefrontmatter}{Abstract font override failed}}
\patchcmd{\keywordfont}{\fontsize{8bp}{9.5bp}}{\fontsize{10bp}{12bp}}{}{\PackageWarning{readablefrontmatter}{Keyword font override failed}}
\AtBeginDocument{\expandafter\patchcmd\csname\string\spiproof\endcsname{\small}{\normalsize}{}{\PackageWarning{readableproofs}{Proof font override failed}}}

\usepackage[utf8]{inputenc}
\usepackage{microtype}
\usepackage{amsmath,amssymb,amsthm,mathtools}
\usepackage{graphicx}
\usepackage{xurl}
\hypersetup{colorlinks=true,citecolor=blue,linkcolor=blue,urlcolor=blue}

\allowdisplaybreaks
\newtheorem{theorem}{Theorem}[section]
\newtheorem{lemma}[theorem]{Lemma}
\newtheorem{proposition}[theorem]{Proposition}
\newtheorem{corollary}[theorem]{Corollary}
\newtheorem*{mainresult}{Main result}
\theoremstyle{definition}
\newtheorem{definition}[theorem]{Definition}
\theoremstyle{remark}
\newtheorem{remark}[theorem]{Remark}

\newcommand{\Hphys}{\mathcal H}
\newcommand{\HS}{\mathrm{HS}}
\newcommand{\wt}{\operatorname{wt}}
\newcommand{\Tr}{\operatorname{Tr}}
\newcommand{\id}{\operatorname{id}}
\newcommand{\rank}{\operatorname{rank}}
\newcommand{\cD}{\mathfrak D}
\newcommand{\cB}{\operatorname{Ball}}
\newcommand{\cE}{\mathcal E}
\newcommand{\cN}{\mathcal N}
\newcommand{\cR}{\mathcal R}
\newcommand{\Fent}{F_{\mathrm{ent}}}
\newcommand{\Mis}{\operatorname{Mis}}

\usepackage{hyperxmp}
\hypersetup{
  pdftitle={Pauli-ball discrepancy and stability of quantum codes},
  pdfauthor={Aryeh Lev Zabokritskiy (Yohananov)},
  pdflang={en-US},
  pdfdisplaydoctitle=true,
  pdfkeywords={quantum error correction, Pauli-ball discrepancy,
    quantum MDS codes, quantum weight enumerators, approximate recovery,
    multipartite entanglement}
}
\hypersetup{bookmarksnumbered=true,bookmarksdepth=2,keeppdfinfo}

\begin{document}
\title{Pauli-ball discrepancy and stability of quantum codes}
\author*[1,2]{\fnm{Aryeh Lev} \sur{Zabokritskiy (Yohananov)}}
\email{yuhanalev@telhai.ac.il}
\affil[1]{\orgdiv{Department of Computer Science},
\orgname{MIGAL--Galilee Research Institute},
\orgaddress{\city{Kiryat Shmona}, \country{Israel}}}
\affil[2]{\orgdiv{Department of Computer Science},
\orgname{Tel-Hai University of Kiryat Shmona and the Galilee},
\orgaddress{\city{Kiryat Shmona}, \country{Israel}}}

\abstract{
The Pauli-ball discrepancy, a quadratic discrepancy summed over all error
radii, measures how uniformly the Pauli error images of a $K$-dimensional
quantum code cover the $d$-dimensional Hilbert space.  At parameters attaining
equality in the quantum Singleton bound with minimum code distance $\geq3$, or in the
one-error Hamming bound with prime-power local dimension, the discrepancy
of every projector of the prescribed rank equals an explicit benchmark
plus an exact nonnegative remainder, which vanishes precisely for pure
quantum maximum-distance-separable codes and nondegenerate perfect codes,
respectively.  For a Pauli ball of radius $e$ below half the target distance
and size $V_e$, this remainder linearly bounds the squared Hilbert--Schmidt
distance of the error-image Gram matrix from the identity.
When $KV_e\leq d$, it linearly bounds the
diamond error and entanglement infidelity of a noise-independent polar
recovery under any noise in that ball's span.  For five qubits, we
obtain the sharp rank-two Gram constant and exact logical channels of rotated
codes attaining this linear order in diamond error.  For pure states, the
discrepancy equals a pure-state benchmark plus a positive weighted sum of
marginal-purity defects, so absolutely maximally entangled states, when they
exist, are exactly the minimizers.
}
\keywords{quantum error correction, Pauli-ball discrepancy, quantum MDS codes,
quantum weight enumerators, approximate recovery, multipartite entanglement}

\maketitle

\section{Introduction}

Discrepancy quantifies departure from uniform distribution.  In Hamming
spaces, quadratic ball discrepancy measures this departure over all centers
and radii; a Stolarsky-type identity expresses it through the distance
distribution of the point set \cite{Barg2021Stolarsky}.  The corresponding
quantum problem concerns a code subspace and its images under Pauli errors,
grouped by weight.  How uniformly do these images cover the physical Hilbert
space, which code projectors minimize the discrepancy, and what does
near-minimality imply for quantum error correction (QEC)?

Fix $n\geq1$, $q\geq2$ and $1\leq K\leq d\coloneqq q^n$, put $Q\coloneqq q^2$,
and let $P$ be a rank-$K$ orthogonal projector on
$\Hphys\coloneqq(\mathbb C^q)^{\otimes n}$.  Let $U_a$ run over the tensor
products of the standard $q$-dimensional shift and phase operators, with
phase ignored, and let $\wt(a)$ be the number of nonidentity tensor factors
of $U_a$.  For $0\leq r\leq n$ and operators $\sigma$ on $\Hphys$ put
\[
 \cB_r\coloneqq\{a:\wt(a)\leq r\},\qquad
 V_r\coloneqq|\cB_r|,\qquad
 T_r(\sigma)\coloneqq\sum_{a\in\cB_r}U_a\sigma U_a^\dagger,
\]
and call $\cB_r$ the Pauli ball of radius $r$.  With $I$ the physical
identity and $\|\cdot\|_{\HS}$ the Hilbert--Schmidt norm, the all-radii
Pauli-ball discrepancy is
\[
 \cD_Q(P)\coloneqq\sum_{r=0}^{n-1}\Phi_r(P),
 \qquad
 \Phi_r(P)\coloneqq\Bigl\|T_r(P)-\frac{KV_r}{d}I\Bigr\|_{\HS}^2,
\]
where $\Phi_r(P)$ is the radius-$r$ covering defect.  Uniform Pauli noise on
$\cB_r$ sends the code state $P/K$ to $V_r^{-1}T_r(P/K)$, so
\[
 \Phi_r(P)=K^2V_r^2\,\bigl\|V_r^{-1}T_r(P/K)-I/d\bigr\|_{\HS}^2 .
\]
We use $\cD_Q$ to compare code projectors; the probabilities of a particular
noise channel enter only the recovery estimates.

Our benchmarks, explicit lower bounds for $\cD_Q$, lift to code projectors
the classical optima of distance-dependent energies of codes in Hamming space
\cite{AshikhminBarg1999Binomial,CohnZhao2014Energy}, in three regimes:
maximum-distance-separable (MDS) parameters, Hamming parameters, and the
parameters of the classical additive codes behind Gottesman's distance-three
quantum family \cite{Gottesman1996Class}.  For classical codes, the
connection of these optima with quadratic discrepancy gives exact minima and
stability estimates \cite{Zabokritskiy2026Perfect,Zabokritskiy2026Stability}.
For a quantum code, $\|T_r(P)\|_{\HS}^2$ and its tightness condition belong
to fusion-frame theory \cite{CasazzaFickus2009Fusion,BachocEhler2013Tight},
whereas the noncommutative coloring problem also called
\emph{quantum discrepancy} \cite{AlishahiEtAl2020Quantum} is distinct from
our coverage measure.

We study stability over all projectors of the prescribed rank: how strongly
near-minimal discrepancy forces orthogonal error images and accurate
recovery.  We call $(n,K)$ formal quantum MDS parameters with target distance
$\Delta\geq3$ if $n\geq2\Delta-2$ and $K=q^{n-2\Delta+2}$, which is equality
in the quantum Singleton bound.  We call $(n,K)$ formal one-error Hamming
parameters if $n\geq3$ and, with $V\coloneqq V_1=1+n(Q-1)$,
\[
 K=q^n/V\in\mathbb N,\qquad V/Q\in\mathbb N,
\]
which is equality in the one-error quantum Hamming bound together with
integrality conditions; the second condition is automatic when $q$ is a
prime power.  Neither definition assumes that a code with these parameters
exists.

At such parameters let $\delta$ be the corresponding benchmark, an explicit
lower bound for $\cD_Q$ over rank-$K$ projectors, and call
$g(P)\coloneqq\cD_Q(P)-\delta\geq0$ the discrepancy excess.  For a radius
$e$, let $G_e$ be the Gram matrix of the vectors $U_ax_i$, $a\in\cB_e$,
$1\leq i\leq K$, for an orthonormal basis $x_1,\ldots,x_K$ of
$\operatorname{ran}P$; up to index order it is the QEC matrix of the error
set $\cB_e$.  Our main stability result is an explicit Gram constant $C_e$,
depending only on the fixed parameters, such that every rank-$K$ projector
$P$ satisfies
\[
 \|G_e-I_{KV_e}\|_{\HS}^2\leq C_e\,g(P)
\]
for $1\leq e<\Delta/2$ at quantum MDS parameters and for $e=1$ at Hamming
parameters.  We call the left side the squared Gram defect.  Here $C_e$ is
the constant $C_e^{\mathrm{Gram}}$ of Corollary~\ref{cor:QMDS-rigidity} at
quantum MDS parameters and $C_1\coloneqq1/c_{n,Q}$, with $c_{n,Q}$ from
Theorem~\ref{thm:main}, at Hamming parameters; the two agree where the
regimes overlap.

Sections~\ref{sec:pauli} and~\ref{sec:enumerators} develop the covering
defect and express $\cD_Q$ through the trace and overlap enumerators, with
coefficients $\sum_{\wt(a)=w}|\Tr(PU_a)|^2$ and
$\sum_{\wt(a)=w}\|PU_aP\|_{\HS}^2$, the $A$ and $B$ enumerators of
Shor--Laflamme and Rains \cite{ShorLaflamme1997MacWilliams,Rains1998Enumerators}.
Intrinsic enumerators and tensor networks describe and compute enumerators
\cite{KubischtaTeixeira2026Intrinsic,CaoEtAl2024Lego}, numerical-range
methods construct exact codes \cite{LipkaBartosikZyczkowski2017Nuclear}, and
signature norms and joint higher-rank numerical ranges describe exact
detection \cite{DuEtAl2026KLGeometry,GuptaEtAl2026Variance}.  Other
enumerator methods give linear-programming bounds under general Kraus models
\cite{OuyangLai2022ChannelLP} and finite-length fidelity bounds under
depolarizing noise \cite{Ashikhmin2014Fidelity}; our estimates instead
control the squared Gram defect and the recovery error of an arbitrary
projector through its discrepancy excess.

Section~\ref{sec:qmds} treats quantum MDS parameters.  Huber and Grassl
determine quantum MDS enumerators from their parameters and establish purity
for $\Delta\geq3$ \cite{HuberGrassl2020QMDS}.  For every rank-$K$
projector, including rank-one and nonstabilizer projectors, $g(P)$ equals an
explicit positive combination of these enumerator coefficients at the
weights where they vanish for a pure quantum MDS code, and
\[
 g(P)=0\iff PU_aP=0\quad\text{for }1\leq\wt(a)<\Delta ,
\]
that is, exactly when $P$ is a pure quantum MDS code; the benchmark is
attained if and only if such a code exists.  The excess also bounds linearly the squared
Hilbert--Schmidt distances $\|\rho_S-I_S/q^{|S|}\|_{\HS}^2$ of the marginals
$\rho_S$ of the code state $P/K$ on sets $S$ of at most $n-\Delta+1$ parties
from the maximally mixed state, where $I_S$ is the identity on $S$.

Sections~\ref{sec:stability} and~\ref{sec:equality} treat one-error Hamming
parameters.  There $KV=d$, so the squared radius-one Gram defect is
$\Phi_1(P)$, and $\Phi_1(P)=0$ means that the radius-one error images form an
orthogonal decomposition of $\Hphys$.  For every projector of this rank,
$g(P)$ equals an explicit nonnegative remainder that vanishes exactly when
$\Phi_1(P)=0$ and, by the Gram bound above, is at least $\Phi_1(P)/C_1$.
The benchmark is
therefore attained precisely by nondegenerate perfect codes; a degenerate
one-error-correcting code of this rank, if one exists, has
$\Phi_1(P)>0$ and does not attain it.
For stabilizer projectors, $\Phi_1(P)\leq C_1\,g(P)$ is the lower bound in
the classical two-sided inequality of \cite{Zabokritskiy2026Stability} for
the normalizer code; the new content is its validity for every projector of
the prescribed rank.

Sections~\ref{sec:qmds-recovery} and~\ref{sec:rigidity} treat recovery when
$KV_e\leq d$, in both regimes.  Let $\operatorname{Enc}$ embed
$\mathbb C^K$ onto $\operatorname{ran}P$.  We use one recovery $\cR_e$,
independent of the noise: the adjoint of the isometric polar factor of the
error-image map $(a,x)\mapsto U_ax$, followed by discarding the error label.
Polar and transpose/Petz recoveries are established
\cite{BarnumKnill2002Reversing,BenyOreshkov2010Approximate,
EldarForney2001SquareRoot,NgMandayam2010Transpose,
ZhengEtAl2024NearOptimal,MaTanWright2025Haar,ElimelechAlbertBarg2026ErrorSet};
on noisy code states $\cR_e$ coincides with the error-set Petz map of
Elimelech, Albert, and Barg, who give universal recovery guarantees when the
noise has Kraus operators in $\operatorname{span}\{U_a:a\in\cB_e\}$
\cite{ElimelechAlbertBarg2026UniversalRecovery}.

Let $\Fent$ be the entanglement fidelity at the maximally mixed logical
input, let $\|\cdot\|_\diamond$ be the unhalved diamond norm, and let
$\cN_{e,p}(\sigma)\coloneqq\sum_{a\in\cB_e}p_aU_a\sigma U_a^\dagger$ be a
stochastic Pauli mixture with largest probability $p_\ast$, including that
of the identity.  For every channel $\cN$ with Kraus operators in
$\operatorname{span}\{U_a:a\in\cB_e\}$, the average-case guarantee of
Elimelech, Albert, and Barg gives the first bound, and
\[
 \begin{aligned}
 1-\Fent(I_K/K,\cR_e\circ\cN\circ\operatorname{Enc})
 &\leq C_e\,g(P)/K,\\
 \|\cR_e\circ\cN\circ\operatorname{Enc}-\id_K\|_\diamond
 &\leq 2C_e\,g(P),\\
 \|\cR_e\circ\cN_{e,p}\circ\operatorname{Enc}-\id_K\|_\diamond
 &\leq 2p_\ast C_e\,g(P),\\
 1-\Fent(I_K/K,\cR_e\circ\cN_{e,p}\circ\operatorname{Enc})
 &\leq p_\ast C_e\,g(P)/K.
 \end{aligned}
\]
Since $p_\ast$ is at least the identity probability, the last two bounds do
not vanish with the noise.  For general channels in the span, the
worst-case guarantees of
\cite{MaTanWright2025Haar,ElimelechAlbertBarg2026UniversalRecovery} give
diamond-error bounds of order $g(P)^{1/2}$.

Section~\ref{sec:five} treats five qubits, where Pauli-rotation families
have been optimized for noise-tailored recovery \cite{ShaoEtAl2026VGQEC} and
the logical process matrices of the five-qubit code are known for
independent unital and coherent noise \cite{Liu2023FiveQubitCoherent}.
Here the quantum MDS and Hamming benchmarks coincide at $K=2$, and every
rank-two projector satisfies
\[
 \|G_1-I_{32}\|_{\HS}^2\leq g(P)/3 .
\]
The family obtained from the $[[5,1,3]]_2$ code projector
\cite{LaflammeEtAl1996Perfect} by the two-qubit rotation
$e^{-i\theta X_4X_5}$, of the type used in \cite{ShaoEtAl2026VGQEC} up to
local-Clifford and qubit-ordering conventions, attains this bound.  For every
stochastic Pauli mixture supported in $\cB_1$ we determine the logical Pauli
channel of each member under its own polar recovery, with its diamond error
and entanglement infidelity in terms of $g(P)$.  Under uniform noise on
$\cB_1$ the diamond error is $g(P)/192+O(g(P)^2)$ as $g(P)\to0$, so the
linear order of the general recovery bound is attained.

Section~\ref{sec:balanced-summary} treats balanced packing beyond perfect
codes, for Gottesman's $[[2^j,2^j-j-2,3]]_2$ codes, $j\geq3$, whose
radius-one error images are pairwise orthogonal but do not fill the space.
These codes minimize $\cD_Q$ at their rank, and every projector of that rank
exceeds their discrepancy by at least a positive multiple of its squared
radius-one Gram defect.

Section~\ref{sec:pure} treats rank-one projectors, building on the
established relation between quantum enumerators, marginal purities, and
multipartite entanglement \cite{Scott2004Multipartite,HuberEtAl2018AME,
FacchiEtAl2008MMES,HelwigCui2013AME,YuEtAl2021Marginal}.  For every $n\geq2$ there is an explicit pure-state benchmark
$\delta^{\mathrm{pure}}_{n,q}$ such that every rank-one projector $P_\psi$,
with marginals $\rho_S$ on $S\subseteq\{1,\ldots,n\}$, satisfies
\[
 \cD_Q(P_\psi)-\delta^{\mathrm{pure}}_{n,q}
 =q^{-n}\sum_{s=1}^{\lfloor n/2\rfloor}\gamma_s
  \sum_{|S|=s}\Bigl\|\rho_S-\frac{I_S}{q^s}\Bigr\|_{\HS}^2,
 \qquad\gamma_s>0 .
\]
Hence absolutely maximally entangled $\operatorname{AME}(n,q)$ states, whose
marginals on at most $\lfloor n/2\rfloor$ parties are maximally mixed, are
exactly the minimizers whenever they exist.

Section~\ref{sec:limitations} lists open questions.

\section{Pauli balls and the covering defect}
\label{sec:pauli}

We first express uniform coverage by the Pauli error images
$U_a\operatorname{ran}P$ of the code as an operator
condition.  The resulting covering defect has two forms: a centered sum of
projectors and a sum of pairwise overlaps.  The first gives its geometric
meaning; the second connects it to quantum weight enumerators.  A
constant shift of it at fixed rank, the packing defect, equals the squared
Gram defect used in the Introduction.

Fix $q\geq2$, put $d\coloneqq q^n$, let $[n]\coloneqq\{1,\ldots,n\}$, and let
$\Hphys\coloneqq(\mathbb C^q)^{\otimes n}$.
Let $X,Z$ be the generalized Weyl operators on $\mathbb C^q$,
\[
 X|j\rangle\coloneqq|j+1\bmod q\rangle,
 \qquad Z|j\rangle\coloneqq e^{2\pi i j/q}|j\rangle.
\]
Ignoring irrelevant phase factors, tensor products of the $X^uZ^v$ are
indexed by
\[
 a=((u_1,v_1),\ldots,(u_n,v_n))\in(\mathbb Z_q^2)^n.
\]
We write
$U_a$ for the corresponding unitary, define
$\operatorname{supp}(a)\coloneqq\{i\in[n]:(u_i,v_i)\ne(0,0)\}$, and set
$\wt(a)\coloneqq|\operatorname{supp}(a)|$.  The basis is
Hilbert--Schmidt orthogonal:
\begin{equation}
 \Tr(U_a^\dagger U_b)=d\,\mathbf 1_{a=b}.
 \label{eq:weyl-orthogonality}
\end{equation}
Put $Q\coloneqq q^2$.  For an integer $0\leq r\leq n$, put
\[
 \begin{aligned}
 \cB_r&\coloneqq\{a:\wt(a)\leq r\},\\
 V_r&\coloneqq|\cB_r|=\sum_{j=0}^r\binom nj(Q-1)^j.
 \end{aligned}
\]

\begin{definition}[Pauli-ball discrepancy]
Let $P$ be an orthogonal projector of rank $0\leq K\leq d$.  Define
\begin{align}
 T_r(P)&\coloneqq\sum_{a\in\cB_r}U_aPU_a^\dagger,\notag\\
 \alpha_r(P)&\coloneqq\frac{\Tr(P)V_r}{d}=\frac{KV_r}{d},
 \label{eq:Tr-def}\\
 \Phi_r(P)&\coloneqq\|T_r(P)-\alpha_r(P) I\|_{\HS}^2,\notag\\
 \cD_Q(P)&\coloneqq\sum_{r=0}^{n-1}\Phi_r(P).
 \label{eq:Phi-D-def}
\end{align}
\end{definition}

We call $\Phi_r(P)$ the radius-$r$ covering defect of $P$.

\begin{remark}[Rank-complement symmetry]
\label{rem:complement}
Let $P^c\coloneqq I-P$, of rank $d-K$.  Since
$T_r(P^c)=V_rI-T_r(P)$ and
$\alpha_r(P^c)=V_r-\alpha_r(P)$,
\begin{equation}
 T_r(P^c)-\alpha_r(P^c)I
 =-\bigl(T_r(P)-\alpha_r(P)I\bigr).
 \label{eq:complement-symmetry}
\end{equation}
Consequently $\Phi_r(P^c)=\Phi_r(P)$ for every $r$ and
$\cD_Q(P^c)=\cD_Q(P)$.  Thus minimizers at rank $K$ correspond exactly to
minimizers at rank $d-K$.  This is the operator analogue of complement
invariance for classical discrepancy.
\end{remark}

For the normalized enumerators and code states below, assume $K\geq1$.

The omitted $r=n$ summand is zero, since Weyl twirling gives
$T_n(P)=dK I$.  The $r=0$ summand is the constant $K-K^2/d$; retaining it
does not affect optimization and makes the radii $r=0,\dots,n-1$ run over
all coefficients of the polynomial in \eqref{eq:W-norm} below.

\begin{proposition}[Fusion-frame and overlap forms]
\label{prop:fusion}
For every $r$,
\begin{equation}
 \Phi_r(P)=\Tr T_r(P)^2-\frac{(KV_r)^2}{d}.
 \label{eq:fusion-excess}
\end{equation}
If $c_r(w)$ is the cardinality of the intersection of two radius-$r$ balls in
the $Q$-ary Hamming space of labels $(\mathbb Z_q^2)^n$, with distance
$\wt(b-a)$, whose centers are at distance $w$, then
\begin{equation}
 \Tr T_r(P)^2=K\sum_{w=0}^n c_r(w)B_w(P),
 \label{eq:overlap-form}
\end{equation}
where
\begin{equation}
 B_w(P)\coloneqq\frac1K\sum_{\wt(a)=w}\Tr(PU_aPU_a^\dagger).
 \label{eq:B-def}
\end{equation}
Thus $\Phi_r$ is the fusion-frame potential of the subspaces
$U_a\operatorname{ran}P$, $a\in\cB_r$, minus its tight-frame value
$(KV_r)^2/d$.
\end{proposition}

\begin{proof}
Since $\Tr T_r=KV_r$, expansion of the square proves
\eqref{eq:fusion-excess}.  In $\Tr T_r^2$, group ordered pairs of errors by
the weight of their difference; $U_a^\dagger U_b$ is a phase times
$U_{b-a}$.  For each fixed difference of weight $w$ the
number of pairs is $c_r(w)$, giving \eqref{eq:overlap-form}.
\end{proof}

The covering defect $\Phi_r$ measures covering uniformity.  It cannot vanish when
$KV_r<d$: since $T_r(P)-P$ is positive semidefinite and $K\geq1$, the
operator $T_r(P)$ has an eigenvalue at least $1>\alpha_r(P)$.  In that
regime, which includes the Gottesman codes of
Section~\ref{sec:balanced-summary}, the quantity that can vanish is the
packing defect
\begin{equation}
 \Psi_r(P)\coloneqq\Tr T_r(P)^2-KV_r
          =\Phi_r(P)-KV_r\bigl(1-\alpha_r(P)\bigr).
 \label{eq:packing-defect}
\end{equation}
The two quantities differ by a constant at fixed rank and coincide precisely
when $KV_r=d$.  If $G_r$ is the Gram matrix of the images of an orthonormal
code basis under $U_a$, $a\in\cB_r$, then $\Tr T_r(P)^2=\|G_r\|_{\HS}^2$ and
$\Tr G_r=KV_r$, so
$\Psi_r(P)=\|G_r-I_{KV_r}\|_{\HS}^2\geq0$ (the squared Gram defect), with equality exactly when the
images $U_a\operatorname{ran}P$, $a\in\cB_r$, are pairwise orthogonal.

There is also a smoothing interpretation.  For the code state $\rho_P\coloneqq P/K$
and the uniform-ball Pauli channel, the output is
\[
 \rho_r\coloneqq\frac1{V_r}\sum_{a\in\cB_r}U_a\rho_PU_a^\dagger
       =\frac{T_r(P)}{KV_r}.
\]
With $\tau\coloneqq I/d$ and the Hilbert--Schmidt quantum chi-square divergence,
\begin{equation}
 \chi^2_{\HS}(\rho_r\|\tau)
 \coloneqq d\|\rho_r-\tau\|_{\HS}^2
 =\frac{d\,\Phi_r(P)}{K^2V_r^2}.
 \label{eq:smoothing}
\end{equation}
\section{Quantum enumerators and the spectral identity}
\label{sec:enumerators}

To minimize the covering defect, we need a scalar description that
retains the constraints on quantum projectors.  Weyl expansion provides this
description: ball averaging acts separately on each weight, and discrepancy
becomes a positive weighted sum of squared Weyl coefficients.  The effective
alphabet size is $Q=q^2$.  The two quantum enumerators below record trace and
overlap data; their MacWilliams relation is the constraint that will make
scalar energy bounds applicable to every rank-$K$ projector.

Together with the overlap enumerator $B_w(P)$ in \eqref{eq:B-def}, we use the
trace enumerator in the Shor--Laflamme normalization
\begin{equation}
 A_w(P)\coloneqq\frac1{K^2}\sum_{\wt(a)=w}|\Tr(PU_a)|^2.
 \label{eq:A-def}
\end{equation}
For any projector $P$, write
\[
 A_P(z)\coloneqq\sum_{w=0}^nA_w(P)z^w,
 \qquad
 B_P(z)\coloneqq\sum_{w=0}^nB_w(P)z^w.
\]
Let
\begin{equation}
 K_j^{(n,Q)}(w)\coloneqq\sum_{\ell=0}^j(-1)^\ell(Q-1)^{j-\ell}
 \binom w\ell\binom{n-w}{j-\ell}
 \label{eq:krawtchouk}
\end{equation}
be the $Q$-ary Krawtchouk polynomial.

These enumerators satisfy the following feasibility and transform relations.

\begin{proposition}[Feasibility and MacWilliams normalization]
\label{prop:macwilliams}
For every rank-$K$ projector,
\begin{gather}
 A_0=B_0=1,\qquad A_w\geq0,\qquad B_w\geq A_w,\label{eq:AB-positive}\\
 \sum_{w=0}^nA_w=\frac dK,\qquad
 \sum_{w=0}^nB_w=dK,\label{eq:AB-masses}\\
 B_j=\frac Kd\sum_{w=0}^nA_wK_j^{(n,Q)}(w),
 \notag\\
 A_j=\frac1{dK}\sum_{w=0}^nB_wK_j^{(n,Q)}(w).
 \label{eq:quantum-macwilliams}
\end{gather}
Thus the vectors $A\coloneqq(A_w)_w$ and $B\coloneqq(B_w)_w$ are both
feasible, that is, nonnegative with nonnegative MacWilliams transforms:
$A$ has mass $d/K$ with transform $B$, while $B$ has mass $dK$ with
transform $A$.
\end{proposition}

\begin{proof}
Only the inequality in \eqref{eq:AB-positive} needs care.  For each $a$, put
$X\coloneqq PU_aP$ as an operator on the $K$-dimensional range of $P$.  Cauchy--Schwarz
gives
\[
 |\Tr X|^2\leq K\Tr(XX^\dagger)
 =K\Tr(PU_aPU_a^\dagger).
\]
Summing at fixed weight gives $A_w\leq B_w$.

Expand $P$ in the orthogonal Weyl basis \eqref{eq:weyl-orthogonality} and use
the character acquired by a Weyl coefficient under conjugation.  The weight
transform of this character is exactly \eqref{eq:krawtchouk}, which gives the
first identity in \eqref{eq:quantum-macwilliams}.  Since the Krawtchouk
matrix satisfies $\sum_jK_i^{(n,Q)}(j)K_j^{(n,Q)}(w)=Q^n\mathbf 1_{i=w}$,
inverting the first identity gives the second.  Its case $j=0$ and its sum
over $j$ give \eqref{eq:AB-masses}.  These are the standard quantum MacWilliams
identities \cite{ShorLaflamme1997MacWilliams,Rains1998Enumerators} in the
normalization \eqref{eq:A-def} and \eqref{eq:B-def}.
\end{proof}

\begin{definition}[Pure codes and distance]
\label{def:pure-enumerator-distance}
A rank-$K$ projector $P$ defines a \emph{pure code of distance at least}
$\Delta$ if $PU_aP=0$ for every nonidentity Weyl operator with
$\wt(a)<\Delta$.
Equivalently, $B_s(P)=0$ for $1\leq s<\Delta$, since each summand in
$B_s(P)$ is $\|PU_aP\|_{\HS}^2/K$.  We use this same pure-code convention
when $K=1$.  The scalar Knill--Laflamme condition by itself would be
vacuous in rank one.
\end{definition}

Define the radius-$r$ Lloyd polynomial and the all-radii spectral weight by
\begin{equation}
 L_r(w)\coloneqq\sum_{j=0}^rK_j^{(n,Q)}(w),
 \qquad
 W_{n,Q}(w)\coloneqq\sum_{r=0}^{n-1}L_r(w)^2.
 \label{eq:L-W}
\end{equation}

\begin{theorem}[Pauli spectral identity]
\label{thm:spectral}
For every rank-$K$ projector,
\begin{align}
 \Phi_r(P)&=\frac{K^2}{d}\sum_{w=1}^nA_w(P)L_r(w)^2,
 \label{eq:Phi-spectral}\\
 \cD_Q(P)&=\frac{K^2}{d}\sum_{w=1}^nA_w(P)W_{n,Q}(w).
 \label{eq:D-spectral}
\end{align}
\end{theorem}

\begin{proof}
The Weyl expansion of $P$ and conjugation characters give
\[
 T_r(P)=\frac1d\sum_b\Tr(U_b^\dagger P)L_r(\wt(b))U_b.
\]
The $b=0$ term is $\alpha_r(P)I$.  Parseval with
\eqref{eq:weyl-orthogonality} proves \eqref{eq:Phi-spectral}; summing in $r$
 proves \eqref{eq:D-spectral}.
\end{proof}

Throughout the linear-programming arguments, $A$ is the energy variable:
by \eqref{eq:D-spectral}, $\cD_Q(P)=(K^2/d)E_W(A)$ with the energy
$E_W(x)\coloneqq\sum_{w\geq1}x_wW_{n,Q}(w)$ of a vector $x=(x_w)_w$, and $B$
is the MacWilliams transform of $A$.

The Fourier identity between squared Lloyd polynomials and ball
intersections, used in recent proofs of the quantum Hamming bound
\cite{ZhangChen2026Degeneracy,ZhangChen2026Dimension}, supplies the same
coefficient $L_r(w)^2$.  Theorem~\ref{thm:spectral} states it for arbitrary
projectors in the Shor--Laflamme normalization and adds the sum over all
radii.

For later coefficient estimates, set
$u(z)\coloneqq1+(Q-1)z$, $b(z)\coloneqq1-z$, and
$\|\sum_jp_jz^j\|_2^2\coloneqq\sum_j|p_j|^2$.  Since
$\sum_jK_j^{(n,Q)}(w)z^j=u(z)^{n-w}b(z)^w$, division by $1-z$ gives
$\sum_rL_r(w)z^r=u(z)^{n-w}b(z)^{w-1}$, a polynomial of degree $n-1$ when
$w\geq1$; hence the coefficient-norm representation
\begin{equation}
 W_{n,Q}(w)=\|u(z)^{n-w}b(z)^{w-1}\|_2^2,
 \qquad 1\leq w\leq n.
 \label{eq:W-norm}
\end{equation}

\section{Quantum MDS stability}
\label{sec:qmds}

The goal of this section is to minimize $\cD_Q$ at the Singleton rank
$K=q^{n-2\Delta+2}$ and to control its excess by the $A$- and
$B$-coordinates that vanish for a pure quantum MDS code.  The spectral identity reduces this operator
question to a Delsarte linear-programming certificate, that is, an auxiliary
function in the sense of \cite[Proposition~5 and Section~IV]{CohnZhao2014Energy}.  Recall that a
feasible Delsarte quasicode \cite[Definition~6]{CohnZhao2014Energy} of length
$n$, alphabet size $Q$, and mass $M$ is a vector
$a=(a_0,\ldots,a_n)$ satisfying, for $0\leq j\leq n$,
\[
 \begin{aligned}
 a_0&=1,\qquad a_w\geq0,\qquad \sum_{w=0}^na_w=M,\\
 a_j^\perp&\coloneqq\frac1M\sum_{w=0}^na_wK_j^{(n,Q)}(w)\geq0.
 \end{aligned}
\]
Write $A(P)\coloneqq(A_0(P),\ldots,A_n(P))$.  Proposition~\ref{prop:macwilliams}
says that $A(P)$ is such a quasicode with
$M=d/K$, and its dual is $B(P)$.

Since $Q=q^2\geq4$, the spectral weight $W_{n,Q}$ is a strictly completely
monotone potential on $\{1,\ldots,n\}$: with
$\nabla f(w)\coloneqq f(w+1)-f(w)$, one has
$(-1)^j\nabla^jW_{n,Q}(w)>0$ whenever $j\geq0$ and
$1\leq w\leq n-j$.  For real $x\geq0$, define the interpolation
\begin{equation}
 \widetilde W_{n,Q}(x+1)\coloneqq\frac1{2\pi}\int_0^{2\pi}
 \alpha_Q(\varphi)^{n-1}
 \left(\frac{\beta(\varphi)}{\alpha_Q(\varphi)}\right)^x\,d\varphi,
 \label{eq:W-complete-monotone}
\end{equation}
where
\[
 \alpha_Q(\varphi)\coloneqq|1+(Q-1)e^{i\varphi}|^2,
 \qquad \beta(\varphi)\coloneqq|1-e^{i\varphi}|^2.
\]
Parseval applied to \eqref{eq:W-norm} gives
$\widetilde W_{n,Q}(w)=W_{n,Q}(w)$ for every integer $1\leq w\leq n$.
Since $0\leq\beta/\alpha_Q\leq1$ when $Q\geq4$, with strict inequalities
except at finitely many $\varphi$, the alternating finite
differences of these integer values are strictly positive.  This is the
spectral form of the complete-monotonicity calculation in
\cite[Proposition~5.1]{Zabokritskiy2026Perfect}.

By \eqref{eq:D-spectral} and Proposition~\ref{prop:macwilliams}, a
Delsarte linear-programming certificate for the potential $W_{n,Q}$
transfers to the quantum $A$-enumerator in the following exact form.  It
gives the quantum MDS remainder below and the Gottesman remainder in
Appendix~\ref{sec:gottesman}.

\begin{lemma}[Certificate transfer]
\label{lem:certificate-transfer}
Let $0\leq m\leq n$ and $\widehat h(0),\ldots,\widehat h(m)\in\mathbb R$, and
put $h(w)\coloneqq\sum_{s=0}^m\widehat h(s)K_s^{(n,Q)}(w)$ and
$r(w)\coloneqq W_{n,Q}(w)-h(w)$ for $1\leq w\leq n$.  Let $(a_w)_{w=0}^n$ be
real, with $a_0=1$, $\sum_{w=0}^na_w=d/K$,
$\sum_{w=0}^na_wK_s^{(n,Q)}(w)=0$ for $1\leq s\leq m$, and $r(w)=0$ whenever
$w\geq1$ and $a_w\neq0$, and put
$\delta\coloneqq\frac{K^2}{d}\sum_{w=1}^na_wW_{n,Q}(w)$.  Then every
rank-$K$ projector $P$ satisfies
\begin{equation}
 \cD_Q(P)-\delta
 =\frac{K^2}{d}\sum_{w=1}^nr(w)A_w(P)
 +K\sum_{s=1}^m\widehat h(s)B_s(P).
 \label{eq:certificate-transfer}
\end{equation}
\end{lemma}

\begin{proof}
By the first identity in \eqref{eq:quantum-macwilliams},
$\sum_{w=0}^nA_w(P)K_s^{(n,Q)}(w)=(d/K)B_s(P)$.  Multiplying by
$K^2\widehat h(s)/d$, summing over $0\leq s\leq m$, and separating
$A_0=1$ gives
\[
 \frac{K^2}{d}\sum_{w=1}^nh(w)A_w(P)
 =K\sum_{s=0}^m\widehat h(s)B_s(P)-\frac{K^2}{d}h(0).
\]
The same computation for $(a_w)$, whose mass is $d/K$ and whose transform
coordinates $1,\ldots,m$ vanish, gives
$\frac{K^2}{d}\sum_{w=1}^nh(w)a_w=K\widehat h(0)-\frac{K^2}{d}h(0)$; since
$h=W_{n,Q}$ wherever $w\geq1$ and $a_w\neq0$, the left side is $\delta$.
Inserting $W_{n,Q}=h+r$ in \eqref{eq:D-spectral}, subtracting, and using
$B_0(P)=1$ proves \eqref{eq:certificate-transfer}.
\end{proof}

Complete monotonicity enters
through the positivity of the interpolation coefficients $\mu_j$ in
\eqref{eq:QMDS-mu} below, since $\mu_j=(-1)^j\nabla^jW_{n,Q}(n-j)$.

\subsection{The exact remainder at Singleton rank}

The first application of Lemma~\ref{lem:certificate-transfer} holds at the
formal Singleton parameters below, for all projectors, including nonadditive and nonstabilizer ones.  Let
$\Delta\geq3$, the range in which every quantum MDS code is pure
\cite{HuberGrassl2020QMDS} and radii $1\leq e$ with $2e<\Delta$ exist; the
Huber--Grassl results are used only to identify the attaining projectors, not
for the identity \eqref{eq:QMDS-exact-gap} or the bound
\eqref{eq:QMDS-stability}.  Put
\begin{equation}
 t\coloneqq\Delta-1,
 \qquad d_*\coloneqq n-t+1=n-\Delta+2,
 \qquad Q\coloneqq q^2,
 \label{eq:QMDS-t-dstar}
\end{equation}
and consider the formal quantum MDS parameters, fixed by the quantum
Singleton equality whether or not a code with them exists,
\begin{equation}
 ((n,K,\Delta))_q,
 \qquad \ell\coloneqq n-2t\geq0,
 \qquad K=q^\ell.
 \label{eq:QMDS-parameters}
\end{equation}
These formal parameters include $\ell=0$.

To define the benchmark at every prescribed parameter set, we use the formal
MDS enumerator determined by those parameters:
\begin{equation}
 \begin{aligned}
 A_0^*&\coloneqq1,\\
 A_w^*&\coloneqq0 \qquad(1\leq w<d_*),\\
 A_w^*&\coloneqq\binom nw(Q-1)
 \sum_{u=0}^{w-d_*}(-1)^u\binom{w-1}{u}Q^{w-d_*-u}
 \qquad(d_*\leq w\leq n).
 \end{aligned}
 \label{eq:QMDS-Astar}
\end{equation}
When no quantum MDS code exists, $A^*$ may have negative entries (at
$((6,4,3))_2$, $A_6^*=-3$).  Only its moments enter below, so the benchmark
defined from it is a lower bound but need not be the energy of a feasible
quasicode.

The binomial moments determine the low-weight MacWilliams coordinates
through the standard identity
\begin{equation}
 \binom{n-w}{j}=Q^{-j}\sum_{s=0}^j
 \binom{n-s}{n-j}K_s^{(n,Q)}(w).
 \label{eq:QMDS-binomial-Krawtchouk}
\end{equation}
The next lemma applies this relation to $A^*$.

\begin{lemma}[Formal MDS moments and quantum normalization]
\label{lem:QMDS-formal-moments}
The vector $A^*$ in \eqref{eq:QMDS-Astar} satisfies
\begin{equation}
 \sum_{w=0}^n\binom{n-w}{j}A_w^*
 =\binom njQ^{t-j}
 \qquad(0\leq j\leq t).
 \label{eq:QMDS-formal-moments}
\end{equation}
Consequently,
\begin{equation}
 \sum_{w=0}^nA_w^*=Q^t,
 \qquad
 \sum_{w=0}^nA_w^*K_s^{(n,Q)}(w)=0
 \quad(1\leq s\leq t).
 \label{eq:QMDS-formal-transform}
\end{equation}
If a quantum MDS projector $P_{\mathrm{MDS}}$ with parameters
\eqref{eq:QMDS-parameters} exists, then, by the purity theorem and enumerator
formula of Huber and Grassl, which hold for every (not necessarily additive)
quantum MDS code of distance at least $3$
\cite[Theorems~5 and~9]{HuberGrassl2020QMDS},
$A(P_{\mathrm{MDS}})=A^*$ and $B_s(P_{\mathrm{MDS}})=0$ for $1\leq s\leq t$.
\end{lemma}

The binomial inversion and the normalization comparison are given in
Appendix~\ref{app:enumerators}.

Define the associated formal benchmark by
\begin{equation}
 \delta^{\mathrm{QMDS}}_{n,q,\Delta}
 \coloneqq\frac{K^2}{q^n}\sum_{w=d_*}^nA_w^*W_{n,Q}(w).
 \label{eq:QMDS-benchmark}
\end{equation}

Abbreviate $W(w)\coloneqq W_{n,Q}(w)$.  For $0\leq j\leq n-1$, let
\begin{equation}
 \mu_j\coloneqq\sum_{u=0}^j(-1)^u\binom ju W(n-j+u)>0.
 \label{eq:QMDS-mu}
\end{equation}
Define the forbidden-coordinate defect, which collects the $A$-coordinates
below the formal dual distance $d_*$, where $A^*$ vanishes, and the
$B$-coordinates below $\Delta$, which purity forbids:
\begin{align}
 \operatorname{Def}_\Delta(P)
 &\coloneqq\frac{K^2}{q^n}\sum_{w=1}^{d_*-1}A_w(P)
   +\frac K{Q^t}\sum_{s=1}^tB_s(P)
 \label{eq:QMDS-operator-defect}\\
 &=\frac1{q^n}\sum_{1\leq\wt(a)<d_*}|\Tr(PU_a)|^2
   \notag\\
 &\quad+\frac1{Q^t}\sum_{1\leq\wt(a)\leq t}
      \|PU_aP\|_{\HS}^2.\notag
\end{align}

\begin{mainresult}[Preview of the exact quantum MDS theorem]
At the formal Singleton parameters \eqref{eq:QMDS-parameters},
Theorem~\ref{thm:QMDS} proves
\[
 \cD_Q(P)-\delta^{\mathrm{QMDS}}_{n,q,\Delta}
 \geq\frac{\mu_t}{2}\,\operatorname{Def}_\Delta(P).
\]
It also gives the full nonnegative remainder.  The formal benchmark is
attained precisely when $P$ defines a pure code of distance at least $\Delta$.
\end{mainresult}

The remainder must detect both low-weight trace and overlap coordinates.
Two consecutive MDS interpolants provide complementary terms: one detects
the last forbidden trace layer, and the other detects the last forbidden
overlap layer.  Their average will control both.  We write $R_m$ for the
trace-coordinate residual and $\widehat h_m$ for the overlap coefficients.
For $m\in\{t-1,t\}$, define
\begin{align}
 h_m(w)&\coloneqq\sum_{j=0}^m\mu_j\binom{n-w}{j},
 \notag\\
 R_m(w)&\coloneqq W(w)-h_m(w),
 \label{eq:QMDS-h-R}\\
 \widehat h_m(s)&\coloneqq\sum_{j=s}^m\mu_jQ^{-j}
 \binom{n-s}{n-j},
 &&0\leq s\leq m,
 \label{eq:QMDS-hhat}
\end{align}
with $\widehat h_m(s)\coloneqq0$ when $s>m$.  Finally set
\begin{equation}
 \overline R\coloneqq\frac{R_{t-1}+R_t}{2},
 \qquad
 \overline h_s\coloneqq\frac{\widehat h_{t-1}(s)+\widehat h_t(s)}2.
 \label{eq:QMDS-averages}
\end{equation}

These are the two consecutive MDS interpolation polynomials of Cohn and Zhao
\cite[Lemma~16 and Proposition~21]{CohnZhao2014Energy}, written for the
potential $W$.

\begin{theorem}[Exact quantum MDS remainder]
\label{thm:QMDS}
Under the formal parameter assumptions
\eqref{eq:QMDS-t-dstar}--\eqref{eq:QMDS-parameters}, every rank-$K$ projector
$P$ satisfies the exact identity
\begin{equation}
 \begin{aligned}
 \cD_Q(P)-\delta^{\mathrm{QMDS}}_{n,q,\Delta}
 &=\frac{K^2}{q^n}\sum_{w=1}^{d_*-1}\overline R(w)A_w(P)
 +K\sum_{s=1}^t\overline h_sB_s(P).
 \end{aligned}
 \label{eq:QMDS-exact-gap}
\end{equation}
The coefficients obey
\begin{equation}
 \min_{1\leq w<d_*}\overline R(w)=\frac{\mu_t}{2},
 \qquad
 \min_{1\leq s\leq t}\overline h_s
 =\frac{\mu_t}{2Q^t},
 \label{eq:QMDS-coefficient-minima}
\end{equation}
and hence
\begin{equation}
 \cD_Q(P)-\delta^{\mathrm{QMDS}}_{n,q,\Delta}
 \geq\frac{\mu_t}{2}\operatorname{Def}_\Delta(P)\geq0.
 \label{eq:QMDS-stability}
\end{equation}
Equality with the formal benchmark,
$\cD_Q(P)=\delta^{\mathrm{QMDS}}_{n,q,\Delta}$, holds if and only if
\begin{equation}
 B_s(P)=0\qquad(1\leq s<\Delta).
 \label{eq:QMDS-equality}
\end{equation}
Thus, using \cite[Section~6]{HuberGrassl2020QMDS}, the formal benchmark is
attained precisely by the pure quantum MDS projectors with these parameters.  If no such projector exists, then
$\cD_Q(P)>\delta^{\mathrm{QMDS}}_{n,q,\Delta}$ for every rank-$K$ projector.
\end{theorem}

\begin{proof}
The strict complete monotonicity in \eqref{eq:W-complete-monotone} and Newton's
interpolation formula at the endpoint $n$ give
\begin{equation}
 \begin{aligned}
 W(w)&=\sum_{j=0}^{n-w}\mu_j\binom{n-w}{j},\\
 R_m(w)&=\sum_{j=m+1}^{n-w}\mu_j\binom{n-w}{j},
 \end{aligned}
 \label{eq:QMDS-Newton}
\end{equation}
for $1\leq w\leq n$, with an empty sum interpreted as zero.  In particular,
$R_{t-1}$ vanishes for $w\geq d_*$, whereas $R_t$ vanishes already for
$w\geq d_*-1$.  Equation~\eqref{eq:QMDS-binomial-Krawtchouk} shows that
$h_m=\sum_{s=0}^m\widehat h_m(s)K_s^{(n,Q)}$, with every coefficient positive.

To isolate the discrepancy excess, fix $m\in\{t-1,t\}$ and apply
Lemma~\ref{lem:certificate-transfer} with $\widehat h=\widehat h_m$,
$h=h_m$, $r=R_m$ and $a=A^*$.  Its hypotheses hold: by
Lemma~\ref{lem:QMDS-formal-moments}, $A^*$ has mass $Q^t=q^n/K$, is supported
on $\{0,d_*,\ldots,n\}$, and its transform coordinates $1\leq s\leq t$
vanish, while $R_m(w)=0$ for $w\geq d_*$ by \eqref{eq:QMDS-Newton}.  By
\eqref{eq:QMDS-benchmark}, the resulting benchmark is
$\delta^{\mathrm{QMDS}}_{n,q,\Delta}$, so
\begin{equation}
 \begin{aligned}
 \cD_Q(P)-\delta^{\mathrm{QMDS}}_{n,q,\Delta}
 &=\frac{K^2}{q^n}\sum_{w=1}^nR_m(w)A_w(P)\\
 &\quad+K\sum_{s=1}^m\widehat h_m(s)B_s(P).
 \end{aligned}
 \label{eq:QMDS-two-gaps}
\end{equation}
Averaging the two identities proves \eqref{eq:QMDS-exact-gap}.

We now read the stability constant $\mu_t/2$ from the two boundary coefficients.
At $w=d_*-1=n-t$ one has
$R_{t-1}(w)=\mu_t$ and $R_t(w)=0$.  At every smaller $w$, the second expression
in \eqref{eq:QMDS-Newton} contains this positive contribution and additional
nonnegative terms.  Hence the first minimum in
\eqref{eq:QMDS-coefficient-minima} is $\mu_t/2$.  Similarly,
$\overline h_t=\mu_t/(2Q^t)$, while for $s<t$ the $j=t$ summand in
$\widehat h_t(s)$, together with the other nonnegative summands, gives a strictly
larger value.  This proves the second minimum and
\eqref{eq:QMDS-stability}; the second line of
\eqref{eq:QMDS-operator-defect} follows directly from
\eqref{eq:A-def} and \eqref{eq:B-def}.

It remains to identify the equality projectors.  All terms in
\eqref{eq:QMDS-exact-gap} are nonnegative, so benchmark equality
implies
$B_s(P)=0$ for $1\leq s\leq t$.  Conversely, these equations say that every
nonidentity Pauli error of weight below $\Delta$ has zero compression to the
code space.  Thus $P$ is a pure code of distance at least $\Delta$.  Its rank
$K=q^{n-2\Delta+2}$ meets the quantum Singleton bound
\cite[Section~6]{HuberGrassl2020QMDS}, so its distance is exactly $\Delta$ and
it is quantum MDS.  The Huber--Grassl enumerator formula
then gives $A(P)=A^*$, and all terms in \eqref{eq:QMDS-exact-gap} vanish.
\end{proof}

\subsection{Error-image overlaps and subsystem purities}

The two enumerator sums in Theorem~\ref{thm:QMDS} have direct geometric
meanings.  They control both the mutual position of Pauli error images
of radius $e$ with $2e<\Delta$ and the mixedness of small subsystem reductions.

Under the assumptions of Theorem~\ref{thm:QMDS}, fix a rank-$K$ projector
$P$ and an isometry $J:\mathbb C^K\to\Hphys$ with $JJ^\dagger=P$.
For every integer $e\geq1$
with $2e<\Delta$, define
\begin{equation}
 \begin{aligned}
 Y_e&:\mathbb C^{V_e}\otimes\mathbb C^K\longrightarrow\Hphys,\\
 Y_e(|a\rangle\otimes x)&\coloneqq U_aJx,\qquad a\in\cB_e,
 \end{aligned}
 \label{eq:QMDS-Ye}
\end{equation}
where the standard basis of $\mathbb C^{V_e}$ is indexed by $\cB_e$, and put
\begin{equation}
 C_e^{\mathrm{Gram}}
 \coloneqq\max_{1\leq w\leq2e}\frac{c_e(w)}{\widehat h_t(w)}.
 \label{eq:QMDS-exact-Gram-constant}
\end{equation}
Here $c_e(w)$ is the intersection size of two radius-$e$ Hamming balls whose
centers are at distance $w$, as in Proposition~\ref{prop:fusion}.

Put $\rho\coloneqq P/K$ and, for $S\subseteq[n]$, let
$\rho_S\coloneqq\Tr_{S^c}\rho$ and let $I_S$ denote the identity on
$(\mathbb C^q)^{\otimes S}$.  The following exact identities are the bridges
from the two enumerator sums in Theorem~\ref{thm:QMDS} to the geometric
quantities used below:
\begin{align}
 \|Y_e^\dagger Y_e-I_{KV_e}\|_{\HS}^2
 &=K\sum_{w=1}^{2e}c_e(w)B_w(P),
 \label{eq:QMDS-Gram-bridge}\\
 \sum_{\substack{S\subseteq[n]\\|S|=s}}
 \|\rho_S-q^{-s}I_S\|_{\HS}^2
 &=q^{-s}\sum_{w=1}^s\binom{n-w}{s-w}A_w(P).
 \label{eq:QMDS-marginal-bridge}
\end{align}
The proof of Corollary~\ref{cor:QMDS-rigidity} below derives these
identities and combines them with the exact identities
\eqref{eq:QMDS-two-gaps}.

\begin{corollary}[Higher-radius Gram and marginal rigidity]
\label{cor:QMDS-rigidity}
For every integer $e\geq1$ with $2e<\Delta$, the error-image map $Y_e$
satisfies
\begin{equation}
 \begin{aligned}
 \|Y_e^\dagger Y_e-I_{KV_e}\|_{\HS}^2
 &\leq C_e^{\mathrm{Gram}}
 \bigl(\cD_Q(P)-\delta^{\mathrm{QMDS}}_{n,q,\Delta}\bigr)\\
 &\leq\frac{V_eq^n}{K\mu_t}
 \bigl(\cD_Q(P)-\delta^{\mathrm{QMDS}}_{n,q,\Delta}\bigr).
 \end{aligned}
 \label{eq:QMDS-Gram-rigidity}
\end{equation}
Moreover, for $1\leq s<d_*$,
\begin{equation}
 \begin{aligned}
 \sum_{|S|=s}\|\rho_S-q^{-s}I_S\|_{\HS}^2
 &\leq\frac{q^{n-s}}{K^2\mu_t}\binom{n-1}{s-1}\binom{n-1}{t}^{-1}
 \bigl(\cD_Q(P)-\delta^{\mathrm{QMDS}}_{n,q,\Delta}\bigr).
 \end{aligned}
 \label{eq:QMDS-marginal-rigidity}
\end{equation}
\end{corollary}

\begin{proof}
Expanding the Gram matrix and grouping ordered error pairs by the weight of
their difference proves \eqref{eq:QMDS-Gram-bridge}.  The definition of
$C_e^{\mathrm{Gram}}$ and the case $m=t$ of \eqref{eq:QMDS-two-gaps}, whose
terms are nonnegative, then give
\[
 \|Y_e^\dagger Y_e-I_{KV_e}\|_{\HS}^2
 \leq C_e^{\mathrm{Gram}}
 \bigl(\cD_Q(P)-\delta^{\mathrm{QMDS}}_{n,q,\Delta}\bigr).
\]
Since $c_e(w)\leq V_e$ and the $j=t$ summand in \eqref{eq:QMDS-hhat} gives
$\widehat h_t(w)\geq\mu_t/Q^t$, one has
\[
 C_e^{\mathrm{Gram}}\leq\frac{V_eQ^t}{\mu_t}
 =\frac{V_eq^n}{K\mu_t}.
\]
This proves \eqref{eq:QMDS-Gram-rigidity}.  Pauli Parseval on each subsystem,
followed by counting the $s$-subsets containing a fixed support, proves
\eqref{eq:QMDS-marginal-bridge}.  For $1\leq w<s$, the ratio
$\binom{n-w}{s-w}/\binom{n-w}{t}$ changes by the factor
$(s-w)/(n-t-w)\leq1$ from $w$ to $w+1$.  This comparison, the case $m=t-1$ of
\eqref{eq:QMDS-two-gaps} and $R_{t-1}(w)\geq\mu_t\binom{n-w}{t}$ from
\eqref{eq:QMDS-Newton} prove \eqref{eq:QMDS-marginal-rigidity}.
\end{proof}

The coefficientwise comparison can be much sharper than its uniform upper
bound.  At the $((5,2,3))_2$ parameters one has
$\mu_2=96$, $\widehat h_2(1)=34$, $\widehat h_2(2)=6$ and
$c_1(1)=4,c_1(2)=2$, so
\[
 C_1^{\mathrm{Gram}}=\frac13,
\]
whereas the uniform upper bound in \eqref{eq:QMDS-Gram-rigidity} is $8/3$.
At the $((6,1,4))_2$ parameters the value of $C_1^{\mathrm{Gram}}$ is $1/30$.
Each rigidity bound uses the single identity in \eqref{eq:QMDS-two-gaps} that
retains its layer, $m=t$ for the Gram bound and $m=t-1$ for the marginal
bound; since $\widehat h_t(w)\geq\overline h_w$, the Gram constant is at most
the one obtained from the averaged certificate
\eqref{eq:QMDS-averages}.  Only \eqref{eq:QMDS-stability}, which
needs both layers, uses the average: among weighted averages of the two
identities \eqref{eq:QMDS-two-gaps}, equal weights maximize its constant
(weights $\lambda,1-\lambda$ give $\min\{\lambda,1-\lambda\}\mu_t$).

\subsection{Polar recovery from the discrepancy excess}
\label{sec:qmds-recovery}

Gram control says that the error images are almost mutually orthogonal.
To construct a recovery, we replace their joint map by its polar isometry and
decode by its adjoint.  The same construction serves the one-error regime of
Section~\ref{sec:rigidity}, so we state it first for an arbitrary set
$\cE$ of Weyl labels and then specialize to $\cE=\cB_e$.

For errors in the span of the $U_a$, a polar recovery produces a channel
whose Kraus operators are linear combinations of blocks of a positive square
root: if $Y=\widetilde YG^{1/2}$ with $\widetilde Y$ an isometry, then
$\widetilde Y^\dagger Y(|a\rangle\otimes x)=\sum_b|b\rangle\otimes(G^{1/2})_{ba}x$,
with $K$-by-$K$ blocks $(G^{1/2})_{ba}$.
That square root is Hermitian, so trace preservation
turns the first-order terms into the quadratic term
$-\frac12(\Theta X+X\Theta)$ of \eqref{eq:polar-quadratic-identity} below,
and the channel error is quadratic in the deviation of $G^{1/2}$ from the
identity.  The following elementary estimate makes this explicit.
Here $\id_K$ denotes the identity channel on operators over $\mathbb C^K$,
and $\|\cdot\|_\diamond$ denotes the unhalved diamond norm.

\begin{lemma}[Quadratic error of a polar recovery]
\label{lem:polar-linear-diamond}
Let $G\succeq0$ act on $\mathbb C^v\otimes\mathbb C^K$, put
$S\coloneqq\sqrt G$ and $T\coloneqq S-I_{vK}$, and let
$c_1,\ldots,c_m\in\mathbb C^v$ satisfy
\[
 \sum_{j=1}^mC_j^\dagger C_j=I_K=\sum_{j=1}^mC_j^\dagger GC_j,
 \qquad C_j\coloneqq c_j\otimes I_K .
\]
Define the channel
\[
 \Lambda(X)\coloneqq\sum_{j=1}^m\sum_{b=1}^vK_{jb}XK_{jb}^\dagger,
 \qquad K_{jb}\coloneqq(\langle b|\otimes I_K)SC_j,
\]
put $T_{jb}\coloneqq(\langle b|\otimes I_K)TC_j$, and
\[
 \Theta\coloneqq\sum_{j=1}^mC_j^\dagger T^2C_j,
 \qquad
 \mathcal A(X)\coloneqq\sum_{j=1}^m\sum_{b=1}^vT_{jb}XT_{jb}^\dagger.
\]
Then
\begin{equation}
 \Lambda(X)-X=\mathcal A(X)-\frac12(\Theta X+X\Theta),
 \label{eq:polar-quadratic-identity}
\end{equation}
and
\begin{equation}
 \|\Lambda-\id_K\|_\diamond
 \leq2\min\{1,\|\Theta\|_{\mathrm{op}}\}
 \leq2\min\{1,\Tr\Theta\}.
 \label{eq:polar-linear-diamond}
\end{equation}
\end{lemma}

\begin{proof}
Since $\sum_b|b\rangle\langle b|\otimes I_K=I_{vK}$, one has
$\sum_{j,b}K_{jb}^\dagger K_{jb}=\sum_jC_j^\dagger GC_j=I_K$, proving trace
preservation.  If $c_{jb}$ is the $b$-th entry of $c_j$, then
$K_{jb}=c_{jb}I_K+T_{jb}$, $\sum_{j,b}|c_{jb}|^2=1$ by the first hypothesis,
and $\sum_{j,b}\overline{c_{jb}}T_{jb}=H\coloneqq\sum_jC_j^\dagger TC_j$,
which is Hermitian because $T$ is.  Hence
$\Lambda(X)=X+HX+XH+\mathcal A(X)$.  Since $G=(I_{vK}+T)^2$, the two
hypotheses give $I_K=I_K+2H+\Theta$, and substituting $H=-\Theta/2$ proves
\eqref{eq:polar-quadratic-identity}.  For the completely positive map
$\mathcal A$, Stinespring factorization gives
$\|\mathcal A\|_\diamond=\|\mathcal A^\dagger(I_K)\|_{\mathrm{op}}
=\|\Theta\|_{\mathrm{op}}$.  Left or right multiplication by $\Theta$ has
diamond norm at most $\|\Theta\|_{\mathrm{op}}$.  The triangle inequality
therefore bounds the channel difference by $2\|\Theta\|_{\mathrm{op}}$.
The universal channel-distance bound $2$ and
$\|\Theta\|_{\mathrm{op}}\leq\Tr\Theta$ give
\eqref{eq:polar-linear-diamond}.
\end{proof}

For a positive operator $X$, let $X^{+1/2}$ denote its Moore--Penrose inverse
square root on its support.

On the range of an error-image map, the recovery below applies the adjoint of
its polar isometry and discards the error label; there it is the polar
decoder of \cite[Definition~2.2]{MaTanWright2025Haar}, the square-root
construction of \cite{EldarForney2001SquareRoot} applied to the error-image
map.  On the images of code states under every channel whose Kraus
operators lie in $\operatorname{span}\{U_a:a\in\cE\}$, it coincides with the
error-set Petz map \cite[Definition~6, Proposition~6 and
Lemma~9]{ElimelechAlbertBarg2026UniversalRecovery}.  We write
$\Fent(\rho,\Lambda)$ for the entanglement fidelity of the channel $\Lambda$
at the state $\rho$.

\begin{lemma}[Fixed polar recovery of an error set]
\label{lem:polar-decoder}
Let $J:\mathbb C^K\to\Hphys$ be an isometry with $JJ^\dagger=P$, put
$\operatorname{Enc}(\rho)\coloneqq J\rho J^\dagger$, and let $\cE$ be a set of
$v$ Weyl labels with $vK\leq d$, where $K$ divides $d$.  Index the standard
basis of $\mathbb C^v$ by $\cE$ and put
\[
 Y_\cE(|a\rangle\otimes x)\coloneqq U_aJx\ (a\in\cE),\qquad
 G_\cE\coloneqq Y_\cE^\dagger Y_\cE,\qquad
 T_\cE(P)\coloneqq Y_\cE Y_\cE^\dagger=\sum_{a\in\cE}U_aPU_a^\dagger.
\]
Then $Y_\cE=\widetilde Y_\cE G_\cE^{1/2}$ for an isometry
$\widetilde Y_\cE:\mathbb C^v\otimes\mathbb C^K\to\Hphys$, and
\begin{equation}
 \|Y_\cE-\widetilde Y_\cE\|_{\HS}^2
 =\Tr\bigl(G_\cE^{1/2}-I_{vK}\bigr)^2
 \leq\|G_\cE-I_{vK}\|_{\HS}^2.
 \label{eq:polar-decoder-error}
\end{equation}
Choose a unitary $D:\Hphys\to\mathbb C^{d/K}\otimes\mathbb C^K$ and an
isometric flag embedding $\iota:\mathbb C^v\to\mathbb C^{d/K}$ such that
$D\widetilde Y_\cE(|a\rangle\otimes x)=\iota|a\rangle\otimes x$ for every
$a\in\cE$; when $vK=d$ one may take $\iota$ to be the identity and
$D=\widetilde Y_\cE^\dagger$.  The CPTP recovery
$\cR(\sigma)\coloneqq\Tr_{\mathbb C^{d/K}}(D\sigma D^\dagger)$ depends on $P$,
$\cE$ and $J$, but on no noise parameter.
\begin{enumerate}
\item For a probability vector $p=(p_a)_{a\in\cE}$, put
$\cN_p(\sigma)\coloneqq\sum_{a\in\cE}p_aU_a\sigma U_a^\dagger$,
$p_\ast\coloneqq\max_ap_a$, and the probability-weighted decoder mismatch
$\Mis_p\coloneqq\sum_{a\in\cE}p_a\|DU_aJ-\iota|a\rangle\otimes I_K\|_{\HS}^2$.
Then $\Mis_p\leq p_\ast\|Y_\cE-\widetilde Y_\cE\|_{\HS}^2$,
\[
 \|\cR\circ\cN_p\circ\operatorname{Enc}-\id_K\|_\diamond
 \leq2\min\{1,\Mis_p\},
 \qquad
 1-\Fent(I_K/K,\cR\circ\cN_p\circ\operatorname{Enc})
 \leq\min\{1,\Mis_p/K\}.
\]
\item On $\operatorname{supp}T_\cE(P)$ the block rows of $D$ are
$((\langle a|\iota^\dagger)\otimes I_K)D=J^\dagger U_a^\dagger T_\cE(P)^{+1/2}$,
so for the uniform distribution $\cR$ is the transpose/Petz recovery there, and
$\Fent(I_K/K,\cR\circ\cN_{\mathrm{unif}}\circ\operatorname{Enc})
=(K^2v)^{-1}\|\Tr_{\mathbb C^K}\sqrt{G_\cE}\|_{\HS}^2$.
\item For every channel $\cN$ whose Kraus operators lie in
$\operatorname{span}\{U_a:a\in\cE\}$,
$\|\cR\circ\cN\circ\operatorname{Enc}-\id_K\|_\diamond
\leq2\min\{1,\Tr(G_\cE^{1/2}-I_{vK})^2\}$ and
$1-\Fent(I_K/K,\cR\circ\cN\circ\operatorname{Enc})
\leq K^{-1}\Tr(G_\cE^{1/2}-I_{vK})^2$.
\end{enumerate}
\end{lemma}

\begin{proof}
Since $vK\leq d$, the polar partial isometry of $Y_\cE$ completes on its
kernel to an isometry $\widetilde Y_\cE$ with
$Y_\cE=\widetilde Y_\cE G_\cE^{1/2}$.  If $\lambda_1,\ldots,\lambda_{vK}$ are
the eigenvalues of $G_\cE$, then
$\|Y_\cE-\widetilde Y_\cE\|_{\HS}^2=\sum_i(\sqrt{\lambda_i}-1)^2
\leq\sum_i(\lambda_i-1)^2$, which is \eqref{eq:polar-decoder-error}.  The
isometry from $\operatorname{ran}(\widetilde Y_\cE)$ onto
$\iota(\mathbb C^v)\otimes\mathbb C^K$ extends to the unitary $D$ because
both ambient spaces have dimension $d$.

For item~1, unitarity of $D$ gives
\[
 \sum_{a\in\cE}\|DU_aJ-\iota|a\rangle\otimes I_K\|_{\HS}^2
 =\|Y_\cE-\widetilde Y_\cE\|_{\HS}^2,
\]
hence $\Mis_p\leq p_\ast\|Y_\cE-\widetilde Y_\cE\|_{\HS}^2$.  The polar identity gives
$DY_\cE=(\iota\otimes I_K)\sqrt{G_\cE}$, including when $G_\cE$ is singular.
Hence, if a channel has Kraus operators $L_j=\sum_{a\in\cE}c_{ja}U_a$ and
$c_j\coloneqq(c_{ja})_{a\in\cE}$, then
$DL_jJ=(\iota\otimes I_K)\sqrt{G_\cE}\,(c_j\otimes I_K)$; since $\iota$ is an
isometry, the recovered channel is the channel $\Lambda$ of
Lemma~\ref{lem:polar-linear-diamond} with $G=G_\cE$.  Its hypotheses hold:
$\sum_jC_j^\dagger G_\cE C_j=J^\dagger(\sum_jL_j^\dagger L_j)J=I_K$, and
\eqref{eq:weyl-orthogonality} gives
$d=\Tr\sum_jL_j^\dagger L_j=d\sum_j\|c_j\|^2$.  For $\cN_p$ take
$c_a=\sqrt{p_a}\,|a\rangle$; then $\Tr\Theta=\Mis_p$, and that lemma gives
the diamond bound.  Let $E$ be a copy of $\mathbb C^v$ with
basis indexed by $\cE$, and let $F\coloneqq\mathbb C^{d/K}$ be the decoded
flag register.  The isometries
$V_{D,p}x\coloneqq\sum_{a\in\cE}\sqrt{p_a}\,|a\rangle_E\otimes DU_aJx$ and
$V_{0,p}x\coloneqq\sum_{a\in\cE}\sqrt{p_a}\,|a\rangle_E\otimes\iota|a\rangle_F\otimes x$
are Stinespring isometries of the recovered and identity channels, and
$\|V_{D,p}-V_{0,p}\|_{\HS}^2=\Mis_p$.  On a maximally entangled logical
input the squared distance between these purifications is $\Mis_p/K$, and
monotonicity of fidelity under partial trace bounds entanglement infidelity
by this quantity \cite{KretschmannEtAl2008Continuity}.  This
entanglement-infidelity estimate, in the form $\Mis_p/K$, can also be read
from \cite[Theorem~7 and Eqs.~(26)--(27) in its
proof]{ElimelechAlbertBarg2026UniversalRecovery}, since $\cR$ is the
error-set Petz map on the images of code states; we include the direct
purification argument for completeness.

For item~2, $((\langle a|\iota^\dagger)\otimes I_K)D
=(\langle a|\otimes I_K)\widetilde Y_\cE^\dagger
=J^\dagger U_a^\dagger T_\cE(P)^{+1/2}$ on $\operatorname{supp}T_\cE(P)$.
These are the transpose/Petz Kraus operators
\cite{ZhengEtAl2024NearOptimal}; that the polar decoder equals the
error-set Petz map is
\cite[Proposition~6]{ElimelechAlbertBarg2026UniversalRecovery}.  The QEC
matrix of the uniform channel, with Kraus operators $v^{-1/2}U_aJ$, is
$G_\cE/v$ up to the order of the tensor factors (error label $\otimes$
logical here, exchanged in \cite{ZhengEtAl2024NearOptimal}), so the fidelity
formula is \cite[Theorem~1]{ZhengEtAl2024NearOptimal}.

For item~3, the argument for item~1 shows that
$\cR\circ\cN\circ\operatorname{Enc}$ is a channel $\Lambda$ of
Lemma~\ref{lem:polar-linear-diamond}.  Since
$\rho\coloneqq\sum_jc_jc_j^\dagger$ has trace one,
$\Tr\Theta=\Tr\bigl((\rho\otimes I_K)(G_\cE^{1/2}-I_{vK})^2\bigr)
\leq\Tr(G_\cE^{1/2}-I_{vK})^2$, which gives the diamond bound.
For the fidelity bound, $\cN$ is controlled by the error set $\{U_a\}_{a\in\cE}$
(in the sense of \cite[Definition~2]{ElimelechAlbertBarg2026UniversalRecovery})
by \eqref{eq:weyl-orthogonality} and
\cite[Proposition~4]{ElimelechAlbertBarg2026ErrorSet}, and on the images of
code states $\cR$ is the error-set Petz map
\cite[Lemma~9 and Proposition~6]{ElimelechAlbertBarg2026UniversalRecovery}.
Let $M_\cE\coloneqq\Pi^\dagger G_\cE\Pi$, where
$\Pi:\mathbb C^K\otimes\mathbb C^v\to\mathbb C^v\otimes\mathbb C^K$ exchanges the
tensor factors; $M_\cE$ is the QEC matrix of
\cite[Eq.~(4)]{ElimelechAlbertBarg2026UniversalRecovery} for the code basis
$J|i\rangle$.  By \cite[Theorem~7 and
Definition~4]{ElimelechAlbertBarg2026UniversalRecovery},
\[
 1-\Fent(I_K/K,\cR\circ\cN\circ\operatorname{Enc})
 \leq\frac1K\min_{\lambda\succeq0}
 \bigl\|\sqrt{M_\cE}-I_K\otimes\sqrt\lambda\bigr\|_{\HS}^2,
\]
the minimum running over positive semidefinite $v\times v$ matrices $\lambda$.
Taking $\lambda=I_v$ and conjugating by $\Pi$ bounds the right side by
$K^{-1}\Tr(G_\cE^{1/2}-I_{vK})^2$.
\end{proof}

\begin{lemma}[Gram control and polar approximate isometry]
\label{lem:QMDS-polar}
Under the hypotheses of Theorem~\ref{thm:QMDS}, let $1\leq e$ satisfy
$2e<\Delta$, and assume
\begin{equation}
 KV_e\leq q^n.
 \label{eq:QMDS-polar-dimension}
\end{equation}
Put
\[
 \begin{aligned}
 g(P)&\coloneqq\cD_Q(P)-\delta^{\mathrm{QMDS}}_{n,q,\Delta},
 &G_e&\coloneqq Y_e^\dagger Y_e,\\
 &&C_e&\coloneqq C_e^{\mathrm{Gram}}.
 \end{aligned}
\]
Then $Y_e=\widetilde Y_eG_e^{1/2}$ for an isometry
$\widetilde Y_e:\mathbb C^{V_e}\otimes\mathbb C^K\to\Hphys$, and
\begin{align}
 \|Y_e-\widetilde Y_e\|_{\HS}^2
 &=\Tr\bigl(G_e^{1/2}-I_{KV_e}\bigr)^2
 \label{eq:QMDS-polar-error}\\
 &\leq\|G_e-I_{KV_e}\|_{\HS}^2
 \leq C_e\,g(P).
 \notag
\end{align}
Moreover,
\begin{equation}
 \gamma_e\coloneqq\|G_e-I_{KV_e}\|_{\mathrm{op}}
 \leq\sqrt{C_e\,g(P)}.
 \label{eq:QMDS-gamma}
\end{equation}
If $C_e g(P)<1$, then the singular values of $Y_e$ lie in
$[\sqrt{1-\gamma_e},\sqrt{1+\gamma_e}]$ and
\begin{equation}
 \begin{aligned}
\|G_e^{1/2}-I_{KV_e}\|_{\mathrm{op}}
 &\leq1-\sqrt{1-\gamma_e}\\
 &\leq1-\sqrt{1-\sqrt{C_e g(P)}}.
 \end{aligned}
 \label{eq:QMDS-delta}
\end{equation}
\end{lemma}

\begin{proof}
Lemma~\ref{lem:polar-decoder} with $\cE=\cB_e$ and $v=V_e$ (here
$K=q^\ell$ divides $q^n$) gives the polar factorization and the identity and
first inequality in \eqref{eq:QMDS-polar-error}.
Corollary~\ref{cor:QMDS-rigidity} gives the final inequality in
\eqref{eq:QMDS-polar-error}, and operator norm is bounded by Hilbert--Schmidt
norm, proving \eqref{eq:QMDS-gamma}.  Under the stated small-excess condition,
scalar functional calculus gives the singular-value interval and
\eqref{eq:QMDS-delta}.
\end{proof}

The coefficientwise constant in Corollary~\ref{cor:QMDS-rigidity} therefore
has an operational consequence at every radius $e$ with $2e<\Delta$ and
$KV_e\leq q^n$.  Let $\cR_e$ be the recovery of
Lemma~\ref{lem:polar-decoder} for $\cE=\cB_e$, with unitary $D_e$ and flag
embedding $\iota_e$.

\begin{corollary}[Recovery of channels in the error span]
\label{cor:QMDS-span-infidelity}
Let $2e<\Delta$ and $KV_e\leq q^n$ at the parameters
\eqref{eq:QMDS-parameters}.  For every channel $\cN$ whose Kraus operators lie in
$\operatorname{span}\{U_a:a\in\cB_e\}$, the logical channel
$\Lambda\coloneqq\cR_e\circ\cN\circ\operatorname{Enc}$ satisfies
\[
 \|\Lambda-\id_K\|_\diamond\leq2\min\{1,C_e\,g(P)\},
 \qquad
 1-\Fent(I_K/K,\Lambda)\leq\frac{C_e\,g(P)}{K},
\]
with no condition on the size of $g(P)$.
\end{corollary}

\begin{proof}
Lemma~\ref{lem:polar-decoder}(3) gives both bounds, since
$\Tr(G_e^{1/2}-I_{KV_e})^2\leq C_e\,g(P)$ by \eqref{eq:QMDS-polar-error}.
\end{proof}

For Pauli mixtures on $\cB_e$, Corollary~\ref{cor:QMDS-recovery} below adds
the factor $p_\ast$, the largest probability in the mixture including that of
the identity.  In the entanglement-fidelity bound, the factor $p_\ast$ can also
be read from
\cite[Eqs.~(26)--(27)]{ElimelechAlbertBarg2026UniversalRecovery}.  Since
$p_\ast$ is at least the identity probability, these bounds do not vanish with
the noise.  This is partly a property of the recovery: if $PU_aP=0$ for
$1\leq\wt(a)\leq e$, the identity column of $G_e^{1/2}$ is
$(I_K,0,\ldots,0)$ and $\cR_e\circ\operatorname{Enc}=\id_K$, but in general
$\cR_e\circ\operatorname{Enc}\neq\id_K$.  For comparison, the worst-case
guarantees of \cite[Theorem~1.10]{MaTanWright2025Haar}
and \cite[Theorem~11]{ElimelechAlbertBarg2026UniversalRecovery} give bounds
of order $g(P)^{1/2}$ on the diamond error here.

\begin{corollary}[Higher-radius polar recovery at quantum MDS parameters]
\label{cor:QMDS-recovery}
Under the hypotheses of Theorem~\ref{thm:QMDS}, let \(1\leq e\) satisfy
\(2e<\Delta\) and \(KV_e\leq q^n\), let \(g(P)\), \(G_e\) and \(C_e\) be as in
Lemma~\ref{lem:QMDS-polar}, and let \(\cR_e\) be the fixed recovery defined
before Corollary~\ref{cor:QMDS-span-infidelity}.
For a probability vector
\(p=(p_a)_{a\in\cB_e}\), put
\begin{equation}
 \cN_{e,p}(\sigma)\coloneqq\sum_{a\in\cB_e}p_aU_a\sigma U_a^\dagger,
 \qquad p_\ast\coloneqq\max_a p_a.
 \label{eq:QMDS-diagonal-noise}
\end{equation}
The same recovery satisfies the following estimates:
\begin{gather}
 \begin{aligned}
 &\|\cR_e\circ\cN_{e,p}\circ\operatorname{Enc}-\id_K\|_\diamond
 \leq2\min\{1,p_\ast C_e\,g(P)\}
 \end{aligned}
 \label{eq:QMDS-diagonal-diamond}\\
 \begin{aligned}
 &1-\Fent(I_K/K,\cR_e\circ\cN_{e,p}\circ\operatorname{Enc})
 \leq\min\left\{1,\frac{p_\ast C_e\,g(P)}K\right\}.
 \end{aligned}
 \label{eq:QMDS-diagonal-fidelity}
\end{gather}
For the uniform distribution on \(\cB_e\), this recovery is the
transpose/Petz recovery on $\operatorname{ran}(Y_e)$, and
\begin{gather}
 \begin{aligned}
 &\Fent(I_K/K,\cR_e\circ\cN_{e,\mathrm{unif}}
       \circ\operatorname{Enc})
 =\frac1{K^2V_e}
 \left\|\Tr_{\mathbb C^K}\sqrt{G_e}\right\|_{\HS}^2
 \end{aligned}
 \label{eq:QMDS-uniform-exact-fidelity}\\
 \begin{aligned}
 &\|\cR_e\circ\cN_{e,\mathrm{unif}}\circ\operatorname{Enc}
       -\id_K\|_\diamond
 \leq2\min\left\{1,\frac{C_e\,g(P)}{V_e}\right\}
 \end{aligned}
 \label{eq:QMDS-uniform-diamond}\\
 \begin{aligned}
 &1-\Fent(I_K/K,\cR_e\circ\cN_{e,\mathrm{unif}}
       \circ\operatorname{Enc})
 \leq\frac{C_e\,g(P)}{KV_e}.
 \end{aligned}
 \label{eq:QMDS-uniform-recovery}
\end{gather}
\end{corollary}

\begin{proof}
By Lemma~\ref{lem:polar-decoder}(1) with $\cE=\cB_e$ and by
\eqref{eq:QMDS-polar-error}, the decoder mismatch $\Mis_{e,p}\coloneqq\Mis_p$
of that lemma satisfies
$\Mis_{e,p}\leq p_\ast\|Y_e-\widetilde Y_e\|_{\HS}^2\leq p_\ast C_e\,g(P)$,
which proves \eqref{eq:QMDS-diagonal-diamond} and
\eqref{eq:QMDS-diagonal-fidelity}.  Lemma~\ref{lem:polar-decoder}(2) with
$v=V_e$ gives the transpose/Petz identification and
\eqref{eq:QMDS-uniform-exact-fidelity}.  Substituting \(p_\ast=1/V_e\) gives
\eqref{eq:QMDS-uniform-diamond} and \eqref{eq:QMDS-uniform-recovery}.
\end{proof}

The same discrepancy estimate extends to full stochastic Pauli noise by
separating the selected ball from its outside probability mass.
Appendix~\ref{app:noise} gives these tail estimates, their depolarizing
specialization, and a comparison for coherent errors in the selected span.

For prime-dimensional stabilizer projectors, the low-weight enumerators
take integer values constrained by scalar orbits.  This yields a positive
gap above an unattained benchmark; Corollary~\ref{cor:QMDS-stabilizer-gap}
in Appendix~\ref{app:stabilizer} states the exact bound.

\section{One-error all-radii stability}
\label{sec:stability}

At one-error Hamming equality, uniform coverage means that the error images
tile the physical space.  We seek a bound that measures the failure of this
tiling directly from the all-radii discrepancy excess.  The arithmetic
conditions below specify the ranks and interpolation nodes for this bound.

Assume henceforth that
\begin{equation}
 \begin{aligned}
 Q&=q^2,\qquad V\coloneqq1+n(Q-1),\\
 K&=\frac dV\in\mathbb N,\qquad
 k\coloneqq\frac VQ\in\mathbb N,\qquad n\geq3.
 \end{aligned}
 \label{eq:admissible}
\end{equation}
The integrality of $k$ is equivalent to
\begin{equation}
 n=Qs+1,\qquad k=1+s(Q-1)
 \label{eq:s-param}
\end{equation}
for an integer $s\geq1$.  Put $W(w)\coloneqq W_{n,Q}(w)$ and
\begin{equation}
 \delta_{n,Q}\coloneqq\frac{K^2}{d}(V-1)W(k).
 \label{eq:benchmark}
\end{equation}
The condition $K\in\mathbb N$ makes the prescribed projector rank integral,
whereas $k\in\mathbb N$ places the root of
$L_1(w)=V-Qw$ on the Hamming lattice and permits interpolation at the adjacent
nodes $k-1,k$.  For prime-power $q$, rank integrality implies
$k\in\mathbb N$: if $q=p^a$ and $V\mid q^n$, then $V=p^b$, while
$n=(p^b-1)/(p^{2a}-1)\in\mathbb N$ forces $2a\mid b$.  For a general integer
local dimension the two conditions are independent.  Attainment of the formal
benchmark is left open at this stage.
We first record the moment identities that connect quantum feasibility to a
classical quadratic minorant on the integer distance set; the minorant's
coefficient $\kappa_{n,Q}$ enters Theorem~\ref{thm:main}.

\begin{lemma}[One-error moments]
\label{lem:moments}
Under \eqref{eq:admissible}, every rank-$K$ orthogonal projector $P$ satisfies
\begin{align}
 \sum_{w=1}^nA_w&=V-1,\label{eq:mass}\\
 \sum_{w=1}^n(w-k)A_w&=-kB_1,\label{eq:first-moment}\\
 \sum_{w=1}^n(w-k)^2A_w&=kB_1+\frac{2k}{Q}B_2.
 \label{eq:second-moment}
\end{align}
Moreover,
\begin{equation}
 \Phi_1(P)=K(QB_1+2B_2)
 =\frac{K^2}{d}Q^2\sum_{w=1}^n(w-k)^2A_w.
 \label{eq:Phi-one-moments}
\end{equation}
\end{lemma}

\begin{proof}
The mass is \eqref{eq:AB-masses}.  Since
$K_1^{(n,Q)}(w)=(Q-1)n-Qw=Q(k-w)-1$, the $j=1$ MacWilliams identity, after separating
the $w=0$ term, gives \eqref{eq:first-moment}.  The identity
\[
 2K_2^{(n,Q)}(w)=K_1^{(n,Q)}(w)^2-(Q-2)K_1^{(n,Q)}(w)-(Q-1)n
\]
and the $j=2$ MacWilliams identity give \eqref{eq:second-moment} by direct
substitution.  Finally $L_1(w)=V-Qw=-Q(w-k)$, so
\eqref{eq:Phi-spectral} gives the second expression in
\eqref{eq:Phi-one-moments}.  Alternatively, two radius-one $Q$-ary balls have
intersection sizes $Q,2,0$ at center distances $1,2,>2$; hence
\eqref{eq:overlap-form} gives the first expression.
\end{proof}

The corresponding classical coefficient estimate and pointwise minorant appear
in \cite[Proposition~4.3 and Theorem~4.2, Eq.~(41)]{Zabokritskiy2026Stability}.
Let
\begin{align}
 a_{n,Q}&\coloneqq W(k-1)-W(k),\label{eq:one-error-slope-def}\\
 \eta_{n,Q}&\coloneqq
 \min_{\substack{1\leq w\leq n\\w\notin\{k-1,k\}}}
 \frac{W(w)-W(k)+a_{n,Q}(w-k)}{(w-k)(w-k+1)},\label{eq:eta-def}\\
 \kappa_{n,Q}&\coloneqq\min\{a_{n,Q},\eta_{n,Q}\}.
 \label{eq:kappa-def}
\end{align}

\begin{proposition}[Classical spectral coefficient]
\label{prop:kappa}
Under \eqref{eq:s-param}, the minimum defining $\eta_{n,Q}$ is attained
uniquely at $w=n$, and
\begin{equation}
 \kappa_{n,Q}=\eta_{n,Q}
 =\frac{W(n)+sW(k-1)-(s+1)W(k)}{s(s+1)}.
 \label{eq:kappa-closed}
\end{equation}
In particular,
\begin{equation}
 \kappa_{n,Q}\geq Q^2.
 \label{eq:kappa-bound}
\end{equation}
For every integer $1\leq w\leq n$,
\begin{equation}
 W(w)\geq W(k)-a_{n,Q}(w-k)+\kappa_{n,Q}(w-k)(w-k+1).
 \label{eq:minorant}
\end{equation}
\end{proposition}

The exact formula is the $Q\geq4$ case of
\cite[Proposition~4.3, Eq.~(45)]{Zabokritskiy2026Stability}.
Appendix~\ref{app:coefficient} includes its proof: comparing $W$ with its
quadratic interpolant at $k-1,k,n$ places the minimum at the endpoint, so
the coefficient can be evaluated from three values of $W$.

Inserting the scalar minorant into the moment identities gives the desired
stability bound for arbitrary projectors.

\begin{theorem}[Arbitrary-projector stability]
\label{thm:main}
Assume \eqref{eq:admissible}, and for $1\leq w\leq n$ let
\[
 \rho(w)\coloneqq W(w)-W(k)+a_{n,Q}(w-k)-\kappa_{n,Q}(w-k)(w-k+1)\geq0
\]
be the slack of the minorant \eqref{eq:minorant}.  Every rank-$K$ orthogonal
projector $P$ on $(\mathbb C^q)^{\otimes n}$ has the exact nonnegative
remainder
\begin{equation}
 \cD_Q(P)-\delta_{n,Q}
 =\frac{K^2}{d}\Bigl(\sum_{w=1}^n\rho(w)A_w
 +\kappa_{n,Q}\sum_{w=1}^n(w-k)^2A_w+(a_{n,Q}-\kappa_{n,Q})kB_1\Bigr),
 \label{eq:main-remainder}
\end{equation}
with $A_w=A_w(P)$ and $B_1=B_1(P)$, and it satisfies
\begin{equation}
 \begin{aligned}
 \cD_Q(P)-\delta_{n,Q}
 &\geq c_{n,Q}\Phi_1(P)\geq\Phi_1(P),\\
 c_{n,Q}&\coloneqq\frac{\kappa_{n,Q}}{Q^2}\geq1.
 \end{aligned}
 \label{eq:main-stability}
\end{equation}
Furthermore,
\begin{equation}
 \cD_Q(P)=\delta_{n,Q}\quad\Longleftrightarrow\quad \Phi_1(P)=0.
 \label{eq:main-equality}
\end{equation}
\end{theorem}

\begin{proof}
Write
$S_2\coloneqq\sum_{w=1}^n(w-k)^2A_w$.  By \eqref{eq:D-spectral},
\eqref{eq:mass}, and \eqref{eq:benchmark},
\[
 \cD_Q(P)-\delta_{n,Q}
 =\frac{K^2}{d}\sum_{w=1}^nA_w\bigl(W(w)-W(k)\bigr).
\]
Summing the defining identity
$W(w)-W(k)=\rho(w)-a_{n,Q}(w-k)+\kappa_{n,Q}(w-k)(w-k+1)$
against $A_w$ and using \eqref{eq:first-moment} gives
\eqref{eq:main-remainder}. Its three terms are nonnegative by
\eqref{eq:minorant}, \eqref{eq:AB-positive},
\eqref{eq:kappa-def}, and \eqref{eq:kappa-bound}.
Dropping the first and third terms and applying
\eqref{eq:Phi-one-moments} yields
\[
 \cD_Q(P)-\delta_{n,Q}
 \geq\frac{K^2}{d}\kappa_{n,Q}S_2
 =\frac{\kappa_{n,Q}}{Q^2}\Phi_1(P).
\]
Together with \eqref{eq:kappa-bound}, this proves
\eqref{eq:main-stability}.

If $\Phi_1=0$, then $S_2=0$.  Positivity of $A_w$ forces its nonzero-weight
mass to be supported at $w=k$, and \eqref{eq:D-spectral} gives
$\cD_Q=\delta_{n,Q}$.  The converse follows immediately from
\eqref{eq:main-stability}.
\end{proof}

\begin{remark}[Concordance with the classical stability inequality]
\label{rem:one-error-companion-concordance}
For stabilizer projectors, Theorem~\ref{thm:main} agrees exactly with the
lower-bound side of the classical two-sided inequality in
\cite[Theorem~4.2]{Zabokritskiy2026Stability}.  Set the alphabet size in that
paper to $Q=q^2$, take the phase-free normalizer code of size
$dK$, and multiply
its normalized discrepancy by $dK^2$; the result is
\eqref{eq:main-stability}, with the same coefficient $c_{n,Q}$.  The imported
ingredient is the scalar quadratic minorant, reproduced in
Proposition~\ref{prop:kappa}; quantum MacWilliams feasibility lifts its lower
bound from stabilizer distance distributions to every rank-$K$ orthogonal
projector.
\end{remark}

\begin{remark}[Overlap of the two formal benchmarks]
\label{rem:QMDS-one-error-overlap}
Set $\Delta\coloneqq3$ and $n\coloneqq Q+1$.  Then
\[
 V=Q^2,\qquad K=q^{n-4},\qquad k=d_*=Q,
\]
and the formal MDS distribution has the single nonzero positive-weight
coordinate $A_Q^*=Q^2-1$.  Hence
\[
 \delta^{\mathrm{QMDS}}_{n,q,3}=\delta_{n,Q}.
\]
Here $\widehat h_2(2)=\mu_2/Q^2$.  With $s=1$ and $k=Q$ in
\eqref{eq:kappa-closed}, one has $\kappa_{n,Q}=\mu_2/2$, hence
$\widehat h_2(2)=2c_{n,Q}$.  Moreover
$\widehat h_2(1)=\mu_1/Q+\mu_2/Q\geq Q\widehat h_2(2)/2$, so with
$c_1(1)=Q$ and $c_1(2)=2$ the maximum in
\eqref{eq:QMDS-exact-Gram-constant} is attained at distance two and
$C_1^{\mathrm{Gram}}=1/c_{n,Q}$, the same radius-one Gram factor as
Theorem~\ref{thm:main}; at
$q=2$ this is $1/3$, shown sharp in Section~\ref{sec:five}.
\end{remark}

\section{Equality and nondegenerate perfect codes}
\label{sec:equality}

The discrepancy bound has a concrete equality case: the radius-one error
images tile the physical space.  The following short argument identifies
this tiling with the nondegenerate Knill--Laflamme condition.

Fix a rank-$K$ orthogonal projector $P$ at the parameters
\eqref{eq:admissible}.  Then $KV=d$, so $\alpha_1=1$.
For $a\in\cB_1$ let $P_a\coloneqq U_aPU_a^\dagger$.

\begin{theorem}[Exact equality characterization]
\label{thm:equality}
Under \eqref{eq:admissible}, the following are equivalent:
\begin{enumerate}
\item $\cD_Q(P)=\delta_{n,Q}$;
\item $\Phi_1(P)=0$;
\item $\sum_{a\in\cB_1}P_a=I$;
\item the subspaces $U_a\operatorname{ran}P$, $a\in\cB_1$, are pairwise
orthogonal and decompose $\Hphys$;
\item for all $a,b\in\cB_1$,
\begin{equation}
 P U_a^\dagger U_bP=\mathbf 1_{a=b}P.
 \label{eq:nondegenerate-KL}
\end{equation}
\end{enumerate}
Thus equality holds exactly for nondegenerate perfect one-error quantum codes,
that is, pure codes of distance at least $3$
(Definition~\ref{def:pure-enumerator-distance}) whose rank $K=d/V$ meets the
quantum Hamming bound with equality.
\end{theorem}

\begin{proof}
The first two statements are equivalent by Theorem~\ref{thm:main}, and the
second and third by definition.  Suppose the third holds.  For fixed $a$,
\[
 0=\Tr\bigl(P_a(I-P_a)\bigr)
  =\sum_{b\ne a}\Tr(P_aP_b).
\]
Every summand is nonnegative, because
$\Tr(P_aP_b)=\|P_aP_b\|_{\HS}^2$.  Hence $P_aP_b=0$ for $a\ne b$.  Conversely,
$V$ mutually orthogonal rank-$K$ projections have total rank $VK=d$, so their
sum is $I$.  Conjugating the orthogonality relation by the Weyl operators gives
\eqref{eq:nondegenerate-KL}, the nondegenerate Knill--Laflamme condition for
the error set $\cB_1$ \cite{KnillLaflamme1997Theory}.  Every label of
weight $1$ or $2$ is $b-a$ for some $a\neq b$ in $\cB_1$, and
$U_a^\dagger U_b$ is a phase times $U_{b-a}$, so
\eqref{eq:nondegenerate-KL} says exactly that $PU_cP=0$ for
$1\leq\wt(c)\leq2$; under \eqref{eq:admissible}, $K=d/V$.
\end{proof}

\begin{remark}[Degeneracy]
The zero set of $\Phi_1$ consists of orthogonal tilings by the individual Pauli error
images.  By \eqref{eq:nondegenerate-KL}, a degenerate one-error-correcting
code of this rank, if one exists, has $PU_a^\dagger U_bP=c_{ab}P$ with some
$c_{ab}\ne0$ for $a\ne b$, hence $\Phi_1>0$, and does not attain the
benchmark; for binary stabilizer codes,
\cite[Corollary~2]{NemecTansuwannont2023HammingLike} excludes such codes
(Appendix~\ref{app:stabilizer}).  Away from $KV=d$, one must instead distinguish the covering defect
$\Phi_1$ from the packing defect $\Psi_1$ in \eqref{eq:packing-defect}.
\end{remark}

\section{Gram rigidity and transpose/Petz recovery}
\label{sec:rigidity}

We use the unhalved diamond norm, as in Lemma~\ref{lem:polar-linear-diamond}.
For a fixed error channel, QEC Gram matrices and transpose/Petz recovery are
standard tools in approximate quantum error correction
\cite{BenyOreshkov2010Approximate,ZhengEtAl2024NearOptimal,
NgMandayam2010Transpose,ElimelechAlbertBarg2026ErrorSet,
ElimelechAlbertBarg2026UniversalRecovery}.  Square-root
measurements and least-squares constructions provide the corresponding
finite-dimensional polar picture \cite{EldarForney2001SquareRoot}.  This
section converts Theorem~\ref{thm:main} into explicit control of the Gram
matrix and the transpose/Petz recovery for the uniform radius-one Weyl error
set.  Here the rank is fixed by arithmetic one-error
Hamming equality, so the error-image map is square.  This gives an exact
identity between the covering defect $\Phi_1(P)$ and the squared Gram defect
$\|G-I_{VK}\|_{\HS}^2$, and the estimates apply to
every projector of the prescribed rank.

Let $J:\mathbb C^K\to\Hphys$ be an isometry with $JJ^\dagger=P$, and let
$\operatorname{Enc}(\rho)\coloneqq J\rho J^\dagger$ denote the encoding channel.  At
Hamming equality, index the standard basis of $\mathbb C^V$ by
$\cB_1$, and define the square error-image map
\begin{equation}
 \begin{aligned}
 Y&:\mathbb C^V\otimes\mathbb C^K\longrightarrow\Hphys,\\
 Y(|a\rangle\otimes x)&\coloneqq U_aJx,\qquad a\in\cB_1.
 \end{aligned}
 \label{eq:Y-def}
\end{equation}
Its block Gram matrix is
\begin{equation}
G\coloneqq Y^\dagger Y=
 \bigl[J^\dagger U_a^\dagger U_bJ\bigr]_{a,b\in\cB_1}.
 \label{eq:Gram}
\end{equation}
For the uniform error channel with Kraus operators
$E_a\coloneqq V^{-1/2}U_aJ$, the standard QEC matrix is
\begin{equation}
 M_{\mathrm{QEC}}=\frac1V G.
 \label{eq:QEC-matrix}
\end{equation}
Here $G$ uses the tensor-factor order error label $\otimes$ logical;
exchanging the factors gives the convention used in
\cite{ElimelechAlbertBarg2026ErrorSet,ZhengEtAl2024NearOptimal}.

\begin{theorem}[Operator rigidity]
\label{thm:operator-rigidity}
Let $\lambda_1,\ldots,\lambda_d$ be the eigenvalues of $G$, and put
\begin{equation}
 \varepsilon\coloneqq\sqrt{\frac{d-1}{d}\Phi_1(P)}.
 \label{eq:epsilon}
\end{equation}
Then
\begin{gather}
 \begin{aligned}
 \Phi_1(P)
 &=\|YY^\dagger-I_{\Hphys}\|_{\HS}^2\\
 &=\|G-I_{VK}\|_{\HS}^2\\
 &=\sum_{i=1}^d(\lambda_i-1)^2,
 \end{aligned}
 \label{eq:Gram-Phi}\\
 |\lambda_i-1|\leq\varepsilon\quad(1\leq i\leq d),
 \label{eq:eigen-bound}\\
 \begin{aligned}
 \max\{0,1-\varepsilon\}I_{VK}
 &\leq G\leq(1+\varepsilon)I_{VK},\\
 \max\{0,1-\varepsilon\}I_{\Hphys}
 &\leq YY^\dagger\leq(1+\varepsilon)I_{\Hphys}.
 \end{aligned}
 \label{eq:frame-bounds}
\end{gather}
If $\varepsilon<1$, then
\begin{equation}
 \operatorname{cond}_2(Y)\leq
 \sqrt{\frac{1+\varepsilon}{1-\varepsilon}}.
 \label{eq:condition}
\end{equation}

Let $\widetilde Y:\mathbb C^V\otimes\mathbb C^K\to\Hphys$ be a unitary polar factor of
$Y$, and let $\mathsf U(d)$ denote the set of unitary maps between these two
$d$-dimensional spaces.  By the Fan--Hoffman nearest-unitary theorem
\cite{FanHoffman1955Metric}, it is a nearest unitary in Hilbert--Schmidt norm, and
\begin{align}
 \operatorname{dist}_{\HS}(Y,\mathsf U(d))^2
 &=\|Y-\widetilde Y\|_{\HS}^2
 =\sum_i(\sqrt{\lambda_i}-1)^2\label{eq:polar-exact}\\
 &=2d-2\Tr\sqrt G.\nonumber
\end{align}
Moreover,
\begin{equation}
 \begin{aligned}
 \frac{\Phi_1(P)}{(1+\sqrt{1+\varepsilon})^2}
 &\leq\|Y-\widetilde Y\|_{\HS}^2\\
 &\leq
 \frac{\Phi_1(P)}
 {(1+\sqrt{\max\{0,1-\varepsilon\}})^2}\\
 &\leq\Phi_1(P).
 \end{aligned}
 \label{eq:polar-bounds}
\end{equation}
\end{theorem}

\begin{proof}
By \eqref{eq:Y-def}, $YY^\dagger=T_1(P)$; the nonzero spectra of
$YY^\dagger$ and $Y^\dagger Y$ coincide, including multiplicity because both
are $d\times d$.  This proves \eqref{eq:Gram-Phi}.  Also
$\sum_i\lambda_i=\Tr G=KV=d$, so $x_i\coloneqq\lambda_i-1$ satisfy $\sum_i x_i=0$.
For each $i$,
\[
 x_i^2=\left(\sum_{j\ne i}x_j\right)^2
 \leq(d-1)\sum_{j\ne i}x_j^2
 =(d-1)(\Phi_1-x_i^2),
\]
which gives \eqref{eq:eigen-bound}.  The frame bounds and condition-number
bound follow.

The nearest-unitary formula is the standard singular-value form of the polar
decomposition, with an arbitrary unitary completion if $Y$ is singular.  For
each $\lambda\geq0$,
\[
 (\sqrt\lambda-1)^2=
 \frac{(\lambda-1)^2}{(\sqrt\lambda+1)^2}.
\]
Apply \eqref{eq:eigen-bound} to the denominator and sum to obtain
\eqref{eq:polar-bounds}.
\end{proof}

Lemma~\ref{lem:polar-decoder} applies with $\cE=\cB_1$ and $v=V$.  Since
$KV=d$, take the flag embedding to be the identity and
$D\coloneqq\widetilde Y^\dagger:\Hphys\to\mathbb C^V\otimes\mathbb C^K$.
By Lemma~\ref{lem:polar-decoder}(2), on the support of $T_1(P)=YY^\dagger$
its block rows satisfy
\begin{equation}
 (\langle a|\otimes I_K)D
 =J^\dagger U_a^\dagger T_1(P)^{+1/2}.
 \label{eq:Petz-blocks}
\end{equation}
Thus tracing out the error-label register gives precisely the transpose/Petz
recovery for the uniform radius-one error channel; the unitary completion of
$\widetilde Y$ specifies an arbitrary channel completion off that support
\cite{BenyOreshkov2010Approximate,ZhengEtAl2024NearOptimal,
ElimelechAlbertBarg2026ErrorSet}.

For channels with Kraus operators in
$\operatorname{span}\{U_a:a\in\cB_1\}$, coherent or stochastic,
\cite[Sections~3--4]{ElimelechAlbertBarg2026UniversalRecovery} establishes
universal recovery guarantees.  The following corollary bounds the diamond
error and infidelity of every such channel under the fixed polar recovery.
For the Pauli mixtures $\cN_{1,p}$, it adds the factor $p_\ast$, for every
probability vector and with the same recovery; in the
entanglement-fidelity bound the factor $p_\ast$ can also be read from
\cite[Eqs.~(26)--(27)]{ElimelechAlbertBarg2026UniversalRecovery}.

\begin{corollary}[Recovery under nonuniform radius-one Pauli noise]
\label{cor:recovery}
Under \eqref{eq:admissible}, let $p=(p_a)_{a\in\cB_1}$ be a probability distribution and put
$p_\ast\coloneqq\max_a p_a$.  Define
\begin{equation}
 \cN_{1,p}(\sigma)\coloneqq\sum_{a\in\cB_1}p_aU_a\sigma U_a^\dagger,
 \qquad
 \cR(\sigma)\coloneqq\Tr_{\mathbb C^V}(D\sigma D^\dagger).
 \label{eq:noise-recovery}
\end{equation}
Put $g\coloneqq\cD_Q(P)-\delta_{n,Q}$ and $c\coloneqq c_{n,Q}$, and let
$\Mis_p(D)$ denote the decoder mismatch of Lemma~\ref{lem:polar-decoder}(1)
for $\cE=\cB_1$ and $\iota=\mathrm{id}$.
Every channel $\cN$ whose Kraus operators lie in
$\operatorname{span}\{U_a:a\in\cB_1\}$ satisfies
\[
 \|\cR\circ\cN\circ\operatorname{Enc}-\id_K\|_\diamond
 \leq2\min\left\{1,\frac gc\right\},
 \qquad
 1-\Fent(I_K/K,\cR\circ\cN\circ\operatorname{Enc})\leq\frac g{cK}.
\]
Moreover,
\begin{equation}
 \begin{aligned}
 \|\cR\circ\cN_{1,p}\circ\operatorname{Enc}-\id_K\|_\diamond
 &\leq2\min\{1,\Mis_p(D)\}\\
 &\leq2\min\left\{1,\frac{p_\ast g}{c}\right\}.
 \end{aligned}
 \label{eq:diamond-bound}
\end{equation}
For the maximally mixed logical input,
\begin{equation}
 1-\Fent(I_K/K,\cR\circ\cN_{1,p}\circ\operatorname{Enc})
 \leq\frac{\Mis_p(D)}K
 \leq\frac{p_\ast g}{cK}.
 \label{eq:weighted-Fe-bound}
\end{equation}

For the uniform distribution $p_a=1/V$, the recovery $\cR$ is the
transpose/Petz recovery described in \eqref{eq:Petz-blocks}.  In this case
\begin{equation}
 \Fent(I_K/K,\cR\circ\cN_{1,\mathrm{unif}}\circ\operatorname{Enc})
 =\frac1{K^2V}
 \left\|\Tr_{\mathbb C^K}\sqrt G\right\|_{\HS}^2.
 \label{eq:Fe-exact}
\end{equation}
The exact identity \eqref{eq:Fe-exact} is the uniform-channel fidelity formula
of \cite{ZhengEtAl2024NearOptimal}.  Moreover,
\eqref{eq:weighted-Fe-bound} specializes to
\begin{equation}
 \begin{aligned}
 1-\Fent(I_K/K,\cR\circ\cN_{1,\mathrm{unif}}\circ\operatorname{Enc})
 &\leq\frac{\|Y-\widetilde Y\|_{\HS}^2}{KV}\\
 &=\frac{\|Y-\widetilde Y\|_{\HS}^2}{d}\\
 &\leq\frac{g}{cd}.
 \end{aligned}
 \label{eq:Fe-bound}
\end{equation}
\end{corollary}

\begin{proof}
The span-channel bounds follow from Lemma~\ref{lem:polar-decoder}(3),
\eqref{eq:polar-exact}, \eqref{eq:polar-bounds} and Theorem~\ref{thm:main}.
Lemma~\ref{lem:polar-decoder}(1), with the identity flag embedding, gives the
first inequalities in \eqref{eq:diamond-bound} and
\eqref{eq:weighted-Fe-bound}, and
$\Mis_p(D)\leq p_\ast\|Y-\widetilde Y\|_{\HS}^2$.
Theorems~\ref{thm:operator-rigidity} and~\ref{thm:main} give
$\|Y-\widetilde Y\|_{\HS}^2\leq\Phi_1(P)\leq g/c$, which proves the second
inequalities.  For uniform $p$, Lemma~\ref{lem:polar-decoder}(2) with $v=V$
and $M_{\mathrm{QEC}}=G/V$ gives \eqref{eq:Fe-exact}.  Finally,
Lemma~\ref{lem:polar-decoder}(1) with $p_\ast=1/V$ gives
$\Mis_p(D)\leq\|Y-\widetilde Y\|_{\HS}^2/V$, and
$KV=d$ gives \eqref{eq:Fe-bound}.
\end{proof}

Appendix~\ref{app:noise} extends these estimates to arbitrary stochastic
Pauli channels through their outside-ball probability.

For a prime-dimensional stabilizer projector, the covering defect
$\Phi_1(P)$ is $K$ times an integer syndrome-collision count.  The resulting alternative is
\(\Phi_1(P)=0\) or
\(\Phi_1(P)\geq K\max\{8,2(q-1)\}\), and
Theorem~\ref{thm:main} converts it into a discrepancy gap.
Appendix~\ref{app:stabilizer} develops this arithmetic refinement, including
a five-qubit example showing that the binary bound is sharp.  The next
section returns to arbitrary projectors and determines the largest admissible
coefficient in \eqref{eq:main-stability} for five qubits.

\section{The five-qubit perfect code}
\label{sec:five}

The five-qubit perfect code lets us test both the sharpness of the
discrepancy bound and its consequences for recovery.  A Pauli rotation
attains the coefficient three with positive covering defect, and the same family
admits an exact calculation of the recovered logical channel.

Let $P_*$ denote the projector of the $[[5,1,3]]_2$ perfect code
\cite{LaflammeEtAl1996Perfect}.  Its parameters are
\[
 \begin{gathered}
 q=2,\quad Q=4,\quad n=5,\quad d=32,\\
 K=2,\quad V=16,\quad k=4.
 \end{gathered}
\]
We write $\cD_4$ for $\cD_Q$ at $Q=4$ and $W_{5,4}$ for $W_{n,Q}$ at these values.  One choice of stabilizer generators is
\[
 \mathtt{XZZXI},\qquad \mathtt{IXZZX},\qquad \mathtt{XIXZZ},\qquad \mathtt{ZXIXZ}.
\]
Direct stabilizer enumeration gives
\begin{equation}
 \begin{aligned}
 A_{P_*}(z)&=1+15z^4,\\
 B_{P_*}(z)&
      =1+30z^3+15z^4+18z^5.
 \end{aligned}
 \label{eq:five-enumerators}
\end{equation}
The sixteen radius-one Pauli errors have pairwise orthogonal code images
that exhaust the physical space, so $\Phi_1=0$.

For $w=4$, the Krawtchouk values and their cumulative sums are
\[
 \begin{aligned}
 (K_0^{(5,4)},\ldots,K_5^{(5,4)})&=(1,-1,-6,14,-11,3),\\
 (L_0,\ldots,L_5)&=(1,0,-6,8,-3,0).
 \end{aligned}
\]
Thus $W_{5,4}(4)=1+0+36+64+9=110$, and
\begin{equation}
 \cD_4(P_*)=\delta_{5,4}
 =\frac{K^2}{d}(V-1)W_{5,4}(4)
 =\frac{825}{4}.
 \label{eq:five-benchmark}
\end{equation}
For the remainder of this section, write
\[
 g(P)\coloneqq\cD_4(P)-\frac{825}{4}
\]
for the discrepancy excess above this benchmark.
The following values determine the one-error coefficient at these parameters:
\[
 \begin{gathered}
 (W(1),W(2),W(3),W(4),W(5))\\
 =(21286,1118,246,110,70),
 \end{gathered}
\]
so $a_{5,4}=136$, $\kappa_{5,4}=48$, and
$c_{5,4}=3$.  Theorem~\ref{thm:main} with $c_{5,4}=3$ therefore gives
\[
 \cD_4(P)-\frac{825}{4}\geq3\Phi_1(P)
\]
for every rank-two projector on five qubits, with equality at the perfect code
when both sides vanish.

The coefficient in this inequality is attained away from the zero set.

\begin{proposition}[Sharpness of the five-qubit coefficient]
\label{prop:five-sharp}
Let $H\coloneqq\mathtt{IIIXX}$ and, for $\theta\in\mathbb R$, define
\begin{equation}
 P_\theta\coloneqq e^{-i\theta H}P_*e^{i\theta H},
 \qquad u\coloneqq\sin^2(2\theta).
 \label{eq:five-rotation}
\end{equation}
Then
\begin{equation}
 \Phi_1(P_\theta)=16u,
 \qquad
 \cD_4(P_\theta)-\frac{825}{4}=48u=3\Phi_1(P_\theta).
 \label{eq:five-sharp}
\end{equation}
Consequently, $3$ is the largest coefficient that can replace $c_{5,4}$ in
\eqref{eq:main-stability} for all rank-two projectors on five qubits.
Since $KV_1=d$, \eqref{eq:packing-defect} identifies $\Phi_1$ with the squared
radius-one Gram defect, so the Gram constant $C_1^{\mathrm{Gram}}=1/c_{5,4}=1/3$ of
Remark~\ref{rem:QMDS-one-error-overlap} is sharp as well.
\end{proposition}

\begin{proof}
Let $S_*$ be the signed stabilizer group generated by the four displayed
Pauli strings.  Its fifteen nonidentity elements have weight four.  The
operator $H$ commutes with seven of them and anticommutes with eight.  Up to
phase, the products $HS$ for the eight anticommuting elements are
\[
 \mathtt{IXZYI},\ \mathtt{XYIZI},\ \mathtt{IYXIZ},\ \mathtt{ZXIIY}
\]
of weight three and
\[
 \mathtt{XXYXZ},\ \mathtt{YYZXY},\ \mathtt{YXXZX},\ \mathtt{ZYYYX}
\]
of weight five.  These strings are distinct and lie outside $S_*$.  If
$S\in S_*$ anticommutes with $H$, then
\[
 e^{-i\theta H}Se^{i\theta H}
 =\cos(2\theta)S-i\sin(2\theta)HS,
\]
whereas a commuting $S$ is unchanged.  The Pauli expansion of the rotated
projector therefore gives
\begin{equation}
 A_{P_\theta}(z)=1+4uz^3+(15-8u)z^4+4uz^5.
 \label{eq:five-rotated-enumerator}
\end{equation}
Here the squared coefficients add without cross terms because the displayed
Pauli strings are all distinct.

At these parameters $K^2/d=1/8$ and $L_1(w)=16-4w$.  Hence
Theorem~\ref{thm:spectral} and \eqref{eq:five-rotated-enumerator} give
\[
 \Phi_1(P_\theta)=\frac18\bigl(4u\cdot16+4u\cdot16\bigr)=16u.
\]
Using the values of $W$ displayed above gives, similarly,
\[
 \cD_4(P_\theta)-\cD_4(P_*)
 =\frac18\bigl(4u\cdot246-8u\cdot110+4u\cdot70\bigr)
 =48u.
\]
For any $\theta$ with $u>0$, equality holds with positive covering defect, proving
optimality of the coefficient.
\end{proof}

The rotation in Proposition~\ref{prop:five-sharp} is also an exactly solvable
approximate-QEC example.  Up to local-Clifford and qubit-ordering conventions,
it is a two-qubit Pauli rotation on one edge, of the type used in the
five-qubit variational graphical architecture of
\cite{ShaoEtAl2026VGQEC}.  The next result evaluates the fixed polar recovery
analytically, with the code and recovery held fixed across all probability
distributions.

Use Hermitian Pauli representatives and number the qubits from left to right.
Choose an encoding $J_*:\mathbb C^2\to(\mathbb C^2)^{\otimes5}$ for $P_*$
such that the physical operators
\[
 \overline X\coloneqq\mathtt{XXXXX},\qquad
 \overline Z\coloneqq\mathtt{ZZZZZ},\qquad
 \overline Y\coloneqq i\overline X\overline Z=\mathtt{YYYYY}
\]
induce the logical Pauli matrices $X_L,Z_L,Y_L$, respectively.  With
$H\coloneqq\mathtt{IIIXX}$, set
\[
 J_\theta\coloneqq e^{-i\theta H}J_*,
 \qquad
 \operatorname{Enc}_\theta(\rho)\coloneqq J_\theta\rho J_\theta^\dagger,
 \qquad
 u\coloneqq\sin^2(2\theta),
\]
and let $\cR_\theta$ denote a polar recovery of
Corollary~\ref{cor:recovery} for $J_\theta$.  Write
\[
 \Lambda_{\theta,p}
 \coloneqq\cR_\theta\circ\cN_{1,p}\circ\operatorname{Enc}_\theta
\]
for a distribution $p=(p_a)_{a\in\cB_1}$, and define
\begin{align}
 \mathcal S_X&\coloneqq\{\mathtt{ZIIII},\mathtt{IIIYI},
          \mathtt{IIZII},\mathtt{IIIIY}\},\notag\\
 \mathcal S_Y&\coloneqq\{\mathtt{XIIII},\mathtt{IIIIZ},
          \mathtt{IIXII},\mathtt{IIIZI}\},\notag\\
 \eta_X&\coloneqq\sum_{a\in\mathcal S_X}p_a,\qquad
 \eta_Y\coloneqq\sum_{a\in\mathcal S_Y}p_a,\qquad
 \eta\coloneqq\eta_X+\eta_Y,\notag\\
 \xi_\theta&\coloneqq\frac{1-\sqrt{1-u}}2
 =\frac{1-\sqrt{1-g(P_\theta)/48}}2.
 \label{eq:five-xi}
\end{align}
For a unitary $L$, write
$\operatorname{Ad}_L(\rho)\coloneqq L\rho L^\dagger$.

\begin{proposition}[Exact logical channel for the five-qubit polar recovery]
\label{prop:five-exact-recovery}
For every $\theta$, every such recovery $\cR_\theta$ and every probability
distribution $p$ on $\cB_1$, the exact channel identity holds:
\begin{equation}
 \Lambda_{\theta,p}
 =(1-\xi_\theta\eta)\id_2{}
 +\xi_\theta\eta_X\operatorname{Ad}_{X_L}{}
 +\xi_\theta\eta_Y\operatorname{Ad}_{Y_L}.
 \label{eq:five-exact-logical-channel}
\end{equation}
Consequently,
\begin{align}
 \|\Lambda_{\theta,p}-\id_2\|_\diamond
 &=2\xi_\theta\eta,
 \label{eq:five-exact-diamond}\\
 1-\Fent(I_2/2,\Lambda_{\theta,p})
 &=\xi_\theta\eta.
 \label{eq:five-exact-fidelity}
\end{align}
For the uniform distribution on $\cB_1$, one has $\eta=1/2$, and hence
\begin{align}
 \|\Lambda_{\theta,\mathrm{unif}}-\id_2\|_\diamond
 &=\frac{1-\sqrt{1-g(P_\theta)/48}}2\notag\\
 &=\frac{g(P_\theta)}{192}+O(g(P_\theta)^2),\notag\\
 1-\Fent(I_2/2,\Lambda_{\theta,\mathrm{unif}})
 &=\frac{1-\sqrt{1-g(P_\theta)/48}}4\notag\\
 &=\frac{g(P_\theta)}{384}+O(g(P_\theta)^2).
 \label{eq:five-uniform-exact}
\end{align}
Moreover,
$A_2(P_\theta)=0$ and $B_2(P_\theta)=4u=g(P_\theta)/12$.
Consequently, every member with $g(P_\theta)>0$ fails exact correction of the
radius-one Pauli error set.
\end{proposition}

\begin{proof}
Let $Y_\theta$ be the radius-one error-image map associated with $J_\theta$
and put $G_\theta\coloneqq Y_\theta^\dagger Y_\theta$.
We compute this Gram matrix, take its positive square root, and read the
logical channel from the resulting orthogonal flag branches.  If
$F_{ab}\coloneqq U_a^\dagger U_b$, exact correction by $P_*$ gives
$J_*^\dagger F_{ab}J_*=\mathbf1_{a=b}I_2$.  When $F_{ab}$ anticommutes with
$H$,
\[
 e^{i\theta H}F_{ab}e^{-i\theta H}
 =\cos(2\theta)F_{ab}+i\sin(2\theta)HF_{ab};
\]
when it commutes with $H$, it is unchanged.  Since $J_*^\dagger F_{ab}J_*=0$
for $a\neq b$, the off-diagonal block $(a,b)$ of $G_\theta-I_{32}$ is
$\sin(2\theta)\,J_*^\dagger(iHF_{ab})J_*$ when $F_{ab}$ anticommutes with $H$,
and it vanishes otherwise; it is nonzero exactly when $HF_{ab}$ commutes with
every stabilizer, that is, when the syndromes of $U_a$ and $U_b$ differ by
the syndrome of $H$, which is nonzero because $H\neq\mathtt{IIIII}$ has weight
below the distance $3$.  The sixteen radius-one errors have distinct syndromes
exhausting the syndrome group, so this condition pairs each label with
exactly one other.  Among the labels in $\cB_1$, only $\mathtt{IIIYI}$, $\mathtt{IIIZI}$, $\mathtt{IIIIY}$ and $\mathtt{IIIIZ}$
anticommute with $H$, and each lies in a different one of the four pairs
displayed below; hence $F_{ab}$ anticommutes with $H$ for exactly these four
pairs.  This gives
\begin{equation}
 G_\theta=I_{32}+s_\theta\Gamma,\qquad s_\theta\coloneqq\sin(2\theta),
 \label{eq:five-Gram-affine}
\end{equation}
where the only nonzero upper-triangular $2$-by-$2$ logical blocks of the
Hermitian matrix $\Gamma$ are
\[
 \begin{array}{c|c|c}
 (a,b)&\text{block}&HF_{ab}\\\hline
 (\mathtt{XIIII},\mathtt{IIIIZ})&-Y_L&\mathtt{XIIXY}\\
 (\mathtt{ZIIII},\mathtt{IIIYI})& X_L&\mathtt{ZIIZX}\\
 (\mathtt{IIXII},\mathtt{IIIZI})&-Y_L&\mathtt{IIXYX}\\
 (\mathtt{IIZII},\mathtt{IIIIY})& X_L&\mathtt{IIZXZ}
 \end{array}
\]
with $HF_{ab}$ written up to phase.  Each listed $HF_{ab}$ is a weight-three
normalizer element equal, up to phase and a stabilizer factor, to
$\overline Y$ or $\overline X$; this fixes the logical class in the second
column.  The signs depend on these sign conventions and do not affect
\eqref{eq:five-exact-logical-channel}, where only $\operatorname{Ad}_L$
enters.  The lower-triangular blocks are their adjoints and all omitted blocks vanish.
Thus $\Gamma^2$ is the identity on the sixteen-dimensional space carried by
the eight labels in $\mathcal S_X\cup\mathcal S_Y$, zero on its orthogonal complement, and
$\Gamma^3=\Gamma$.  This also recovers
$\|G_\theta-I_{32}\|_{\HS}^2=16s_\theta^2=\Phi_1(P_\theta)$.

Put
\[
 \alpha_\theta\coloneqq\frac{\sqrt{1+s_\theta}+\sqrt{1-s_\theta}}2,\qquad
 \beta_\theta\coloneqq\frac{\sqrt{1+s_\theta}-\sqrt{1-s_\theta}}2.
\]
Functional calculus in \eqref{eq:five-Gram-affine} gives
\begin{align}
 \sqrt{G_\theta}
 &=I_{32}+\beta_\theta\Gamma
  +(\alpha_\theta-1)\Gamma^2,\notag\\
 \beta_\theta^2&=\xi_\theta,\qquad
 \alpha_\theta^2=1-\xi_\theta.
\end{align}
Writing the square polar decomposition as
$Y_\theta=\widetilde Y_\theta\sqrt{G_\theta}$ and
$D_\theta\coloneqq\widetilde Y_\theta^\dagger$, one has
$D_\theta Y_\theta=\sqrt{G_\theta}$.  This remains true when $|s_\theta|=1$:
any unitary completion of the polar partial isometry agrees on the support of
$\sqrt{G_\theta}$.

For an error label outside $\mathcal S_X\cup\mathcal S_Y$, the decoded isometry is exactly
$|a\rangle\otimes I_2$.  For each of the four paired blocks displayed above,
it has two orthogonal flag branches,
\[
 \alpha_\theta|a\rangle\otimes I_2
 \ \mathbin{\pm}\
 \beta_\theta|a'\rangle\otimes L,
 \qquad L\in\{X_L,Y_L\}.
\]
Tracing the flag therefore gives
$(1-\xi_\theta)\id_2{}+\xi_\theta\operatorname{Ad}_L$ for each label in
$\mathcal S_X\cup\mathcal S_Y$ and $\id_2$ for every other label.  Averaging with $p$ proves
\eqref{eq:five-exact-logical-channel}.  Its difference from $\id_2$ is
\[
 \xi_\theta\eta_X(\operatorname{Ad}_{X_L}-\id_2)
 +\xi_\theta\eta_Y(\operatorname{Ad}_{Y_L}-\id_2).
\]
The triangle inequality gives diamond norm at most $2\xi_\theta\eta$.
Applying the difference to one half of a maximally entangled state gives
the Bell-diagonal eigenvalues
$-\xi_\theta\eta,\xi_\theta\eta_X,\xi_\theta\eta_Y,0$, so the matching
lower bound holds.  The identity coefficient in the logical Pauli channel is
$1-\xi_\theta\eta$, proving \eqref{eq:five-exact-fidelity}.  Uniformity gives
$\eta_X=\eta_Y=1/4$, and Taylor expansion at $g=0$ proves
\eqref{eq:five-uniform-exact}.  Finally, applying
\eqref{eq:quantum-macwilliams} to
\eqref{eq:five-rotated-enumerator} gives $B_2(P_\theta)=4u$, while
$A_2(P_\theta)=0$.  Exact correction of $\cB_1$ would force
$A_2(P_\theta)=B_2(P_\theta)$ by the Knill--Laflamme conditions
\cite{KnillLaflamme1997Theory}, proving the final assertion.
\end{proof}

The general guarantee \eqref{eq:diamond-bound} specializes to
$\min\{2,g/24\}$ for every rank-two five-qubit projector under uniform
ball noise and its polar recovery.  The positive slope $1/192$ in
\eqref{eq:five-uniform-exact} shows that the linear order in $g$ is attained
within this class.

For a concrete point on this family, take
\(\sin(2\theta)=3/5\).  The discrepancy excess is \(432/25\).
Under uniform ball noise, the diamond error is \(1/10\) and the
entanglement infidelity at \(I_2/2\) is \(1/20\).  The general estimates in
Corollary~\ref{cor:recovery} give the respective upper bounds \(18/25\) and
\(9/50\).  At the same excess \(g=432/25\), these bounds apply to every
rank-two projector on five qubits with its corresponding polar recovery.
Corollary~\ref{cor:five-nonexact-benchmark} in Appendix~\ref{app:noise}
also gives the corresponding depolarizing-noise estimate.

\section{Balanced packing beyond perfect codes}
\label{sec:balanced-summary}

Perfect Hamming codes exhaust the physical space with their radius-one
error images.  Gottesman's \([[2^j,2^j-j-2,3]]_2\) family, \(j\geq3\),
instead packs these images orthogonally into part of that space
\cite{Gottesman1996Class,CalderbankEtAl1998GF4}: every syndrome occurs at
most once, so the images are pairwise orthogonal.  We call a stabilizer code
balanced when its radius-one syndrome multiplicities differ by at most one;
here this means multiplicities $0$ and $1$ (balanced packing).  For \(n=2^j\) and
\(K=2^n/(4n)\), the covering defect of its projector \(P_G\) is
\[
 \Phi_1(P_G)=K\frac{(3n+1)(n-1)}{4n}.
\]
Subtracting this value gives the packing defect \(\Psi_1\) from
\eqref{eq:packing-defect}.  Let \(\widehat h_n(2)\) be the coefficient of
\(K_2^{(n,4)}\) in the quadratic polynomial \(h_n\) interpolating \(W_{n,4}\)
at \(3n/4,3n/4+1,n\).  Lemma~\ref{lem:certificate-transfer} with this interpolant, together with
the sign checks of Lemma~\ref{lem:Gottesman-certificate}, yields the
following estimate.

\begin{mainresult}[Preview of Theorem~\ref{thm:Gottesman-stability}]
For \(n=2^j\), \(j\geq3\), and \(K=2^n/(4n)\),
\[
 \cD_4(P)-\cD_4(P_G)\geq\frac{\widehat h_n(2)}{2}\Psi_1(P)
 \qquad\text{for every rank-\(K\) projector \(P\)},
\]
with \(\widehat h_n(2)>0\), explicitly evaluated in \eqref{eq:Gottesman-c2} in
Appendix~\ref{sec:gottesman}.
\end{mainresult}

Equality in the discrepancy minimum is characterized by the trace enumerator
of \(P_G\), \(A_P(z)=1+(4n-4)z^{3n/4}+3z^n\)
(Lemma~\ref{lem:Gottesman-enumerator}); this condition implies orthogonal
packing, \(\Psi_1(P)=0\), which amounts to \(B_1=B_2=0\) (hence
\(A_1=A_2=0\), since \(A_w\leq B_w\)); we do not address whether \(B_1=B_2=0\) alone forces that enumerator for
every rank-\(K\) projector at these parameters.

Theorem~\ref{thm:Gottesman-stability} in Appendix~\ref{sec:gottesman}
gives the exact nonnegative remainder underlying this estimate.
Its certificate interpolates at the two target weights and an adjacent
integer, using the universal-energy construction of Cohn and Zhao
\cite{CohnZhao2014Energy}.  Thus balanced nonperfect packing fits the same
discrepancy-to-overlap analysis; the appendix contains the enumerator
calculation, equality argument, and the first two explicit lengths.

\section{Rank-one projectors and marginal-purity rigidity}
\label{sec:pure}

For rank-one projectors, discrepancy has a natural entanglement
interpretation.  Weyl coefficients determine subsystem purities, allowing
the excess above a formal benchmark to be written as a weighted sum of
marginal-purity defects.  Positivity of the weights identifies absolutely maximally
entangled states as minimizers whenever they exist.

Let $P=P_\psi\coloneqq|\psi\rangle\langle\psi|$ be rank one.  For
$S\subseteq[n]$, write
\[
 \rho_S\coloneqq\Tr_{S^c}P_\psi
\]
and let $I_S$ denote the identity on the $q^{|S|}$-dimensional subsystem.
For $0\leq s\leq n$, set
\begin{equation}
 \Pi_s(P_\psi)\coloneqq\sum_{\substack{S\subseteq[n]\\|S|=s}}\Tr(\rho_S^2).
 \label{eq:marginal-purity-sum}
\end{equation}
The coefficients below come from binomial inversion: $\omega_s$ expresses
$\cD_Q(P_\psi)$ through the marginal purities $\Pi_s(P_\psi)$, and $\gamma_s$
combines the complementary sizes $s$ and $n-s$, whose purities agree for a
pure state.  Define $f(0)\coloneqq0$, $f(w)\coloneqq W_{n,Q}(w)$ for
$1\leq w\leq n$, and
\begin{equation}
 \omega_s\coloneqq\sum_{w=s}^n(-1)^{w-s}\binom{n-s}{w-s}f(w),
 \qquad 0\leq s\leq n.
 \label{eq:omega-def}
\end{equation}
Put $m\coloneqq\lfloor n/2\rfloor$ and, for $1\leq s\leq m$, define
\begin{equation}
 \gamma_s\coloneqq
 \begin{cases}
  q^s\omega_s+q^{n-s}\omega_{n-s},&s<n/2,\\
  q^s\omega_s,&2s=n.
 \end{cases}
 \label{eq:gamma-def}
\end{equation}
Recall that $|\psi\rangle$ is an absolutely maximally entangled
$\operatorname{AME}(n,q)$ state when every reduction to at most $m$ parties
is maximally mixed \cite{HelwigCui2013AME}.  Finally, let
\begin{align}
 C_{n,q}&\coloneqq q^{-n}(\omega_0+q^n\omega_n),
 \label{eq:pure-constant}\\
 \delta^{\mathrm{pure}}_{n,q}
 &\coloneqq C_{n,q}+q^{-n}\sum_{s=1}^m
   \gamma_s\binom{n}{s}q^{-s}.
 \label{eq:pure-benchmark}
\end{align}

\begin{theorem}[Pure-state marginal-purity identity]
\label{thm:pure-marginal}
For every $q\geq2$ and $n\geq2$, the coefficients in
\eqref{eq:omega-def} satisfy $\omega_s>0$ for $1\leq s\leq n$.
Consequently $\gamma_s>0$ for $1\leq s\leq m$, and every rank-one
projector satisfies the exact identity
\begin{equation}
 \cD_Q(P_\psi)-\delta^{\mathrm{pure}}_{n,q}
 =q^{-n}\sum_{s=1}^m\gamma_s
  \sum_{\substack{S\subseteq[n]\\|S|=s}}
  \left\|\rho_S-\frac{I_S}{q^s}\right\|_{\HS}^2.
 \label{eq:pure-stability}
\end{equation}
Thus, if $\operatorname{AME}(n,q)$ states exist, then
$\delta^{\mathrm{pure}}_{n,q}$ is the pure-state minimum of $\cD_Q$, and
the minimizers are precisely the $\operatorname{AME}(n,q)$ states.  If no
such state exists, minimizing $\cD_Q$ is equivalent to minimizing the
strictly positive weighted marginal-purity potential on the right-hand side
of \eqref{eq:pure-stability}.
\end{theorem}

The proof expands each marginal purity in Weyl coordinates, applies
binomial inversion, and pairs complementary subsystems.  For $s\geq1$,
substituting $u=w-s$ in \eqref{eq:omega-def} gives $\omega_s=\mu_{n-s}$ with
$\mu_j$ from \eqref{eq:QMDS-mu}, so the sign needed for the equality
statement is the complete monotonicity of Section~\ref{sec:qmds}.  The
complete derivation is in
Appendix~\ref{app:pure-calculations}.

At even length \(n=2m\geq4\), the rank-one quantum MDS benchmark with
\(\Delta=m+1\) equals \(\delta^{\mathrm{pure}}_{n,q}\);
Remark~\ref{rem:pure-QMDS-concordance} gives the coefficient-level
comparison in Appendix~\ref{app:pure-calculations}.

Identity~\eqref{eq:marginal-enumerator} in Appendix~\ref{app:pure-calculations}
is the unitary-enumerator relation
developed in the established connection between quantum error correction and
multipartite entanglement
\cite{Scott2004Multipartite,HuberEtAl2018AME}.  Facchi, Florio, Parisi, and
Pascazio define maximally multipartite entangled states through minimization
of the average purity of balanced reductions \cite{FacchiEtAl2008MMES}; AME
states and their applications are studied in \cite{HelwigCui2013AME}, and
their existence can be formulated as a pure-state marginal problem
\cite{YuEtAl2021Marginal}.  The potential in
  \eqref{eq:pure-stability} is a different, fixed multiscale weighting: every
  subsystem size up to $m$ enters with the positive coefficient dictated by the
  Pauli-ball spectral weight.  When AME states do not exist, minimizing it is
  a weighted version of the maximally multipartite entangled-state problem,
  and the pure-state minimum of $\cD_Q$ lies strictly above $\delta^{\mathrm{pure}}_{n,q}$.

For four qubits, where an AME state does not exist, the known
balanced-purity optimum
\cite{BandyopadhyayGourSanders2007Duality,GourWallach2010FourQubits}
gives \(\min\cD_4=155/4\), a gap of \(8\) above the formal benchmark.
For five qubits the minimum is \(2265/16\), attained exactly by
\(\operatorname{AME}(5,2)\) states.  Corollaries~\ref{cor:four-pure}
and~\ref{cor:five-rank-one} in Appendix~\ref{app:pure-calculations}
give the full formulas and equality arguments.

\section{Conclusions and open problems}
\label{sec:limitations}

The exact remainders connect a scalar discrepancy excess to the geometry
of error images.  At formal quantum MDS and one-error Hamming parameters,
they characterize benchmark attainment and quantify the failure of
orthogonality.  When $KV_e\leq d$, this control gives bounds linear in the
excess for a fixed polar recovery under noise with Kraus operators in the
span of $\cB_e$, with $2e<\Delta$ ($e=1$ in the Hamming case).  For
stochastic Pauli noise they are proportional to the largest probability,
including that of the identity, so they do not vanish with the noise.  The five-qubit family attains this linear order and gives
both a sharp Gram constant and an explicit recovered logical channel.
The Gottesman and pure-state results extend the same
discrepancy $\cD_Q$ to balanced nonperfect packing and multipartite
entanglement.

The remaining questions concern broader parameter regimes, the geometry of
minimizers, and the information retained by discrepancy.

\begin{enumerate}
\item \emph{Other ranks and higher radii.}  Our Gram bounds require the Singleton rank with
$2e<\Delta$, or $e=1$ at one-error Hamming or Gottesman parameters;
Hamming equality for $e\geq2$ and other ranks remain open.  A further problem is to construct sharper linear-programming certificates
for the spectral weight of Section~\ref{sec:enumerators} that are adapted to
higher radii and to the Shor--Laflamme constraints.  Such certificates could extend the bounds to
non-Singleton ranks and other equality regimes.

\item \emph{Balanced nonperfect packing.}
For a rank-$K$ stabilizer projector, the number of syndrome classes is
$M=d/K$, so the mean radius-$r$ multiplicity is $V_r/M=KV_r/d$.
Beyond the Gottesman family, which stabilizer families have all radius-$r$
syndrome multiplicities in $\{\lfloor V_r/M\rfloor,\lceil V_r/M\rceil\}$ for a
given radius $r$, and when does this single-radius balance minimize
$\cD_Q$?  Even for Gottesman parameters, we do not determine whether every
orthogonal packing of this rank is a minimizer, or classify all projectors with
enumerator \eqref{eq:Gottesman-equality-enumerator} in
Appendix~\ref{sec:gottesman}.

\item \emph{Pure states without absolute maximal entanglement.}
Beyond the four-qubit case treated in Section~\ref{sec:pure}, the weighted marginal-purity problem
\eqref{eq:pure-stability} asks for the pure-state minimum of $\cD_Q$, which
lies strictly above $\delta^{\mathrm{pure}}_{n,q}$, the minimizing states,
and explicit bounds on the distance to that set at other lengths and local
dimensions.

\item \emph{Degenerate codes.}  The radius-one covering defect $\Phi_1$
vanishes on nondegenerate perfect codes.  Can a defect on a suitable quotient of the
error algebra capture degenerate correctability?

\item \emph{Beyond a fixed Pauli error span.}
Appendix~\ref{app:noise} treats arbitrary stochastic Pauli distributions by
their probability outside the ball.  When $KV_e\leq d$, for coherent Kraus combinations
within the span of the ball, the entanglement infidelity is linear in the excess by
\cite[Theorem~7]{ElimelechAlbertBarg2026UniversalRecovery}, and so is the
diamond error by Corollaries~\ref{cor:QMDS-span-infidelity}
and~\ref{cor:recovery}.  Quantitative tail estimates for coherent Kraus components outside the
selected low-weight error span are not developed here.

\item \emph{Sharper recovery parameters.}
The recovery bounds replace operator norms by Hilbert--Schmidt norms, which
loses information.
Can the spectrum of $G_e$, or an operator-valued discrepancy norm, avoid
this loss and give diamond bounds linear in $g(P)$ for all channels in the
error span beyond the parameter sets treated here?  The exact uniform-channel Petz formula
\eqref{eq:Fe-exact} retains the full Gram matrix, while the
stochastic-channel bounds use $p_\ast$ and a Hilbert--Schmidt distance.
A further question is whether $p_\ast$ can be replaced by a quantity that
retains more information about the error probabilities.

\item \emph{Existence and nonattainment.}  Stabilizer arithmetic gives the
explicit gap \eqref{eq:stabilizer-gap} in Appendix~\ref{app:stabilizer} at
one-error Hamming parameters, and Corollary~\ref{cor:QMDS-stabilizer-gap} gives
one at quantum MDS parameters, both for prime $q$; for
arbitrary projectors,
compactness gives $\min g>0$ when the benchmark is unattained.  Determining
an explicit gap, and extending the analysis to nearby ranks and subprojectors
of a known minimizer, remain open.  The complement identity
\eqref{eq:complement-symmetry} transfers rank-$K$ answers to rank $d-K$.

\item \emph{Estimation.}  For qubits, two-copy Bell sampling provides direct
shadow-enumerator measurements \cite{MillerEtAl2026ShadowMeasurement}.  The
sampling and classical processing costs of estimating the all-radii
discrepancy to prescribed additive accuracy remain to be quantified.
\end{enumerate}

Selected finite diagnostic checks are archived in
\cite{Zabokritskiy2026PauliDiagnostics}; the proofs do not depend on these
checks.

\appendix
\setcounter{equation}{0}
\renewcommand{\theequation}{S\arabic{equation}}
\renewcommand{\theHequation}{S.\arabic{equation}}

\section{Enumerator calculations}
\label{app:enumerators}

Equations in the appendices are numbered (S1), (S2), \ldots.
The formal quantum MDS enumerator $A^*$ enters the exact remainder
through its binomial moments.  We establish those moments and match the
normalization to the known quantum MDS enumerator.

\begin{proof}[Proof of Lemma~\ref{lem:QMDS-formal-moments}]
Use the parameters in \eqref{eq:QMDS-parameters} and the formal
coefficients in \eqref{eq:QMDS-Astar}.
Put $z\coloneqq n-w$.  For $0\leq z\leq t-1$, writing
$Q^{t-z-h}-1=(Q-1)\sum_{i=0}^{t-z-h-1}Q^i$ and using the iterated Pascal
identity $\sum_{h=0}^{u}(-1)^h\binom{n-z}{h}=(-1)^u\binom{n-z-1}{u}$ rewrites
\eqref{eq:QMDS-Astar} as
\[
 A_{n-z}^*
 =\binom nz\sum_{h=0}^{t-z-1}(-1)^h\binom{n-z}{h}
       \bigl(Q^{t-z-h}-1\bigr).
\]
This is the binomial inverse (the pair
$\sum_{z\geq j}\binom zjx_z=y_j\Leftrightarrow
x_z=\sum_{j\geq z}(-1)^{j-z}\binom jzy_j$, together with
$\binom jz\binom nj=\binom nz\binom{n-z}{j-z}$) of
\[
 \sum_{z=j}^{t-1}\binom zjA_{n-z}^*
 =\binom nj\bigl(Q^{t-j}-1\bigr)
 \qquad(0\leq j<t).
\]
Adding the contribution $\binom njA_0^*=\binom nj$ proves
\eqref{eq:QMDS-formal-moments}; the case $j=t$ is immediate because the
positive-weight support has $z<t$.  Taking $j=0$ gives the mass in
\eqref{eq:QMDS-formal-transform}.  The binomial--Krawtchouk identity \eqref{eq:QMDS-binomial-Krawtchouk},
applied successively for $j=1,\ldots,t$, then forces the first $t$
nonconstant Krawtchouk transform coordinates to vanish.

For the final assertion, let $\widehat A_w,\widehat B_w$ be the unnormalized
enumerators of Huber and Grassl.  Direct comparison of definitions gives
$\widehat A_w=K^2A_w$ and $\widehat B_w=KB_w$.  In \cite{HuberGrassl2020QMDS},
$D$ is the local dimension, $K=D^k$, and $\alpha$ is the index in their
Eq.~(31).  Substituting
$D=q$, $k=\ell=n-2t$, and $\alpha=n-t=d_*-1$ in their Eq.~(31) yields
\eqref{eq:QMDS-Astar}; their purity result gives the stated $B$-zeros
\cite[Theorems~5 and~9 and Eq.~(31)]{HuberGrassl2020QMDS}.
\end{proof}

\section{The one-error spectral coefficient}
\label{app:coefficient}

The one-error stability theorem uses a scalar lower bound for the
all-radii spectral weight.  We reproduce the coefficient calculation
here so that the constant in Proposition~\ref{prop:kappa} can be
verified within this article.

\begin{proof}[Proof of Proposition~\ref{prop:kappa}]
With $n=Qs+1$ and $k=1+s(Q-1)$ as in Section~\ref{sec:stability},
write $W(w)\coloneqq W_{n,Q}(w)$ and $a\coloneqq a_{n,Q}$.
The proof compares $W$ with its quadratic interpolant at $k-1,k,n$, which
gives \eqref{eq:kappa-closed} and \eqref{eq:minorant}, and then bounds the
endpoint quotient by the curvature bound \eqref{eq:curvature}, which gives
\eqref{eq:kappa-bound}.

Since $s\geq1$, the nodes $k-1<k<n$ are distinct, and
\[
 H(w)\coloneqq W(k)-a(w-k)+\eta^*(w-k)(w-k+1),
 \qquad
 \eta^*\coloneqq\frac{W(n)-W(k)+as}{s(s+1)},
\]
is the quadratic interpolating $W$ at them.  By
\eqref{eq:W-complete-monotone}, $W(w)$ is the average over $\varphi$ of
$f_\varphi(w)\coloneqq\alpha_Q(\varphi)^{n-1}\rho(\varphi)^{w-1}$, where
$\rho\coloneqq\beta/\alpha_Q\in(0,1)$ for all but finitely many $\varphi$.
Since $H$ is linear in the interpolated values, $W-H$ is the average of the
interpolation errors of the quadratic interpolants of $f_\varphi$ at
$k-1,k,n$.  For each such $\varphi$, $f_\varphi'''=(\ln\rho)^3f_\varphi<0$,
and the error at an integer $w\in[1,n]$ outside the nodes is
$-f_\varphi'''(\xi)\,(w-k+1)(w-k)(n-w)/6>0$ for some $\xi$.  Hence
$W(w)>H(w)$ off $\{k-1,k,n\}$, with equality at the nodes; this is the
interpolation argument of \cite[Lemma~19]{CohnZhao2014Energy}.  Dividing by
$(w-k)(w-k+1)>0$ shows that the quotient in \eqref{eq:eta-def} exceeds
$\eta^*$ for $w\neq n$ and equals $\eta^*$ at $w=n$.  Thus the minimum is
attained uniquely at $w=n$, and $\eta_{n,Q}=\eta^*$.  Strict complete
monotonicity (Section~\ref{sec:qmds}) gives $a>0$ and $W(n)<W(k)$, so
$\eta^*<a/(s+1)<a$.  Hence $\kappa_{n,Q}=\eta^*$, which is
\eqref{eq:kappa-closed}, and $W\geq H$ is \eqref{eq:minorant}.

It remains to prove \eqref{eq:kappa-bound}.
Recall the coefficient-norm representation \eqref{eq:W-norm}, with
$u(z)=1+(Q-1)z$ and $b(z)=1-z$.
Put $R_Q(z)\coloneqq z^2+(Q-2)z+1$.  On $|z|=1$ one has
$|u|^2-|b|^2=Qz^{-1}R_Q(z)$, so taking second differences in
\eqref{eq:W-norm} and applying Parseval gives
\begin{equation}
 \begin{aligned}
 &W(w+2)-2W(w+1)+W(w)\\
 &\quad=Q^2\|u^{n-w-2}b^{w-1}R_Q\|_2^2\\
 &\quad\geq2Q^2.
 \end{aligned}
 \label{eq:curvature}
\end{equation}
The last inequality uses the distinct nonzero constant and leading
coefficients.  Write $C_j$ for the left side of \eqref{eq:curvature} at
$w=j$.  Telescoping from $W(k-1)-W(k)=a$ gives
\[
 W(n)-W(k)+as=\sum_{h=0}^{s-1}(s-h)C_{k-1+h}\geq2Q^2\,\frac{s(s+1)}2,
\]
so $\kappa_{n,Q}=\eta^*\geq Q^2$.
\end{proof}

\section{Stabilizer refinements}
\label{app:stabilizer}
\label{sec:stabilizer}

Stabilizer structure turns the covering defects into integer counts.
We first use scalar orbits for the quantum MDS gap, and then exploit
syndrome collisions for the stronger one-error bound.
Throughout this appendix, \(q\) is prime.
For phase-free labels \(a=(u,v)\) and \(b=(u',v')\) in
\(\mathbb F_q^{2n}\), we use the symplectic form
\[
 \langle a,b\rangle_{\mathrm{sp}}
 \coloneqq v\mathbin{\cdot}u'-u\mathbin{\cdot}v',
\]
and write \(S^{\perp_{\mathrm{sp}}}\) for the subspace orthogonal to
\(S\) under this form.

\subsection{Quantum MDS gap}

The exact remainder also has an arithmetic consequence unavailable for
arbitrary projectors.

\begin{corollary}[Stabilizer gap at quantum MDS parameters]
\label{cor:QMDS-stabilizer-gap}
Assume that \(q\) is prime and that \(P\) is an \([[n,\ell]]_q\) stabilizer
projector with the rank in \eqref{eq:QMDS-parameters}, and write
\(g(P)\coloneqq\cD_Q(P)-\delta^{\mathrm{QMDS}}_{n,q,\Delta}\).  Then either \(P\) is
quantum MDS and \(g(P)=0\), or
\begin{equation}
 g(P)\geq(q-1)\frac{K^2\mu_t}{q^n}.
 \label{eq:QMDS-stabilizer-gap}
\end{equation}
In particular, if no stabilizer quantum MDS code exists at the prescribed
parameters, the corresponding positive lower bound applies to every
stabilizer projector of that rank.
\end{corollary}

\begin{proof}[Proof of Corollary~\ref{cor:QMDS-stabilizer-gap}]
Let $S_0^{\mathrm{lab}}\subseteq\mathbb F_q^{2n}$ and
$N_0^{\mathrm{lab}}\coloneqq(S_0^{\mathrm{lab}})^{\perp_{\mathrm{sp}}}$ be the phase-free label sets
of the stabilizer and its normalizer.  With generator phases chosen to fix the
code, Weyl orthogonality gives
\begin{align*}
 |\Tr(PU_a)|^2
 &=
 \begin{cases}
 K^2,&a\in S_0^{\mathrm{lab}},\\
 0,&a\notin S_0^{\mathrm{lab}},
 \end{cases}\\
 \Tr(PU_aPU_a^\dagger)
 &=
 \begin{cases}
 K,&a\in N_0^{\mathrm{lab}},\\
 0,&a\notin N_0^{\mathrm{lab}}.
 \end{cases}
\end{align*}
Thus \(A_w\) and \(B_w\) count weight-$w$ stabilizer and normalizer labels,
respectively, in the normalization \eqref{eq:A-def} and \eqref{eq:B-def}.
Every nonzero weight layer is a union of
\(\mathbb F_q^*\)-orbits, so each positive coordinate is at least \(q-1\).
If \(P\) is not quantum MDS, Theorem~\ref{thm:QMDS} gives
\(B_s>0\) for some \(1\leq s\leq t\).  The case \(m=t\) of
\eqref{eq:QMDS-two-gaps} has nonnegative terms, and the \(j=t\) summand in
\eqref{eq:QMDS-hhat} gives \(\widehat h_t(s)\geq\mu_t/Q^t\); therefore
\[
 g(P)\geq K\frac{\mu_t}{Q^t}(q-1)
 =(q-1)\frac{K^2\mu_t}{q^n}.
\]
\end{proof}

\subsection{Syndrome collisions}

This subsection uses prime-$q$ stabilizer structure to express the radius-one
covering defect as an integer syndrome-collision count and derive the explicit
nonattainment gap in Corollary~\ref{cor:discrete-gap}.

Let $P$ be an
$[[n,\ell]]_q$ stabilizer-code projector, so $K=q^\ell$.  Put
$h\coloneqq n-\ell$.  Choose commuting
stabilizer generators $S_1,\ldots,S_h$ that fix the code, and let
$g_1,\ldots,g_h\in\mathbb F_q^{2n}$ be their phase-free Weyl labels.  Identify the
syndrome set with
$\mathcal S\coloneqq\mathbb F_q^h$, where
$M\coloneqq|\mathcal S|=q^h=d/K$.  Define the linear syndrome map by
\[
 \sigma(a)\coloneqq
 \bigl(\langle g_1,a\rangle_{\mathrm{sp}},\ldots,
       \langle g_h,a\rangle_{\mathrm{sp}}\bigr).
\]
Put $\omega\coloneqq e^{2\pi i/q}$ and
$S(c)\coloneqq\prod_{i=1}^hS_i^{c_i}$ for $c\in\mathbb F_q^h$.  The mutually
orthogonal rank-$K$ syndrome projectors are
\begin{equation}
 P_s\coloneqq\frac1M\sum_{c\in\mathbb F_q^h}\omega^{-c\cdot s}S(c),
 \qquad s\in\mathcal S,
 \label{eq:syndrome-projectors}
\end{equation}
so that $S_iP_s=\omega^{s_i}P_s$ and
$\sum_{s\in\mathcal S}P_s=I$.  The Weyl commutation rule gives
$S_iU_a=\omega^{\langle g_i,a\rangle_{\mathrm{sp}}}U_aS_i$; hence
\begin{equation}
 U_aPU_a^\dagger=P_{\sigma(a)}.
 \label{eq:syndrome-action}
\end{equation}
Define
\begin{equation}
 m_r(s)\coloneqq|\{a\in\cB_r:\sigma(a)=s\}|.
 \label{eq:multiplicity}
\end{equation}
Then
\begin{equation}
 \begin{aligned}
 T_r(P)&=\sum_s m_r(s)P_s,\\
 \Phi_r(P)&=K\sum_s\left(m_r(s)-\frac{V_r}{M}\right)^2.
 \end{aligned}
 \label{eq:stabilizer-variance}
\end{equation}
Thus \eqref{eq:stabilizer-variance} is the classical syndrome-multiplicity
variance of an additive code.

The linearity, scalar orbits, and symplectic self-orthogonality used below are
standard features of the geometric description of stabilizer codes
\cite{BierbrauerEtAl2008Geometry,BallCentellesHuber2023Geometry}.  The moment
identities in Lemma~\ref{lem:syndrome-moment} are their coordinate form.
Corollary~\ref{cor:discrete-gap} converts these standard constraints into an
explicit lower bound on $\sum_s(m_1(s)-1)^2$.

Primality is used precisely in this stabilizer refinement: it makes the chosen
Weyl labels and syndrome geometry $\mathbb F_q$-linear, so every nonzero layer
splits into full $\mathbb F_q^*$-orbits.  The arbitrary-projector results in
the preceding sections require only $q\geq2$.  For non-prime $q$ the labels
$\mathbb Z_q^2$ do not form an $\mathbb F_q$-vector space; for $q=p^m$,
stabilizer codes are naturally described in an $\mathbb F_q$-Pauli basis and
may be $\mathbb F_q$-linear or merely $\mathbb F_p$-additive.  Extending the
orbit-based constants to any non-prime $q$, including $\mathbb F_q$-linear
codes, remains open.

For binary stabilizer codes that actually correct one error, the stronger
Hamming-like bound of \cite[Corollary~2]{NemecTansuwannont2023HammingLike}
already forces nondegeneracy at the present equality parameters.
Corollary~\ref{cor:discrete-gap}
has a different scope: it applies to every stabilizer projector of the fixed
rank, without any distance or correction hypothesis.  In particular, the
positive-defect witness below is not a one-error-correcting code.

At the one-error parameters \eqref{eq:admissible}, $M=V$, and
\begin{equation}
 \frac{\Phi_1(P)}K=\sum_s(m_1(s)-1)^2
 =\sum_s m_1(s)(m_1(s)-1).
 \label{eq:integer-defect}
\end{equation}
The last expression counts ordered collisions.
The syndrome multiplicities are invariant under nonzero scalar multiplication.
In the binary case, stabilizer commutativity also imposes
a quadratic moment.

\begin{lemma}[Syndrome symmetries and moments]
\label{lem:syndrome-moment}
Under \eqref{eq:admissible}, put $d_s\coloneqq m_1(s)-1$.  Every stabilizer syndrome
profile satisfies
\begin{equation}
 \sum_{s\in\mathcal S}d_s=0,
 \qquad
 \sum_{s\in\mathcal S}d_s s=0
 \quad\text{in }\mathbb F_q^h.
 \label{eq:syndrome-moment}
\end{equation}
Moreover,
\begin{equation}
 d_{\alpha s}=d_s
 \quad(\alpha\in\mathbb F_q^*,\ s\ne0),
 \qquad q-1\mid d_0.
 \label{eq:syndrome-scalar}
\end{equation}
For every prime $q$, one also has
\begin{equation}
 \sum_{s\in\mathbb F_q^h}d_s ss^{\mathsf T}=0
 \quad\text{in }\mathbb F_q^{h\times h}.
 \label{eq:syndrome-quadratic}
\end{equation}
\end{lemma}

\begin{proof}[Proof of Lemma~\ref{lem:syndrome-moment}]
Recall
$\mathcal S=\mathbb F_q^h$, $h=n-\ell$, and $d_s=m_1(s)-1$.
At Hamming equality $q^h=M=V=1+n(q^2-1)$.  Since $n\geq3$, this forces
$h\geq2$.  The radius-one Pauli labels are the union of the $n$ coordinate
copies of $\mathbb F_q^2$, with their common zero counted once.  The sum of
the vectors in each copy is zero, and linearity of $\sigma$ therefore gives
\[
 \sum_{s\in\mathcal S}m_1(s)s
 =\sum_{a\in\cB_1}\sigma(a)=0.
\]
Every coordinate value occurs $q^{h-1}$ times in $\mathbb F_q^h$, so
$\sum_{s\in\mathcal S}s=0$ as well.  Subtracting these identities proves
\eqref{eq:syndrome-moment}.

Multiplication by any $\alpha\in\mathbb F_q^*$ permutes the nonzero
radius-one labels, while
$\sigma(\alpha a)=\alpha\sigma(a)$.  This proves the first assertion in
\eqref{eq:syndrome-scalar}.  The nonidentity labels of syndrome zero split
into scalar orbits of size $q-1$, so $m_1(0)-1=d_0$ is divisible by $q-1$.

It remains to prove \eqref{eq:syndrome-quadratic}.  Let
$x_i,z_i\in\mathbb F_q^h$ be the syndromes of the two standard basis labels
in coordinate $i$.  When $q=2$, the three nonidentity labels in that
coordinate have syndromes $x_i,z_i,x_i+z_i$, whence
\[
 \sum_{a\in\cB_1}\sigma(a)\sigma(a)^{\mathsf T}
 =\sum_{i=1}^n(x_i z_i^{\mathsf T}+z_i x_i^{\mathsf T})=0.
\]
The final equality is exactly the pairwise commutativity condition for the
chosen stabilizer generators.  If $q>2$, the contribution of one coordinate
is instead
\[
 \sum_{u,v\in\mathbb F_q}(ux_i+vz_i)(ux_i+vz_i)^{\mathsf T}=0,
\]
because $\sum_u u=0$ and the diagonal sums carry a factor $q=0$ in
$\mathbb F_q$.  Thus the radius-one second moment vanishes for every prime
$q$.  The uniform second moment vanishes as well: for $q=2$, admissibility
with $n\geq3$ forces $h\geq4$, and for $q>2$ this follows directly by
summing each coordinate over $\mathbb F_q^h$.  Subtracting the uniform
moment proves \eqref{eq:syndrome-quadratic}.
\end{proof}

These constraints force a nonconstant syndrome profile to have a definite
squared distance from the uniform profile.  Combining that arithmetic
separation with Theorem~\ref{thm:main} gives a positive discrepancy gap.

\begin{corollary}[Stabilizer nonattainment gap]
\label{cor:discrete-gap}
Under \eqref{eq:admissible}, a stabilizer projector satisfies either
$\Phi_1(P)=0$ or
\begin{equation}
 \Phi_1(P)\geq\gamma_qK,\qquad
 \cD_Q(P)-\delta_{n,Q}\geq\gamma_qKc_{n,Q}
 \geq\gamma_qK,
 \label{eq:stabilizer-gap}
\end{equation}
where
\begin{equation}
 \gamma_q\coloneqq\max\{8,2(q-1)\}.
 \label{eq:gamma-q}
\end{equation}
For $q=2$, the defect bound $\Phi_1(P)\geq8K$ is sharp already at the
five-qubit Hamming parameters.
\end{corollary}

\begin{proof}[Proof of Corollary~\ref{cor:discrete-gap}]
First suppose that $q$ is odd.  On each projective orbit
$O\subseteq\mathbb F_q^h\setminus\{0\}$, write $d_s\eqqcolon e_O$, and write
$d_0\eqqcolon(q-1)f$.  Equations \eqref{eq:syndrome-moment} and
\eqref{eq:syndrome-scalar} give
\[
 \begin{aligned}
 f+\sum_Oe_O&=0,\\
 \sum_s d_s^2
 &=(q-1)\left((q-1)f^2+\sum_Oe_O^2\right).
 \end{aligned}
\]
If $f=0$ and the profile is nonconstant, the nonzero integer vector
$(e_O)_O$ has sum zero, so $\sum_Oe_O^2\geq2$.  If $f\ne0$, then
$\sum_Oe_O^2\geq\sum_O|e_O|\geq|f|$, and the displayed expression is at
least $q(q-1)$.  In either case it is at least $2(q-1)$.

For $q\geq5$, this already gives \eqref{eq:gamma-q}.  When $q=3$, the only
positive values below $8$ left by the displayed formula are $4$ and $6$.
Indeed, if $f=0$ the value $2\sum_Oe_O^2<8$ forces $\sum_Oe_O^2=2$ (the
value $3$ is impossible for entries $\pm1$ summing to $0$); if $f\neq0$ it
forces $f=\pm1$ and a single orbit with $e_O=-f$.  Value $4$ would require $f=0$ and two distinct projective orbits with
coefficients $1$ and $-1$.  Equation \eqref{eq:syndrome-quadratic} would
then equate their rank-one matrices, forcing the two orbits to coincide.
Value $6$ would require $f=\pm1$ and a single nonzero orbit coefficient
$-f$; the same quadratic moment would force a nonzero rank-one matrix to
vanish.  Both are impossible, so the positive defect is at least $8$ also
for $q=3$.

Now suppose that $q=2$.  Reduce $(d_s)_s$ modulo two and let
$\Omega\coloneqq\{s:d_s\text{ is odd}\}$.  Equations
\eqref{eq:syndrome-moment} and \eqref{eq:syndrome-quadratic} say that the
 indicator of $\Omega$ is orthogonal to the evaluations on $\mathbb F_2^h$
 of every polynomial of degree at most two.  Hence it belongs to the binary
 Reed--Muller code
 $\operatorname{RM}(2,h)^\perp=\operatorname{RM}(h-3,h)$, whose minimum
distance is $8$
\cite{DelsarteGoethalsMacWilliams1970RM}
\cite[Ch.~13, Sec.~3, Thms.~3--4]{MacWilliamsSloane1977}.  If $\Omega$ is
nonempty, then $\sum_s d_s^2\geq|\Omega|\geq8$.  If $\Omega$ is empty,
every $d_s$ is even.  Since $d_s=m_1(s)-1\geq-1$, each $d_s$ is then
nonnegative, and $\sum_s d_s=0$ forces $d_s=0$ for every $s$.  Thus the
nonconstant case necessarily has $\Omega\ne\varnothing$.

The support bound $|\Omega|\geq8$ corresponds, after adjoining the zero
vector, to the smallest case $y=7$ (a plane) in
\cite[Proposition~5.5]{BierbrauerEtAl2008Geometry}.  Thus
the Reed--Muller/projective support structure used here is known; the new
content is its conversion into an explicit lower bound on $\Phi_1(P)/K$
and then into the discrepancy gap.

Equation \eqref{eq:integer-defect} and Theorem~\ref{thm:main} now give
\eqref{eq:stabilizer-gap}.  To see binary sharpness at the defect level,
take the length-five stabilizer generated by
\[
 \mathtt{XIXZY},\qquad \mathtt{IZZZZ},\qquad \mathtt{IXXII},\qquad \mathtt{IIIXX}.
\]
With syndrome coordinates ordered by these generators, its radius-one
profile is
\[
 m_1(s_1,s_2,s_3,s_4)=
 \begin{cases}
  2,&(s_3,s_4)=(0,0),\\
  0,&(s_3,s_4)=(1,1),\\
  1,&\text{otherwise}.
 \end{cases}
\]
Thus it has multiplicity two on four syndromes, one on eight syndromes, and
zero on four syndromes, and
$\Phi_1(P)/K=4+4=8$.
\end{proof}

At the common prime-$q$ parameters $\Delta=3$ and $n=Q+1$, the two
stabilizer gaps have the same benchmark and zero set.  Here
$c_{n,Q}=\mu_2/(2Q^2)$, so Corollary~\ref{cor:QMDS-stabilizer-gap} gives
$g(P)\geq2(q-1)Kc_{n,Q}$, whereas
Corollary~\ref{cor:discrete-gap} gives
$g(P)\geq\max\{8,2(q-1)\}Kc_{n,Q}$, which is strictly stronger only for
$q\in\{2,3\}$.
Corollary~\ref{cor:QMDS-stabilizer-gap} detects a first nonzero
low-weight normalizer layer through scalar orbits;
Corollary~\ref{cor:discrete-gap} uses the full
syndrome-collision profile and its linear and quadratic moment constraints.

\section{Recovery extensions and noise examples}
\label{app:noise}

The main text treats noise with Kraus operators in the span of a specified
ball.  Here we record outside-ball extensions and a comparison for coherent
errors, followed by a concrete point on the five-qubit family.  Each extension uses the
same polar recovery as its corresponding main-text result.

\subsection{Higher-radius extensions}

Use the quantum MDS parameters, the quantities \(g(P)\) and \(C_e\) of
Lemma~\ref{lem:QMDS-polar}, and the recovery \(\cR_e\) defined before
Corollary~\ref{cor:QMDS-span-infidelity}.

\begin{remark}[Coherent errors in the selected span]
\label{rem:QMDS-coherent-span}
For every channel whose Kraus operators lie in
$\operatorname{span}\{U_a:a\in\cB_e\}$,
Corollary~\ref{cor:QMDS-span-infidelity} bounds the diamond error of $\cR_e$
by $2\min\{1,C_eg(P)\}$.  For comparison, the operator-norm estimate
\eqref{eq:QMDS-delta} can be combined with the
approximate-isometry decoder theorem of Ma, Tan, and Wright
\cite[Theorem~1.10 and Lemma~2.4]{MaTanWright2025Haar}.  Their hypothesis is
that the singular values of the error-image map lie in $[1-\zeta,1+\zeta]$
for some $0\leq\zeta<1$ (their $\delta$), and their conclusion bounds one half of the diamond
error by $\zeta$ \cite[Definitions~1.2, 1.7 and 1.9]{MaTanWright2025Haar};
compare \cite[Remark~1]{ElimelechAlbertBarg2026UniversalRecovery}.  If
$C_eg(P)<1$, then \eqref{eq:QMDS-delta} gives this hypothesis with
$\zeta=1-\sqrt{1-\sqrt{C_eg(P)}}$.  Their result therefore applies to the
polar recovery
$\cR_e$ of Lemma~\ref{lem:polar-decoder} with $\cE=\cB_e$, for every
completely positive trace-preserving channel whose Kraus operators lie in
$\operatorname{span}\{U_a:a\in\cB_e\}$, with unhalved diamond error at most
\[
 2\left(1-\sqrt{1-\sqrt{C_eg(P)}}\right)
 \leq2\sqrt{C_eg(P)}.
\]
This bound is of order $g(P)^{1/2}$ and needs $C_eg(P)<1$; the entanglement
infidelity of these channels is at most $C_e\,g(P)/K$ by
\cite[Theorem~7]{ElimelechAlbertBarg2026UniversalRecovery}.
\end{remark}

\begin{corollary}[Full Pauli noise with an outside-ball tail]
\label{cor:QMDS-noise-tail}
Under the hypotheses of Corollary~\ref{cor:QMDS-recovery}, let
\begin{equation}
 \cN_p^{\mathrm{full}}(\sigma)\coloneqq\sum_{a\in(\mathbb Z_q^2)^n}
 p_aU_a\sigma U_a^\dagger
 \label{eq:full-Pauli-channel}
\end{equation}
be an arbitrary stochastic Pauli channel.  Define
\begin{equation}
 s_e\coloneqq\sum_{a\in\cB_e}p_a,
 \qquad \tau_e\coloneqq1-s_e,
 \qquad p_{e,\ast}\coloneqq\max_{a\in\cB_e}p_a.
 \label{eq:noise-tail-parameters}
\end{equation}
Thus $\tau_e$ is the probability mass outside the radius-$e$ ball and
$p_{e,\ast}$ is the largest unconditional probability of a Pauli label in the
ball, the identity included.  The same polar recovery
$\cR_e$ satisfies
\begin{gather}
 \begin{aligned}
 &\|\cR_e\circ\cN_p^{\mathrm{full}}\circ\operatorname{Enc}
       -\id_K\|_\diamond
 \leq 2\tau_e+2\min\!\left\{
 s_e,p_{e,\ast}C_eg(P)\right\}
 \end{aligned}
 \label{eq:noise-tail-diamond}\\
 \begin{aligned}
 &1-\Fent(I_K/K,\cR_e\circ\cN_p^{\mathrm{full}}
       \circ\operatorname{Enc})
 \leq \tau_e+\min\!\left\{
 s_e,\frac{p_{e,\ast}C_eg(P)}K\right\}.
 \end{aligned}
 \label{eq:noise-tail-fidelity}
\end{gather}
These bounds hold for every value of $g(P)$.  In particular, if $g(P)=0$, the
right-hand sides reduce to $2\tau_e$ and $\tau_e$, respectively.
\end{corollary}

\begin{proof}
If $s_e=0$, the claims are the standard upper bounds $2$ and $1$.  Suppose
$s_e>0$ and condition the error distribution on $\cB_e$.  Its maximum
probability is $p_{e,\ast}/s_e$, so
Corollary~\ref{cor:QMDS-recovery} bounds the recovered conditional channel by
\[
 2\min\!\left\{1,
 \frac{p_{e,\ast}}{s_e}C_eg(P)\right\}
\]
in diamond norm and bounds its entanglement infidelity by
$\min\{1,p_{e,\ast}C_eg(P)/(s_eK)\}$.  The recovered full channel is the convex
combination of this conditional channel with weight $s_e$ and a recovered
outside-ball channel with weight $\tau_e$.  Convexity of the unhalved diamond
norm contributes at most $2\tau_e$ from the latter channel, proving
\eqref{eq:noise-tail-diamond}.  Entanglement fidelity at a fixed input is
linear in the channel and lies in $[0,1]$, which proves
\eqref{eq:noise-tail-fidelity}.
\end{proof}

Because $p_{e,\ast}$ is at least the identity probability, the discrepancy
term does not vanish with the noise; the recovery itself is exact on
noise-free input when $PU_aP=0$ for $1\leq\wt(a)\leq e$, but not in general
(Section~\ref{sec:qmds-recovery}).

\begin{remark}[Independent depolarizing noise]
\label{rem:depolarizing-tail}
For $0\leq p_{\mathrm{dep}}\leq1$, consider the single-qudit Pauli channel
\[
 \mathcal P_{p_{\mathrm{dep}}}(\rho)
 \coloneqq(1-p_{\mathrm{dep}})\rho+
 \frac{p_{\mathrm{dep}}}{q^2-1}
 \sum_{b\in\mathbb Z_q^2\setminus\{0\}}U_b\rho U_b^\dagger,
\]
where $U_b$ is a single-qudit phase-free Weyl operator, and take its
independent $n$-fold product.  Then
\begin{align}
 \tau_e(p_{\mathrm{dep}})
 &=\sum_{w=e+1}^n\binom nw p_{\mathrm{dep}}^w
 (1-p_{\mathrm{dep}})^{n-w},
 \label{eq:depolarizing-tail}\\
 p_{e,\ast}(p_{\mathrm{dep}})
 &=\max_{0\leq w\leq e}(1-p_{\mathrm{dep}})^{n-w}
 \left(\frac{p_{\mathrm{dep}}}{q^2-1}\right)^w.
 \label{eq:depolarizing-pmax}
\end{align}
For $0\leq p_{\mathrm{dep}}\leq1-q^{-2}$, the maximum in
\eqref{eq:depolarizing-pmax} is $(1-p_{\mathrm{dep}})^n$.  At the exact
five-qubit code, $e=1$ and $g(P)=0$.  Write
\[
 \Lambda_{\mathrm{dep}}
 \coloneqq\cR_1\circ\mathcal P_{p_{\mathrm{dep}}}^{\otimes5}
 \circ\operatorname{Enc}.
\]
The guarantees become
\[
 \begin{aligned}
 \|\Lambda_{\mathrm{dep}}-\id_2\|_\diamond
 &\leq2\tau_1(p_{\mathrm{dep}}),\\
 1-\Fent(I_2/2,\Lambda_{\mathrm{dep}})
  &\leq\tau_1(p_{\mathrm{dep}}),
 \end{aligned}
\]
where
\[
 \tau_1(p_{\mathrm{dep}})
 =10p_{\mathrm{dep}}^2-20p_{\mathrm{dep}}^3
 +15p_{\mathrm{dep}}^4-4p_{\mathrm{dep}}^5.
\]
The pure tail estimate at $g(P)=0$ is standard for an exactly correctable
error ball \cite{KnillLaflamme1997Theory}; enumerator-based finite-length
fidelity bounds for depolarizing noise are developed in
\cite{Ashikhmin2014Fidelity}.  Corollary~\ref{cor:QMDS-noise-tail} adds the
explicit discrepancy penalty for a general projector.  The estimate is a
finite-block recovery guarantee, not a fault-tolerance threshold.
\end{remark}

\subsection{One-error extensions}

We now use the Hamming parameters \eqref{eq:admissible} and the recovery
of Corollary~\ref{cor:recovery}.

The same recovery also applies to a full Pauli distribution.  Separating
the radius-one ball from its tail gives an in-ball term controlled by
discrepancy and an outside-ball term controlled by its probability.

\begin{corollary}[Full Pauli noise with a radius-one tail]
\label{cor:one-error-noise-tail}
Under \eqref{eq:admissible}, let
\[
 \cN_p^{\mathrm{full}}(\sigma)
 \coloneqq\sum_{a\in(\mathbb Z_q^2)^n}p_aU_a\sigma U_a^\dagger
\]
be a stochastic Pauli channel, and define
\[
 s_1\coloneqq\sum_{a\in\cB_1}p_a,\qquad
 \tau_1\coloneqq1-s_1,\qquad
 p_{1,\ast}\coloneqq\max_{a\in\cB_1}p_a.
\]
For $g\coloneqq\cD_Q(P)-\delta_{n,Q}$ and $c\coloneqq c_{n,Q}$, the same recovery $\cR$
as in Corollary~\ref{cor:recovery} satisfies
\begin{gather}
 \begin{aligned}
 &\|\cR\circ\cN_p^{\mathrm{full}}\circ\operatorname{Enc}
       -\id_K\|_\diamond
 \leq2\tau_1+
 2\min\left\{s_1,\frac{p_{1,\ast}g}{c}\right\}
 \end{aligned}
 \label{eq:one-error-tail-diamond}\\
 \begin{aligned}
 &1-\Fent(I_K/K,\cR\circ\cN_p^{\mathrm{full}}
       \circ\operatorname{Enc})
 \leq\tau_1+
 \min\left\{s_1,\frac{p_{1,\ast}g}{cK}\right\}.
 \end{aligned}
 \label{eq:one-error-tail-fidelity}
\end{gather}
The estimates hold for every $g\geq0$.  At the five-qubit parameters,
$c=3$ and $K=2$, so the discrepancy term in
\eqref{eq:one-error-tail-fidelity} is $p_{1,\ast}g/6$.
\end{corollary}

\begin{proof}
As in the proof of Corollary~\ref{cor:QMDS-noise-tail}, for $s_1>0$,
condition the distribution on $\cB_1$ and apply
Corollary~\ref{cor:recovery} with maximum conditional probability
$p_{1,\ast}/s_1$.  The full recovered channel is a convex combination of this
conditional channel with weight $s_1$ and an outside-ball channel with weight
$\tau_1$.  Convexity of the unhalved diamond norm and linearity of
entanglement fidelity give the displayed bounds.  The case $s_1=0$ reduces to
the universal bounds $2$ and $1$.
\end{proof}

At the common parameters $\Delta=3$, $n=Q+1$ and $e=1$, choosing the same
unitary completion makes this recovery the recovery $\cR_1$ of
Corollary~\ref{cor:QMDS-recovery}.  By
Remark~\ref{rem:QMDS-one-error-overlap}, $C_1^{\mathrm{Gram}}=1/c_{n,Q}$, so
Corollaries~\ref{cor:one-error-noise-tail} and~\ref{cor:QMDS-noise-tail} then
give the same bounds.  The quantum MDS corollaries retain their
broader higher-distance and higher-radius scope.

\begin{remark}[Coherent radius-one errors]
\label{rem:one-error-coherent-span}
The $p_\ast$ refinement above concerns Pauli mixtures; for every channel in
$\operatorname{span}\{U_a:a\in\cB_1\}$, Corollary~\ref{cor:recovery} gives
diamond error at most $2\min\{1,g/c\}$.  For comparison, put
$\gamma\coloneqq\sqrt{g/c}$; by \eqref{eq:eigen-bound} and
Theorem~\ref{thm:main}, $\|G-I_{VK}\|_{\mathrm{op}}\leq\gamma$.  When
$\gamma<1$, the approximate-isometry theorem of
Ma, Tan, and Wright \cite[Theorem~1.10]{MaTanWright2025Haar}, applied to the
polar recovery $\cR$ of Corollary~\ref{cor:recovery} with
$\zeta=1-\sqrt{1-\gamma}$ as in Remark~\ref{rem:QMDS-coherent-span}, gives,
for every channel $\cN$ whose Kraus operators lie
in $\operatorname{span}\{U_a:a\in\cB_1\}$, the bound
\[
 \|\cR\circ\cN\circ\operatorname{Enc}-\id_K\|_\diamond
 \leq2\bigl(1-\sqrt{1-\gamma}\bigr)\leq2\gamma.
\]
This bound is of order $g^{1/2}$; the entanglement infidelity of these
channels is at most $g/(cK)$ by
\cite[Theorem~7]{ElimelechAlbertBarg2026UniversalRecovery}.
Errors outside that span and general local noise require a larger
error-set analysis.  The recovery is an existence construction at the operator
level; efficient circuit synthesis, gate counts, and hardware resources lie
outside its scope.
\end{remark}

\subsection{A concrete five-qubit noise example}

Choosing a rotation with rational sine and cosine gives a concrete
comparison between the exact logical error and the general discrepancy
guarantee.  The same example also illustrates the addition of an
outside-ball noise contribution.

\begin{corollary}[A five-qubit example beyond exact correction]
\label{cor:five-nonexact-benchmark}
Choose $\theta_\circ\in(0,\pi/4)$ by
$\sin(2\theta_\circ)=3/5$ and put $P_\circ\coloneqq P_{\theta_\circ}$.  Abbreviate
$\cR_\circ\coloneqq\cR_{\theta_\circ}$ and
$\operatorname{Enc}_\circ\coloneqq\operatorname{Enc}_{\theta_\circ}$.  Then
\[
 g(P_\circ)=\frac{432}{25}>0,\qquad
 \xi_{\theta_\circ}=\frac1{10}.
\]
Moreover, $A_2(P_\circ)=0<B_2(P_\circ)=36/25$, so the radius-one Pauli
error set is not exactly correctable.  For uniform radius-one Pauli noise,
\begin{equation}
 \Lambda_{\theta_\circ,\mathrm{unif}}
 =\frac{19}{20}\id_2{}
  +\frac1{40}\operatorname{Ad}_{X_L}{}
  +\frac1{40}\operatorname{Ad}_{Y_L},
 \label{eq:five-benchmark-channel}
\end{equation}
and therefore
\[
 \|\Lambda_{\theta_\circ,\mathrm{unif}}-\id_2\|_\diamond=\frac1{10},
 \qquad
 \Fent(I_2/2,\Lambda_{\theta_\circ,\mathrm{unif}})=\frac{19}{20}.
\]
The general discrepancy estimates
\eqref{eq:diamond-bound}--\eqref{eq:weighted-Fe-bound} yield the upper
bounds $18/25$ and $9/50$, respectively, for the diamond error and entanglement
infidelity at this same point.

For independent depolarizing noise with $0\leq p_{\mathrm{dep}}\leq1$, use the
single-qubit channel
\[
 \mathcal P_{p_{\mathrm{dep}}}
 =(1-p_{\mathrm{dep}})\id_2{}
 +\frac{p_{\mathrm{dep}}}{3}
  \bigl(\operatorname{Ad}_X{}+\operatorname{Ad}_Y{}+
        \operatorname{Ad}_Z\bigr)
\]
and define
\[
 \tau_1\coloneqq1-(1-p_{\mathrm{dep}})^4(1+4p_{\mathrm{dep}}),
 \qquad
 \eta_1\coloneqq\frac{8p_{\mathrm{dep}}}{3}(1-p_{\mathrm{dep}})^4.
\]
Here $\tau_1$ is the probability of at least two physical errors, while
$\eta_1$ is the total probability of the eight active weight-one errors in
$\mathcal S_X\cup\mathcal S_Y$.
The same recovery satisfies
\begin{align}
 \|\cR_\circ\circ\mathcal P_{p_{\mathrm{dep}}}^{\otimes5}
       \circ\operatorname{Enc}_\circ-\id_2\|_\diamond
 &\leq2\tau_1+\frac{\eta_1}{5},
 \notag\\
 1-\Fent(I_2/2,\cR_\circ\circ\mathcal P_{p_{\mathrm{dep}}}^{\otimes5}
       \circ\operatorname{Enc}_\circ)
 &\leq\tau_1+\frac{\eta_1}{10}.
 \label{eq:five-depolarizing-benchmark}
\end{align}
For example, at $p_{\mathrm{dep}}=0.1$ the right-hand sides are, respectively,
$0.197912$ and $0.098956$.
\end{corollary}

\begin{proof}
The values of $g$, $\xi$, $A_2$, and $B_2$ follow from
\eqref{eq:five-sharp}, \eqref{eq:five-xi}, and
Proposition~\ref{prop:five-exact-recovery}.
Equation~\eqref{eq:five-benchmark-channel} is the uniform specialization of
Proposition~\ref{prop:five-exact-recovery}.

Under independent depolarizing noise, the probability of at least two
errors is $\tau_1$.  Each of the eight active weight-one errors in
$\mathcal S_X\cup\mathcal S_Y$ has probability
$p_{\mathrm{dep}}(1-p_{\mathrm{dep}})^4/3$, giving total active mass $\eta_1$.
The in-ball contribution follows exactly from
\eqref{eq:five-exact-diamond}--\eqref{eq:five-exact-fidelity}; the
outside-ball channel contributes at most $2\tau_1$ and $\tau_1$,
respectively.  Substitution of $p_{\mathrm{dep}}=0.1$ gives the displayed numbers.
\end{proof}

\section{Balanced Gottesman codes}
\label{sec:gottesman}

Gottesman's family illustrates a different packing regime: its radius-one
error images are mutually orthogonal but leave part of the physical Hilbert
space uncovered.  We subtract the balanced-packing value from the covering
defect and ask how the all-radii discrepancy controls the remaining overlap.
Put
\begin{equation}
 \begin{aligned}
 n&\coloneqq2^j,\qquad j\geq3,\\
 K&\coloneqq\frac{2^n}{4n},\qquad M\coloneqq\frac{2^n}{K}=4n,\\
 a_G&\coloneqq\frac{3n}{4},
 \end{aligned}
 \label{eq:Gottesman-parameters}
\end{equation}
and abbreviate $W(w)\coloneqq W_{n,4}(w)$.

The weight enumerator of this family is known from the stabilizer/additive
code construction; see Calderbank--Rains--Shor--Sloane
\cite[Theorem~10]{CalderbankEtAl1998GF4} and Gottesman
\cite[Section~III and Appendices~B--C]{Gottesman1996Class}.  The next lemma
records that enumerator in the normalization used here and recalls the
balanced syndrome count.  We then derive its exact discrepancy remainder and
stability inequality.

\begin{lemma}[Gottesman enumerator and balanced syndromes]
\label{lem:Gottesman-enumerator}
The stabilizer projector $P_G$ of Gottesman's
$[[n,n-j-2,3]]_2$ code has
\begin{equation}
 \begin{aligned}
 A_{P_G}(z)&=1+(4n-4)z^{a_G}+3z^n,\\
 B_1(P_G)&=B_2(P_G)=0.
 \end{aligned}
 \label{eq:Gottesman-enumerator}
\end{equation}
Its radius-one syndrome multiplicities consist of $3n+1$ ones and $n-1$
zeros.  Consequently
\begin{equation}
 \Phi_1(P_G)=K\frac{(3n+1)(n-1)}{4n}.
 \label{eq:Gottesman-Phi-one}
\end{equation}
\end{lemma}

\begin{proof}[Proof of Lemma~\ref{lem:Gottesman-enumerator}]
Here $n=2^j$ and $M=4n$, as in \eqref{eq:Gottesman-parameters}.
The enumerator is also the specialization of the known family in
\cite[Theorem~10]{CalderbankEtAl1998GF4}.  For a direct normalization and
syndrome check, Gottesman's construction has $j+2$ stabilizer generators
\cite[Section~III and Appendices~B--C]{Gottesman1996Class}.  In its standard
check-matrix form, two generators are the global $X$ and $Z$ rows.  Label
each coordinate by the last $j$ syndrome bits $x\in\mathbb F_2^j$ of its $X$ error,
and write $\pi(x)$ for those of its $Z$ error.  Gottesman's construction makes
both $\pi$ and $x\mapsto x+\pi(x)$ permutations of $\mathbb F_2^j$.  If
$u\in\mathbb F_2^j\setminus\{0\}$ selects a nonzero combination of the last
$j$ stabilizer rows, then each of
\[
 u\mathbin{\cdot}x,\qquad u\mathbin{\cdot}\pi(x),\qquad
 u\mathbin{\cdot}(x+\pi(x))
\]
is equidistributed on $\mathbb F_2$.  These three equations imply that the pair
$(u\cdot\pi(x),u\cdot x)$ takes each value of $\mathbb F_2^2$ exactly $n/4$
times.  Multiplication by either global row only translates these four
values.  Hence each local Pauli symbol occurs $n/4$ times.
Thus the $4n-4$ stabilizer elements involving a nonglobal row have weight
$3n/4$, while the three nonidentity products of the two global rows have
weight $n$.  This proves the first identity.  The code is nondegenerate and
has distance three, so its normalizer has no elements of weights one or two,
which gives $B_1=B_2=0$.

The identity and the $3n$ nonidentity single-qubit Pauli errors have distinct
syndromes.  They occupy $3n+1$ of the $M=4n$ syndrome classes once each,
leaving $n-1$ holes.  Substitution in
\eqref{eq:stabilizer-variance} gives \eqref{eq:Gottesman-Phi-one}.
\end{proof}

Since the target enumerator has $B_1=B_2=0$, we use a quadratic certificate:
after MacWilliams transformation, its nonconstant terms contribute only
$B_1$ and $B_2$ to the energy.
We interpolate $W$ at the target weights $a_G,n$ and the adjacent lattice
point $a_G+1$; the next lemma verifies the sign of the residual.  Let
\begin{equation}
 h_n(w)\eqqcolon\widehat h_n(0)+\widehat h_n(1)K_1^{(n,4)}(w)+\widehat h_n(2)K_2^{(n,4)}(w)
 \label{eq:Gottesman-h}
\end{equation}
be this quadratic interpolant.  Direct substitution gives
\begin{align}
 \widehat h_n(2)&=\frac{1}{2n(n-4)}
 \Bigl(4(W(n)-W(a_G))\notag\\
 &\qquad-n(W(a_G+1)-W(a_G))\Bigr),
 \label{eq:Gottesman-c2}\\
 \widehat h_n(1)&=3\widehat h_n(2)+\frac{W(a_G)-W(a_G+1)}4,\notag\\
 \widehat h_n(0)&=W(a_G)+\frac{3n}{2}\widehat h_n(2).
 \label{eq:Gottesman-c1c0}
\end{align}
\begin{lemma}[Cohn--Zhao certificate at the Gottesman nodes]
\label{lem:Gottesman-certificate}
For the interpolant \eqref{eq:Gottesman-h},
\begin{equation}
 \begin{aligned}
 r_n(w)&\coloneqq W(w)-h_n(w)\geq0,\\
 r_n(w)=0&\Longleftrightarrow w\in\{a_G,a_G+1,n\}.
 \end{aligned}
 \label{eq:Gottesman-residual}
\end{equation}
Moreover, with $\lambda_i\coloneqq W(i)-W(i+1)$, the nonconstant Krawtchouk coefficients
satisfy
\begin{equation}
 \widehat h_n(2)=\frac{n\lambda_{a_G}-4\sum_{i=a_G}^{n-1}\lambda_i}{2n(n-4)}>0,
 \qquad \widehat h_n(1)=3\widehat h_n(2)+\frac{\lambda_{a_G}}{4}>2\widehat h_n(2).
 \label{eq:Gottesman-coefficient-positive}
\end{equation}
\end{lemma}

\begin{proof}
Since $h_n$ is the quadratic interpolant of $W$ at $a_G,a_G+1,n$, the
argument in the proof of Proposition~\ref{prop:kappa}, with $Q=4$ and these
nodes in place of $k-1,k,n$, gives $r_n(w)>0$ off $\{a_G,a_G+1,n\}$ and
$r_n=0$ at the nodes.  This is the interpolation argument in the proof of
\cite[Lemma~19]{CohnZhao2014Energy}, in strict form.

For the coefficient statement, strict decrease and strict convexity give
$\lambda_{a_G}>\lambda_{a_G+1}>\cdots>\lambda_{n-1}>0$.  Substituting
$W(n)-W(a_G)=-\sum_{i=a_G}^{n-1}\lambda_i$ in \eqref{eq:Gottesman-c2} gives the stated
formula for $\widehat h_n(2)$, which is positive because
$4\sum_{i=a_G}^{n-1}\lambda_i<4(n-a_G)\lambda_{a_G}=n\lambda_{a_G}$.  The formula for $\widehat h_n(1)$ follows from
\eqref{eq:Gottesman-c1c0}.
\end{proof}

The nodes $\{a_G,a_G+1,n\}$ form a pair covering of the target distribution
\eqref{eq:Gottesman-enumerator} in the sense of
\cite[Definition~20]{CohnZhao2014Energy}, none of them below $3n/4$, and that
distribution is a $2$-design; hence \cite[Proposition~22]{CohnZhao2014Energy}
already shows that it is a universally optimal quasicode, and the lower bound
$\cD_4(P)\geq\delta_n^G$ below is that classical result lifted by
Proposition~\ref{prop:macwilliams}.  What we add is the exact remainder
\eqref{eq:Gottesman-exact-gap} and its control of $\Psi_1$.

Define the benchmark and recall the packing defect
\eqref{eq:packing-defect} at these parameters:
\begin{align}
 \delta_n^G
 &\coloneqq\frac{K^2}{2^n}\bigl((4n-4)W(a_G)+3W(n)\bigr),
 \label{eq:Gottesman-benchmark}\\
 \Psi_1(P)
 &=\Phi_1(P)-K\frac{(3n+1)(n-1)}{4n}
 \notag\\
 &=K\bigl(4B_1(P)+2B_2(P)\bigr).
 \label{eq:Gottesman-Psi}
\end{align}
In particular, $\Psi_1(P)\geq0$ for every projector of the rank in
\eqref{eq:Gottesman-parameters}; the subtracted term is the radius-one
covering defect of $P_G$.

\begin{theorem}[Gottesman balanced-packing stability]
\label{thm:Gottesman-stability}
At the parameters \eqref{eq:Gottesman-parameters}, every rank-$K$ projector
$P$ satisfies the exact identity
\begin{equation}
 \begin{aligned}
 \cD_4(P)-\delta_n^G
 &=\frac{K^2}{2^n}\sum_{w=1}^nr_n(w)A_w(P)\\
 &\quad+K\widehat h_n(1)B_1(P)+K\widehat h_n(2)B_2(P),
 \end{aligned}
 \label{eq:Gottesman-exact-gap}
\end{equation}
and therefore
\begin{equation}
 \cD_4(P)-\delta_n^G
 \geq\frac{\widehat h_n(2)}{2}\Psi_1(P)\geq0.
 \label{eq:Gottesman-stability}
\end{equation}
Equality with the benchmark, $\cD_4(P)=\delta_n^G$, and equality in the
first inequality of \eqref{eq:Gottesman-stability} each hold if and only if
\begin{equation}
 A_P(z)=1+(4n-4)z^{3n/4}+3z^n.
 \label{eq:Gottesman-equality-enumerator}
\end{equation}
Thus every Gottesman projector $P_G$ attains the global minimum, and every
minimizer has the displayed enumerator.
\end{theorem}

\begin{proof}
Apply Lemma~\ref{lem:certificate-transfer} with $Q=4$, $m=2$,
$\widehat h=\widehat h_n$, $r=r_n$, and $a$
the target distribution in \eqref{eq:Gottesman-enumerator}.  It has mass
$4n=2^n/K$, its first two dual coordinates vanish because
$B_1(P_G)=B_2(P_G)=0$, and $r_n$ vanishes on its support $\{a_G,n\}$ by
\eqref{eq:Gottesman-residual}.  The resulting benchmark is $\delta_n^G$, which
proves \eqref{eq:Gottesman-exact-gap}.  Equations
\eqref{eq:Gottesman-residual} and
\eqref{eq:Gottesman-coefficient-positive} make every term nonnegative.
Moreover,
\[
 K\widehat h_n(1)B_1+K\widehat h_n(2)B_2
 \geq K(2\widehat h_n(2)B_1+\widehat h_n(2)B_2)=\frac{\widehat h_n(2)}{2}\Psi_1,
\]
which proves \eqref{eq:Gottesman-stability}.

Equality with the benchmark forces equality in the first inequality of
\eqref{eq:Gottesman-stability}.  If the latter holds, then the display above,
in which $\widehat h_n(1)>2\widehat h_n(2)$, and \eqref{eq:Gottesman-residual}
give $B_1=0$ and $A_w$ supported, apart from zero, on $a_G,a_G+1,n$.  Write
the three masses as $x_0,x_1,x_2$.  The mass and first two MacWilliams
equations reduce to
\[
 \begin{aligned}
 x_0+x_1+x_2&=4n-1,\\
 4x_1+nx_2&=3n,\\
 4(n-4)x_1&=-8nB_2.
 \end{aligned}
\]
Since $n\geq8$ and $x_1,B_2\geq0$, they give $x_1=B_2=0$, $x_2=3$, and
$x_0=4n-4$.  Conversely, this enumerator has $B_1=B_2=0$ and is supported on
zeros of $r_n$, so it gives $\cD_4(P)=\delta_n^G$ and $\Psi_1(P)=0$.
\end{proof}

\begin{remark}[The first two lengths]
\label{rem:Gottesman-small}
For $n=8$, one has
\[
 h_8=7560+768K_1^{(8,4)}+72K_2^{(8,4)},
 \qquad \delta_8^G=49446,
\]
and \eqref{eq:Gottesman-exact-gap} is
\begin{align*}
 \cD_4(P)-49446
 &=\frac14\bigl(65795584A_1+1728512A_2\\
 &\qquad{}+150528A_3+20480A_4\\
 &\qquad{}+2560A_5\bigr)\\
 &\qquad{}+6144B_1+576B_2.
\end{align*}
Thus
\[
 \cD_4(P)-49446\geq36\left(\Phi_1(P)-\frac{175}{4}\right).
\]
By Theorem~\ref{thm:Gottesman-stability}, equality in this stability bound
holds exactly when $A_P(z)=1+28z^6+3z^8$.

For $n=16$, the target enumerator is $1+60z^{12}+3z^{16}$, and evaluating
\eqref{eq:Gottesman-c2}, \eqref{eq:Gottesman-c1c0} and
\eqref{eq:Gottesman-benchmark} gives
\[
 \begin{aligned}
 \widehat h_{16}(1)&=25054464,\qquad\widehat h_{16}(2)=\frac{4553008}{3},\\
 \delta_{16}^G&=331392418560.
 \end{aligned}
\]
The compact stability form is
\[
 \cD_4(P)-331392418560
 \geq\frac{2276504}{3}\bigl(\Phi_1(P)-11760\bigr).
\]
The same quaternary two-weight parameters also arise from the difference-matrix
construction in \cite[Section~2, Eq.~(2)]{BoyvalenkovEtAl2021TwoWeight};
the displays above give the Pauli-discrepancy stability constants at these
parameters.
\end{remark}

\section{Marginal-purity calculations}
\label{app:pure-calculations}

We give the coefficient-positivity argument for the pure-state identity,
compare the even-length benchmark with the quantum MDS remainder, and
work out the four- and five-qubit cases.
As in Section~\ref{sec:pure}, \(P_\psi\) has rank one, \(q\geq2\) and
\(n\geq2\); \(\Pi_s\), \(\omega_s\), and \(\gamma_s\) are defined in
\eqref{eq:marginal-purity-sum}--\eqref{eq:gamma-def}.

\begin{proof}[Proof of Theorem~\ref{thm:pure-marginal}]
The Weyl expansion on a fixed subsystem $S$ gives
\[
 \Tr(\rho_S^2)=q^{-|S|}
 \sum_{\operatorname{supp}(a)\subseteq S}|\Tr(P_\psi U_a)|^2.
\]
After summing over all $s$-subsets and counting the subsets that contain a
fixed support of size $w$, this becomes the rank-one form of the identity
\eqref{eq:QMDS-marginal-bridge}, which holds for every projector, written
with its $w=0$ term:
\begin{equation}
 \Pi_s(P_\psi)=q^{-s}\sum_{w=0}^s
 \binom{n-w}{s-w}A_w(P_\psi).
 \label{eq:marginal-enumerator}
\end{equation}
If $X_s\coloneqq q^s\Pi_s$, binomial inversion of \eqref{eq:marginal-enumerator} says
\begin{equation}
 \sum_{w=0}^nf(w)A_w=\sum_{s=0}^n\omega_sX_s.
 \label{eq:binomial-inversion-pure}
\end{equation}
Since $K=1$, Theorem~\ref{thm:spectral} and $f(0)=0$ therefore give
\begin{equation}
 \cD_Q(P_\psi)=q^{-n}\sum_{s=0}^nq^s\omega_s\Pi_s(P_\psi).
 \label{eq:pure-marginal-D}
\end{equation}

It remains to prove the positivity needed to pair the complementary
marginals.  For $s\geq1$, substituting $u=w-s$ in \eqref{eq:omega-def} gives
$\omega_s=\mu_{n-s}$ with $\mu_j$ from \eqref{eq:QMDS-mu}, and $\mu_j>0$ for
$0\leq j\leq n-1$ by the strict complete monotonicity of $W_{n,Q}$ in
Section~\ref{sec:qmds}.  Thus $\omega_s>0$.

For a pure global state, complementary reductions have the same nonzero
spectrum.  Hence $\Pi_s=\Pi_{n-s}$, while $\Pi_0=\Pi_n=1$.  Pairing the terms in
\eqref{eq:pure-marginal-D} therefore gives
\begin{equation}
 \cD_Q(P_\psi)=C_{n,q}+q^{-n}\sum_{s=1}^m\gamma_s\Pi_s(P_\psi).
 \label{eq:pure-paired-D}
\end{equation}
Finally,
\[
 \Tr(\rho_S^2)-q^{-s}
 =\left\|\rho_S-q^{-s}I_S\right\|_{\HS}^2.
\]
Subtracting \eqref{eq:pure-benchmark} from
\eqref{eq:pure-paired-D} proves \eqref{eq:pure-stability}.  Every coefficient
is positive, so equality holds exactly when every reduction of at most $m$
parties is maximally mixed, which is the definition of an
$\operatorname{AME}(n,q)$ state.  The last assertion follows from the same
identity and compactness of the pure-state space.
\end{proof}

\begin{remark}[Rank-one quantum MDS concordance]
\label{rem:pure-QMDS-concordance}
Let $m\geq2$, $n\coloneqq2m$, $K\coloneqq1$, and $\Delta\coloneqq m+1$.  Pure quantum MDS projectors
are then
exactly $\operatorname{AME}(2m,q)$ states, and the two formal benchmarks agree:
\[
 \delta^{\mathrm{QMDS}}_{2m,q,m+1}
 =\delta^{\mathrm{pure}}_{2m,q}.
\]
Indeed, $A_w(P_\psi)=B_w(P_\psi)$, and the two exact remainders agree:
\[
 \begin{aligned}
 &q^{-n}\sum_{w=1}^m\overline R(w)A_w(P_\psi)\\
 &\quad+\sum_{w=1}^m\overline h_wA_w(P_\psi)\\
 &\quad=
 q^{-n}\sum_{s=1}^m\gamma_s
 \sum_{\substack{S\subseteq[n]\\|S|=s}}
 \left\|\rho_S-\frac{I_S}{q^s}\right\|_{\HS}^2.
 \end{aligned}
\]
To compare coefficients directly, use $t=m$, $\omega_s=\mu_{n-s}$
for $s\geq1$, and the definitions
\eqref{eq:QMDS-hhat}, \eqref{eq:QMDS-averages}, and
\eqref{eq:gamma-def}, together with \eqref{eq:QMDS-Newton}.
For $1\leq w\leq m$, averaging the two consecutive residuals and
coefficients, then setting $s=n-j$ in the residual tail, gives
\[
\begin{aligned}
 q^{-n}\overline R(w)+\overline h_w
 &=q^{-n}\sum_{s=w}^{m-1}
   (\mu_{n-s}+q^{n-2s}\mu_s)\binom{n-w}{s-w}
   +q^{-n}\mu_m\binom{n-w}{m-w}\\
 &=q^{-n}\sum_{s=w}^m\gamma_s q^{-s}\binom{n-w}{s-w}.
\end{aligned}
\]
Here the two half-weight middle terms combine because $n=2m$.
Expanding the marginal defects by \eqref{eq:marginal-enumerator} therefore
identifies the two remainders as linear forms in $A_1,\ldots,A_m$.
The exact identities \eqref{eq:QMDS-exact-gap} and
\eqref{eq:pure-stability}, applied to any pure state, then give equality
of the benchmarks; no AME existence is needed.
The quantum MDS coordinates retain nontrivial information about the
error-image Gram matrix and its polar geometry.  For $K=1$, however, the
logical input is one-dimensional, so the associated channel-recovery error
bounds are tautological.  The marginal coordinates expose the
marginal-purity defects directly.
The equality of the formal benchmarks remains valid when the common equality
set is empty.
\end{remark}

The first qubit length without an AME state admits a closed solution of this
weighted problem.  The exact optimum of the balanced two-qubit purity and its
equality conditions were established previously
\cite{BandyopadhyayGourSanders2007Duality,GourWallach2010FourQubits}.  The
corollary below records the corresponding Pauli-ball discrepancy formula and
gives a short shadow-enumerator derivation.

\begin{corollary}[Four-qubit pure-state minimum]
\label{cor:four-pure}
Every rank-one projector $P_\psi$ on four qubits satisfies
\begin{equation}
 \cD_4(P_\psi)=\frac{123}{4}+96A_1(P_\psi)+4A_2(P_\psi)
 \geq\frac{155}{4}.
 \label{eq:four-rank-one}
\end{equation}
Equality holds if and only if
\begin{equation}
 A_1(P_\psi)=0,\qquad A_2(P_\psi)=2.
 \label{eq:four-equality}
\end{equation}
In particular, the formal pure-state benchmark is
$\delta^{\mathrm{pure}}_{4,2}=123/4$, and the unavoidable discrepancy gap is
$8$.
The equality conditions \eqref{eq:four-equality} are equivalently that
every one-qubit reduction is maximally mixed and
\begin{equation}
 \frac1{6}\sum_{\substack{S\subseteq[4]\\|S|=2}}\Tr(\rho_S^2)=\frac13.
 \label{eq:four-average-purity}
\end{equation}
The Higuchi--Sudbery state attains equality.
\end{corollary}

\begin{proof}
For a pure state $A=B$, and the MacWilliams fixed-point equations give
\begin{equation}
 A_3=12+A_1-2A_2,\qquad A_4=3-2A_1+A_2.
 \label{eq:four-elimination}
\end{equation}
Direct Krawtchouk evaluation gives
\[
 \begin{gathered}
 (W_{4,4}(1),W_{4,4}(2),W_{4,4}(3),W_{4,4}(4))\\
 =(1540,116,36,20).
 \end{gathered}
\]
Substitution in Theorem~\ref{thm:spectral} proves the identity in
\eqref{eq:four-rank-one}.

The qubit shadow inequality (equivalently, the state-inversion inequality)
\cite{HuberEtAl2018AME} gives
\begin{equation}
 0\leq\sum_{w=0}^4(-1)^wA_w
 =4(A_2-A_1-2).
 \label{eq:four-shadow}
\end{equation}
Indeed, if
$\widetilde P_\psi\coloneqq\sigma_y^{\otimes4}P_\psi^{\mathsf T}
\sigma_y^{\otimes4}$, then Pauli expansion gives
$16\Tr(P_\psi\widetilde P_\psi)=\sum_w(-1)^wA_w\geq0$.
Thus $A_2\geq A_1+2$, and
\[
 \cD_4(P_\psi)
 \geq\frac{155}{4}+100A_1
 \geq\frac{155}{4}.
\]
Both inequalities are equalities precisely under
\eqref{eq:four-equality}.

Equation~\eqref{eq:marginal-enumerator} gives
\[
 \Pi_1=2+\frac{A_1}{2},\qquad
 \Pi_2=\frac{6+3A_1+A_2}{4}.
\]
Since each marginal-purity defect is nonnegative, $A_1=0$ is equivalent to
all four one-qubit reductions being maximally mixed.  Under this condition,
$A_2=2$ is equivalent to $\Pi_2=2$, which is
\eqref{eq:four-average-purity}.

For $\zeta\coloneqq e^{2\pi i/3}$, the normalized Higuchi--Sudbery vector
\cite{HiguchiSudbery2000Couples} is
\[
 \begin{aligned}
 &\frac1{\sqrt6}\bigl(
 |0011\rangle+|1100\rangle\\
 &\qquad+\zeta(|0101\rangle+|1010\rangle)\\
 &\qquad+\zeta^2(|0110\rangle+|1001\rangle)\bigr).
 \end{aligned}
\]
This state has every one-qubit reduction equal to $I_2/2$ and every two-qubit reduction of
purity $1/3$, so it attains the bound.  Higuchi and Sudbery introduced and
analyzed this state, and Facchi et al. numerically identified the balanced
purity value $1/3$ \cite{HiguchiSudbery2000Couples,FacchiEtAl2008MMES}.  The
exact purity optimum and equality conditions agree with
\cite{BandyopadhyayGourSanders2007Duality,GourWallach2010FourQubits}; the
calculation above converts them into the explicit Pauli-ball discrepancy
identity and value in \eqref{eq:four-rank-one}.
\end{proof}

For five qubits, the rank-one MacWilliams fixed point gives a closed formula
in the two lowest enumerator coefficients.  It yields an explicit instance
of the marginal-purity minimization problem.  Recall that $P_*$ denotes the
rank-two five-qubit perfect-code projector of Section~\ref{sec:five}.

\begin{corollary}[Pure five-qubit minimizers]
\label{cor:five-rank-one}
For every unit vector $|\psi\rangle\in(\mathbb C^2)^{\otimes5}$ and
$P_\psi\coloneqq|\psi\rangle\langle\psi|$,
\begin{equation}
 \cD_4(P_\psi)=\frac{2265}{16}
 +674A_1(P_\psi)+26A_2(P_\psi).
 \label{eq:five-rank-one}
\end{equation}
Consequently,
\begin{equation}
 \min_{\rank P=1}\cD_4(P)=\frac{2265}{16},
 \label{eq:five-rank-one-minimum}
\end{equation}
and equality holds exactly for $2$-uniform states, equivalently for
$\operatorname{AME}(5,2)$ states.  In particular, every rank-one
subprojector of $P_*$ attains the minimum.
\end{corollary}

\begin{proof}
For a rank-one projector, $A_w=B_w$ term by term.  Thus the MacWilliams
transform \eqref{eq:quantum-macwilliams} is a fixed-point relation.
In this five-qubit calculation, write $W(w)\coloneqq W_{5,4}(w)$.  The fixed-point
equations at $j=0,1,2$ are
\begin{align*}
 \sum_{w=0}^{5}A_w&=32,\\
 32A_1&=15+11A_1+7A_2\\
      &\quad+3A_3-A_4-5A_5,\\
 32A_2&=90+42A_1+10A_2\\
      &\quad-6A_3-6A_4+10A_5.
\end{align*}
Solving for the three upper weights gives
\begin{equation}
 \begin{aligned}
 A_3&=10+2A_1-A_2,\\
 A_4&=15-A_2,\\
 A_5&=6-3A_1+A_2.
 \end{aligned}
 \label{eq:five-rank-one-elimination}
\end{equation}
Substitution into \eqref{eq:D-spectral}, using
$(W(1),\ldots,W(5))=(21286,1118,246,110,70)$, yields
\begin{align*}
 \cD_4(P_\psi)
 &=\frac1{32}\sum_{w=1}^5A_wW(w)\\
 &=\frac{2265}{16}+674A_1+26A_2,
\end{align*}
which proves \eqref{eq:five-rank-one}.  Since $A_1,A_2\geq0$, equality in
\eqref{eq:five-rank-one-minimum} is equivalent to $A_1=A_2=0$.

For a pure state, this vanishing says that the expectation of every
nonidentity Pauli operator supported on at most two qubits is zero.  Pauli
expansion of the reduced density operators shows that this is equivalent to
every one- and two-qubit reduction being maximally mixed.  For five qubits,
that is precisely $2$-uniformity, or absolute maximal entanglement.  Finally,
Theorem~\ref{thm:equality} gives
$P_*UP_*=0$ for every nonidentity Pauli $U$ of weight at most two.  Hence
every unit vector in $\operatorname{ran}P_*$ has $A_1=A_2=0$.
\end{proof}

The local-unitary classes of five-qubit absolutely maximally entangled states
have recently been determined in \cite{Tan2025FiveQubitAME}.
Corollary~\ref{cor:five-rank-one} gives the corresponding explicit
two-coefficient formula for the all-radii discrepancy.

\newif\ifdeclarations
\declarationsfalse
\ifdeclarations
\section*{Declarations}\pdfbookmark[1]{Declarations}{declarations}

\paragraph{Funding.}
No funding was received for this work.

\paragraph{Competing interests.}
The author declares no competing interests.

\paragraph{Data availability.}
The finite parameter values are specified in the article. Outputs of the
selected diagnostic checks are archived with the code in
\cite{Zabokritskiy2026PauliDiagnostics}.

\paragraph{Code availability.}
An exact-arithmetic script and its expected output for selected scalar
constants and finite parameter instances are archived in
\cite{Zabokritskiy2026PauliDiagnostics}. These checks are supplementary
diagnostics; the mathematical proofs do not rely on their execution.
\fi
\pdfbookmark[1]{References}{references}
\begingroup
\interlinepenalty=10000
\bibliography{references}
\endgroup

\end{document}